\documentclass[11pt]{article}
\usepackage{ifthen}
\usepackage{mdwlist}
\usepackage{amsmath,amssymb,amsfonts,amsthm}
\usepackage{bm}
\usepackage{hyperref}
\usepackage{enumerate}
\usepackage{graphicx}
\usepackage{xspace}
\usepackage{verbatim}
\usepackage{algorithm}
\usepackage{algpseudocode}
\usepackage[margin=1in]{geometry}
\usepackage{color}
\usepackage{thm-restate}
\usepackage{latexsym}
\usepackage{epsfig}

\def\a{\alpha}
\def\b{\beta}

\def\d{\delta}

\def\e{\eta}
\def\eps{\ve}
\renewcommand{\epsilon}{\ve}
\def\ve{\varepsilon}

\def\g{\gamma}

\def\m{\mu}

\def\p{\pi}
\def\P{\Pi}

\def\S{\Sigma}

\def\to{\rightarrow}

\def\R{\mathbb{R}}

\def\co{{\cal O}}

\def\cc{{\cal C}}

\def\cF{\mathcal{F}}
\def\cP{\mathcal{P}}
\def\cX{\mathcal{X}}

\def\tr{\mathrm{tr}}
\def\mutilde{\widetilde{\mu}}

\def\yhat{\widehat{y}}

\newcommand{\Shat}{\widehat{\Sigma}}
\newcommand{\Stilde}{\widetilde{\Sigma}}
\newcommand{\muhat}{\widehat{\mu}}

\newcommand{\Var}{\mathop{\mbox{\bf Var}}}

\newcommand{\erf}{\mbox{\text erf}}

\newcommand{\pr}[2][]{\mbox{Pr}\ifthenelse{\not\equal{}{#1}}{_{#1}}{}\!\left[#2\right]}

\newcommand{\normal}{\mathcal{N}}

\newcommand{\dtv}{d_{\mathrm {TV}}}

\newcommand{\E}{\mathop{\mbox{\bf E}}}

\newcommand{\Ssym}{\mathcal{S}_{\mathrm{sym}}}
\newcommand\numberthis{\addtocounter{equation}{1}\tag{\theequation}}

\newtheorem{theorem}{Theorem}

\newtheorem{remark}{Remark}
\newtheorem{fact}{Fact}
\newtheorem{lemma}{Lemma}
\newtheorem{claim}{Claim}
\newtheorem{corollary}{Corollary}

\newtheorem{definition}{Definition}

\newcommand{\ignore}[1]{}
\providecommand{\poly}{\operatorname*{poly}}

\newcommand{\bg}[1]{\medskip\noindent{\bf #1}}

\definecolor{Red}{rgb}{1,0,0}

\newcommand{\oldbound}[1]{{}}

\makeatletter
\newcommand{\algmargin}{\the\ALG@thistlm}
\makeatother
\algnewcommand{\parState}[1]{\State%
  \parbox[t]{\dimexpr\linewidth-\algmargin}{\strut #1\strut}}

\title{Robustly Learning a Gaussian: Getting Optimal Error, Efficiently}

\author {
Ilias Diakonikolas\thanks{Supported by NSF  CAREER Award CCF-1652862, a Sloan Research Fellowship, and a Google Faculty Research Award.} \\
CS, USC \\
\tt{diakonik@usc.edu}
\and
Gautam Kamath\thanks{Supported by NSF CCF-1551875, CCF-1617730, CCF-1650733, and ONR N00014-12-1-0999.} \\
EECS \& CSAIL, MIT \\
\tt{g@csail.mit.edu}
\and
Daniel M. Kane\thanks{Supported by NSF  CAREER Award CCF-1553288 and a Sloan Research Fellowship.} \\
CSE \& Math, UCSD \\
\tt{dakane@cs.ucsd.edu}
\and
Jerry Li \thanks{Supported by NSF CAREER Award CCF-1453261, CCF-1565235, a Google Faculty Research Award, and an NSF Graduate Research Fellowship.}\\
EECS \& CSAIL, MIT \\
\tt{jerryzli@mit.edu}
\and
Ankur Moitra\thanks{Supported by NSF CAREER Award CCF-1453261, CCF-1565235, a Packard Fellowship, a Sloan Research Fellowship, a grant from the MIT NEC Corporation, and a Google Faculty Research Award.} \\
Math \& CSAIL, MIT \\
\tt{moitra@mit.edu}
\and
Alistair Stewart\thanks{Supported by a USC startup grant.}\\
CS, USC \\
\tt{alistais@usc.edu}
}

\begin{document}
\maketitle
\begin{abstract}
\normalsize
We study the fundamental problem of learning the parameters of a high-dimensional Gaussian in the presence of noise \---- where an $\epsilon$-fraction of our samples were chosen by an adversary. We give robust estimators that achieve estimation error $O(\epsilon)$ in the total variation distance, which is optimal up to a universal constant that is independent of the dimension. 

In the case where just the mean is unknown, our robustness guarantee is optimal up to a factor of $\sqrt{2}$ and the running time is polynomial in $d$ and $1/\epsilon$. When both the mean and covariance are unknown, the running time is polynomial in $d$ and quasipolynomial in $1/\epsilon$. Moreover all of our algorithms require only a polynomial number of samples. Our work shows that the same sorts of error guarantees that were established over fifty years ago in the one-dimensional setting can also be achieved by efficient algorithms in high-dimensional settings. 
\end{abstract}
\thispagestyle{empty}

\newpage

\setcounter{page}{1}
\section{Introduction} 

\subsection{Background}

The most popular and widely used modeling assumption is that data is approximately Gaussian. 
This is a convenient simplification to make when modeling velocities of particles in an ideal gas~\cite{Goodell15}, 
measuring physical characteristics across a population (after controlling for gender), 
and even modeling fluctuations in a stock price on a logarithmic scale. However, 
real data is not actually Gaussian and is at best crudely approximated by a Gaussian (e.g., with heavier tails). 
What's worse is that estimators designed under this assumption can perform poorly in practice 
and be heavily biased by just a few errant samples that do not fit the model.  

For over fifty years, the field of robust statistics~\cite{Huber09, HampelEtalBook86, rousseeuw2005robust} has studied exactly 
this phenomenon \---- the sensitivity or insensitivity of estimators to small deviations in the model. Unsurprisingly, 
one of the central questions that shaped its development was the problem of learning the parameters 
of a one-dimensional Gaussian distribution when a small fraction of the samples are arbitrarily corrupted. 
More precisely, in 1964, Huber \cite{huber1964} introduced the following model:

\begin{definition}
In {\em Huber's contamination model}, we are given samples from a distribution 
$$\mathcal{D} = (1-\epsilon)\mathcal{N}(\mu, \sigma^2) + \epsilon \mathcal{Z} \;,$$ 
where $\mathcal{N}(\mu, \sigma^2)$ is a Gaussian of mean $\mu$ and variance $\sigma^2$, 
and $\mathcal{Z}$ is an arbitrary distribution chosen by an adversary. 
\end{definition}

Intuitively, among our samples, about a $(1-\epsilon)$ fraction will have been generated from a Gaussian 
and are called {\em inliers}, and the rest are called {\em outliers} or {\em gross corruptions}. 
We will work with an even more challenging\footnote{None of the results in our paper were previously 
known in Huber's contamination model either. The reason we work with this stronger model is because 
we can \---- nothing in our analysis relies on the inliers and outliers being independent.} model 
\---- called the {\em strong contamination model} (Definition~\ref{def:strongcontam}) \---- where 
the adversary is allowed to look at the inliers and then decide on the outliers. The literature on robust statistics 
has given numerous explanations and empirical investigations \cite{gelman2014bayesian, hampel2001robust} 
into how such outliers might arise as the result of equipment failure, data being entered incorrectly, 
or even from a subpopulation that was not accounted for in a medical study. 
These types of errors are erratic and difficult to model, so instead our goal is to design a procedure 
that accurately estimates $\mu$ and $\sigma^2$ without making any assumptions about them.

 
In one dimension, the median and median absolute deviation are well-known robust estimators for the mean and variance respectively. 
In particular, given samples $X_1, X_2, \ldots, X_n$, we can compute 
$$\widehat{\mu}= \mbox{median}(X_1, X_2, \ldots, X_n) \mbox{ and } \widehat{\sigma} = \frac{\mbox{median}(|X_i - \widehat{\mu}|)}{\Phi^{-1}(3/4)} \;,$$
where $\Phi$ is the cumulative distribution of the standard Gaussian. 
(This scaling constant is needed to ensure that $\widehat{\sigma}$ is an unbiased estimator when there is no noise.) 
If $n \geq C \frac{\log 1/\delta}{\epsilon^2}$, then with probability at least $1-\delta$ we have that 
$d_{TV}(\mathcal{N}(\mu, \sigma^2), \mathcal{N}(\widehat{\mu}, \widehat{\sigma}^2)) \leq C \epsilon$. 
In Huber's contamination model, this is the strongest type of error guarantee we could hope for\footnote{See Lemma~\ref{lem:gaussian-huber-LB}.} 
and captures both the task of learning the underlying parameters $\mu$ and $\sigma^2$, 
and finding the approximately best fit to the observed distribution within the family of one-dimensional Gaussians. 
In fact there are plentifully many other estimators \---- such as the {\em trimmed mean}, {\em winsorized mean}, 
{\em Tukey's biweight function}, and the {\em interquartile range} \---- that achieve the same sorts of error guarantees, up to constant factors. 
The design of robust estimators for {\em location} (e.g., estimating $\mu$) and {\em scale} (e.g., estimating $\sigma^2$) is guided 
by certain overarching principles, such as the notion of the influence curve \cite{HampelEtalBook86} or the notion 
of breakdown point \cite{rousseeuw2005robust}. In some cases, it is even possible to design robust estimators 
that are minimax optimal \cite{huber1964}. 

These days, much of modern data analysis revolves around high-dimensional data \---- for example, when we model 
documents \cite{lda}, images \cite{olshausenfield}, and genomes \cite{novembre} as vectors in a very high-dimensional space. 
The need for robust estimators is even more pressing in these applications, since it is infeasible to remove obvious outliers by inspection. 
However, adapting robust statistics to high-dimensional settings is fraught with challenges. The principles that guided the design 
of robust estimators in one dimension seem to inherently lead to high-dimensional estimators that are hard to compute \cite{Bernholt, HardtM13}. 

In this paper, we focus on the central problem of learning the parameters of a multivariate Gaussian $\mathcal{N}(\mu, \Sigma)$ in the strong contamination model. 
The textbook estimators for the mean and covariance -- such as the {\em Tukey median}~\cite{Tukey75} 
and {\em minimum volume enclosing ellipsoid}~\cite{Rous85} --  essentially 
search for directions where the projection of $\mathcal{D}$ is suitably non-Gaussian. 
However, trying to find a direction where the projection 
is non-Gaussian can be like looking for a needle in an exponentially-large haystack -- these statistics are not efficiently computable, in general.
Furthermore, a random projection will look Gaussian 
with high probability \cite{klartag}.

In this paper, our main result is an efficiently computable estimator for a high-dimensional Gaussian 
that achieves error $$d_{TV}(\mathcal{N}(\mu, \Sigma), \mathcal{N}(\widehat{\mu}, \widehat{\Sigma})) \leq C \epsilon$$ 
in the strong contamination model, for a universal constant $C$ that is independent of the dimension. 
For a Gaussian distribution, we consider estimation in terms of total variation distance, which is equivalent to estimating the parameters under the natural measures.
Our main idea is to use various regularity conditions satisfied by the inliers to make the problem of searching 
for non-Gaussian projections easier. When just the mean $\mu$ is unknown, our algorithm runs 
in time polynomial in the dimension $d$ and $1/\epsilon$. When both the mean and covariance are unknown, 
our algorithm runs in time polynomial in $d$ and quasi-polynomial in $1/\epsilon$. 
All of our algorithms achieve polynomial sample complexity. 
  
Prior to our work, the best known algorithm of Diakonikolas et al. \cite{DKKLMS} 
achieved estimation error $O(\epsilon \log 1/\epsilon)$ for this problem\footnote{
We note that, as stated, the results in~\cite{DKKLMS} give estimation error $O(\ve \log^{3/2} 1/\ve)$. 
However, combining the techniques in~\cite{DKKLMS} with the arguments in Section~\ref{sec:full} of this paper gives the stated bound. 
This argument will be included in the full version of~\cite{DKKLMS}.}, 
again with respect to total variation distance. Concurrently, Lai, Rao and Vempala \cite{LaiRV16} 
gave an algorithm which achieves estimation error roughly $O(\epsilon^{1/2} \log^{1/2} d)$. 
In fact, the algorithm of Diakonikolas et al. \cite{DKKLMS} works in a stronger model 
than what we consider here, where an adversary gets to look at the samples and then decides 
on an $\epsilon$-fraction to move arbitrarily. Such errors are both additive and subtractive (because inliers are removed). 
Interestingly, Diakonikolas, Kane and Stewart \cite{DiakonikolasKS16c} proved that any Statistical Query learning algorithm 
that works in such an additive and subtractive model and achieves an error guarantee asymptotically better 
than $O(\epsilon \log^{1/2} 1/\epsilon)$ must make a super-polynomial number of statistical queries. 
Our work shows a natural conclusion that in an additive only model it is possible to algorithmically 
achieve the same error guarantees as are possible in the one-dimensional case, up to a universal constant. 

\subsection{Our Results and Techniques}

In what follows, we will explain both our work as well as prior work through the following lens: 
\begin{quote}
{\em At the core of any robust estimator is some procedure to certify that the estimates 
have not been moved too far away from the true parameters by a small number of corruptions.}
\end{quote}

First, we consider the subproblem where the covariance $\Sigma = I$ is known and only the mean $\mu$ is unknown. 
In the terminology of robust statistics, this is called {\em robust estimation of location}. 
If we could compute the Tukey median, we would have an estimate that satisfies 
$d_{TV}(\mathcal{N}(\mu, I), \mathcal{N}(\widehat{\mu}, I)) \leq C \epsilon$.
The way that the Tukey median guarantees that it is close to the true mean is that along every direction $u$ 
it is close to the median of the projection of the samples. More precisely, at least a $\frac{1-\epsilon}{2}$ fraction 
of the samples satisfy $u^T X_i \geq u^T \widehat{\mu}$, and at least a $\frac{1-\epsilon}{2}$ fraction of the samples 
satisfy $ u^T \widehat{\mu} \geq u^T X_i $. However, if we have a candidate $\widehat{\mu}$, 
finding a direction $u$ that violates this condition is again like searching for a needle in an exponentially large haystack. 
  
The approach of Diakonikolas et al. \cite{DKKLMS} was essentially a data-dependent way 
to search for appropriate directions $u$, by looking for directions where the empirical variance 
is larger than it should be (if there were no corruptions). However, because their approach considers 
only a single direction at a time, it naturally gets stuck at  error $\Theta(\epsilon \log^{1/2} 1/\epsilon)$. 
This is because along the direction $u$, only when a point is $\Omega(\log^{1/2} 1/\epsilon)$ away 
from most of the rest of the samples can we be relatively confident that it is an outlier. 
Thus, an adversary could safely place all the corruptions in the tails and move 
the mean by as much as $\Theta(\epsilon \log^{1/2} 1/\epsilon)$. This would not affect the Tukey median by as much, 
but would affect an estimate based on the empirical mean (because the algorithm could find no other outliers to remove) 
by considerably more. 
 
Our approach is to consider logarithmically many directions at once. 
Even though an inlier can be logarithmically many standard deviations away from the mean along a single direction $u$ with reasonable probability, 
it is unlikely to be that many standard deviations away simultaneously across many orthogonal directions. 
Essentially, this allows us to remove the influence of outliers on all but a logarithmic dimensional subspace. 
Combining this with an algorithm for robustly learning the mean in time exponential in the dimension 
(but polynomial in the number of samples), we obtain our first main result:
 
 \begin{theorem}
\label{thm:mean-intro}
Suppose we are given a set of $n = \poly(d, 1/\epsilon)$ samples from the strong contamination model, 
where the underlying $d$-dimensional Gaussian is $\mathcal{N}(\mu, I)$.   Let $\epsilon \leq \epsilon_0$, 
where $\epsilon_0$ is a positive universal constant. For any $\beta > 0$, there is an algorithm to learn an estimate 
$\mathcal{N}(\widehat \mu, I)$ that with high probability satisfies
$$d_{TV}(\mathcal{N}(\mu, I), \mathcal{N}(\widehat \mu, I)) \leq  \left( \frac{1}{\sqrt{2}} + O\left(\frac{1}{\sqrt{\beta}} + \epsilon^2 \right) \right) \eps \;.$$
Moreover, the algorithm runs in time $\poly(n, (1/\epsilon)^\beta)$. 
\end{theorem}

We prove an almost matching lower bound of $\frac{\epsilon}{2} + \Omega(\epsilon^2)$ on the estimation error. 
Thus, our robustness guarantee is optimal up to a factor of $\sqrt{2}$, 
even among computationally inefficient robust estimators. Interestingly, 
our extra factor of $\sqrt{2}$ comes from the following geometric fact which we make crucial use of: 
Any convex body of diameter $D$ in any dimension can be covered by a ball of radius $ D/\sqrt{2}$, 
and moreover such a ball can be (approximately) found in time exponential in the dimension. 
Suppose that along some direction $u$ we have an estimate $p$ that is guaranteed to be within 
$\epsilon/2$ of the projection of the true mean $\mu$. We can now confine $\mu$ to a slab of 
width $\epsilon$, and by taking the intersection of all such slabs we get a convex body 
that contains $\mu$ and has diameter of at most $\epsilon$. By covering the body with a ball of radius 
$ \epsilon/ \sqrt{2}$, we are guaranteed that the center of the ball is within $ \epsilon / \sqrt{2}$ of the true mean. 
This gives us a general way to combine one-dimensional robust estimates along a net of directions. 

We note that, for general isotropic sub-Gaussian distributions, the bound of $O(\ve \log^{1/2} 1/\ve)$ of~\cite{DiakonikolasKKLMS17} is optimal for robust mean estimation, even in one dimension.
See Section~\ref{sec:app-lb} for a proof of this fact.
However, our results can be seen to hold more generally than stated above -- indeed, the same arguments work for a class of symmetric isotropic sub-Gaussian distributions which are sufficiently smooth near their mean.
More precisely, we require that along any univariate projection, the mean is robustly estimated by the median.

We next consider the subproblem where the mean $\mu = 0$ is known 
and only the covariance $\Sigma$ is unknown. In the terminology of robust statistics, 
this is called {\em robust estimation of scale}. In this case, we want to compute an estimate $\widehat{\Sigma}$ 
that satisfies\footnote{More precisely, to obtain $O(\eps)$ error guarantee 
with respect to the total variation distance, we need to robustly approximate $\Sigma$ within $O(\eps)$ in Mahalanobis distance, 
which is a stronger metric than the Frobenius norm. As part of our approach, we are able to efficiently reduce to the case that 
$\Sigma$ is close to the identity matrix, in which case the Frobenius error suffices.} $\| \Sigma - \widehat{\Sigma}\|_F \leq C \epsilon$. 
When $\widehat{\Sigma}$ does not satisfy this condition, it can be shown (in Section~\ref{sec:poly-gaussian}) that there is a degree-two polynomial $p(X)$, where
$$\E_{X \sim \mathcal{N}(0, \Sigma)} [ p(X)] =1 \mbox{ and } \E_{X \sim \mathcal{N}(0, \widehat{\Sigma})} [ p(X)] =1 + C' \epsilon \;.$$ 
It turns out that, even given the polynomial $p(X)$, deciding whether or not the above conditions approximately hold is challenging. 
Given $p(X)$ and $\widehat{\Sigma}$, we can certainly compute $\E_{X \sim \mathcal{N}(0, \widehat{\Sigma})} [ p(X)]$. 
But given only contaminated samples from $\mathcal{N}(0, \Sigma)$ and without knowing what $\Sigma$ is, 
can we estimate $\E_{X \sim \mathcal{N}(0, \Sigma)} [ p(X)]$?

Often, univariate robust estimation problems are considered easy, with a simple recipe:
Construct an unbiased estimator for the statistic for which each sample point has low influence.
However, in our setting, it is highly non-trivial to construct such an estimator.
The naive attempt in this case would be the median -- this immediately fails since the distribution of $p(X)$ is asymmetric.
Even if there were no noise, that would not necessarily be an unbiased estimator. 
So how can we dampen the influence of outliers, if there is no natural symmetry in the distribution? 
We construct a robust estimator crucially using the fact that $p(X)$ is the weighted sum of chi-squared 
random variables when there is no noise. The key structural fact we exploit is the following: Given two sums of chi-squared random variables, 
if the random variables are far in total variation distance, 
most of their difference must lie close to their means. 
We use this fact to show how, given a weak estimate of the mean (i.e., one which is only accurate up $\omega(\ve)$), one can improve the estimate by a constant factor.
Our result follows by an iterative application of this technique.

However, there is still a major complication in utilizing our low-dimensional estimator 
to obtain a high-dimensional estimator. In the unknown mean case, we knew the higher-order moments 
(since we assumed that the covariance is the identity). Here, we do not have control over the higher-order 
moments of the unknown Gaussian. Overcoming this difficulty requires several new techniques, 
which are quite complicated, and we defer the full details to Section \ref{sec:covHighDim}. 
Our second main result is:
  
 \begin{theorem}
Suppose we are given a set of $n = \poly(d, 1/\epsilon)$ samples from the strong contamination model, 
where the underlying $d$-dimensional Gaussian is $\mathcal{N}(0, \Sigma)$. 
There is an algorithm to learn an estimate $\normal (0, \widehat \S)$ that runs in time $\poly(n, (1/\epsilon)^{O(\log^4 1/\epsilon)})$
and with high probability satisfies
$$d_{TV}(\mathcal{N}(0, \Sigma), \mathcal{N}(0, \widehat \Sigma)) \leq C\epsilon \;,$$
for a universal constant $C$ that is independent of the dimension. 
\end{theorem}

A key technical problem arises when we attempt to combine estimates for the covariance 
restricted to a subspace and its orthogonal complement. We refer to this as a {\em stitching} problem, 
where if we write $\Sigma$ as
$$
\Sigma = \left[ \begin{matrix} \Sigma_V & A^T \\ A & \Sigma_{V^\perp} \end{matrix}\right] \;,
$$
and have accurate estimates for $\Sigma_V$ and $\Sigma_{V^\perp}$, 
we still need to accurately estimate $A$. Our algorithm utilizes an unexpected connection to the unknown mean case: 
We show that, under a carefully chosen projection scheme, we can simulate noisy samples from a Gaussian 
with identity covariance, where the mean of this distribution encodes the information needed to recover $A$. 
We defer the full details to Section~\ref{sec:stitching}. 
  
It turns out that we can solve the general case when both $\mu$ and $\Sigma$ are unknown, 
by directly reducing to the previous subproblems, exactly as was done in \cite{DKKLMS} (with some caveats, addressed in Section~\ref{sec:inexact-cov}). 
Since all of our error guarantees are optimal up to constant factors, there is only a constant factor loss in this reduction. 
Finally, we obtain the following corollary:
 
\begin{corollary}
Suppose we are given a set of $n = \poly(d, 1/\epsilon)$ samples from the strong contamination model, 
where the underlying $d$-dimensional Gaussian is $\mathcal{N}(\mu, \Sigma)$. There is an algorithm 
to learn an estimate $\normal (\widehat \m, \widehat \S)$ that runs in time $\poly(n, (1/\epsilon)^{O(\log^4 1/\epsilon)})$ and 
with high probability satisfies
$$d_{TV}(\mathcal{N}(\mu, \Sigma), \mathcal{N}(\widehat \mu, \widehat \Sigma)) \leq C\epsilon \;,$$
for a universal constant $C$ that is independent of the dimension. 
\end{corollary}

This essentially settles the complexity of robustly learning a high-dimensional Gaussian. 
The sample complexity of our algorithm depends polynomially on $d$ and $1/\epsilon$, 
and the running time depends polynomially on $d$ and quasi-polynomially on $1/\epsilon$. 
Up to a constant factor, ours is the first high-dimensional algorithm 
that achieves the same error guarantees as in the one-dimensional case, 
where results were known for more than fifty years! It is an interesting open problem 
to reduce the running time to polynomial in $1/\epsilon$ (while still being polynomial in $d$). 
As we explain in Section~\ref{sec:quasibarrier}, this seems to require fundamentally new ideas.

\subsubsection*{More Related Work}  
In addition to the works mentioned above, there has been an exciting flurry of recent work 
on robust high-dimensional estimation. This includes studying graphical models in the presence of noise~\cite{DiakonikolasKS16b}, 
tolerating much more noise by allowing the algorithm to output a list of candidate hypotheses~\cite{CharikarSV16}, 
formulating general conditions under which robust estimation is possible~\cite{SteinhardtCV17}, developing robust algorithms 
under sparsity assumptions~\cite{Li17, DBS17, BalakrishnanDLS17} where the number of samples is sublinear in the dimension, 
and leveraging theoretical insights to give practical algorithms that can be applied to genomic data~\cite{DiakonikolasKKLMS17}. 
We note that, in comparison to all these other works, ours is the only to efficiently achieve the information-theoretically optimal error guarantee (up to constant factors).
Despite all of this rapid progress, there are still many interesting theoretical and practical questions left to explore. 

\subsection{Organization}
In Section~\ref{sec:prelims}, we go over preliminaries and notation that we will use throughout the paper.
In Section~\ref{sec:meanLowD}, we describe an algorithm for robustly estimating the mean of a Gaussian in low-dimensional settings, and crucially apply it in the design of an algorithm for mean-estimation in high dimensions, described in Section~\ref{sec:mean}.
Similarly, in Section~\ref{sec:covLowD}, we give an algorithm for robustly estimating the mean of degree-two polynomials in certain settings, which is applied in the context of our covariance-estimation algorithm in Section~\ref{sec:covHighDim}.
Finally, we put these tools together and describe our general algorithm for robustly estimating a Gaussian in Section~\ref{sec:full}.

\section{Preliminaries}
\label{sec:prelims}
In this section, we give various definitions and lemmata we will require throughout the paper. 
First, given a distribution $F$, we let $\E_F [f(X)] = \E_{X \sim F} [f(X)]$ denote the expectation of $f(X)$ under $F$.
If $S$ is a finite set, we let $\E_S [f(X)] = \E_{X \sim \mathrm{unif} (S)} [f(X)]$ denote the expectation of $f(X)$ 
under the uniform distribution over points in $S$ (i.e., the empirical mean of $f$ under $S$).
Given any subspace $V \subseteq \R^d$, we let $\Pi_V: \R^d \to \R^d$ be the projection operator onto $V$.
If $V = \mathrm{span} (v)$ is 1-dimensional, we will denote this projection as $\Pi_v$.

\subsection{The Strong Contamination Model}

Here we formally define the {\em strong contamination model}. 

\begin{definition}\label{def:strongcontam}
Fix $\eps > 0$.
We say a set of samples $X_1, \ldots, X_n$ was generated from the {\em strong contamination model} 
on a distribution $F$, if it was generated via the following process:
\begin{enumerate}
\item
We produce $(1 - \eps) n$ i.i.d. samples $G$ from $F$.
\item
An adversary is allowed to observe these samples and \emph{add} $\eps n$ points $E$ arbitrarily.
\end{enumerate}
We are then given the set of samples $G \cup E$ in random order. 
Also, we will say that the samples $X_1, \ldots, X_n$ are $\eps$-corrupted. 
Moreover given an $\eps$-corrupted set of samples $S$, we will write $S = (G, E)$ where $G$ is the set of \emph{uncorrupted} points and $E$ is the set of \emph{corrupted} points.
Moreover, given a subset $S' \subset S$, we will also write $S' = (G', E')$, where $G' = S' \cap G$ and $E' = S' \cap E$ denote the set of uncorrupted points and corrupted points remaining in $S'$.
$L$ will denote $G \setminus G'$, which is the set of ``lost'' uncorrupted points.
\end{definition}

Given a contaminated set $S' = (G', E')$ and a set $G$ so that $G' \subseteq G$, define the following quantities 
\begin{equation}
\label{eq:Delta}
\phi (S', G) = \frac{|G \setminus G'|}{|S'|} \;,~~~~  \psi (S', G) = \frac{|E'|}{|S'|} \;,~~~~ \Delta (S', G) = \psi (S', G) + \phi (S', G) \log \frac{1}{\phi (S', G)} \; .
\end{equation}
\noindent
In particular, observe that if $\Delta (S', G) < O(\eps)$, then a simple calculation implies that $\phi(S', G) \leq O(\eps/\log 1/\eps)$.
Equivalently, we have removed at most an $O(\eps / \log 1 / \eps)$ fraction of good points from $G$.
This is crucial, as if we throw out an $\eps$-fraction of good points then we essentially put ourselves in the subtractive model, and there our guarantees no longer hold.

There are two differences between the strong contamination model and Huber's contamination model. First, the number of corrupted points is fixed to be $\epsilon n$ instead of being a random variable. However, this difference is negligible. It follows from basic Chernoff bounds that $n$ samples from Huber's contamination model with parameter $\eps$ (for $n$ sufficiently large) can be simulated by a $(1 + o(1)) \eps$-corrupted set of samples, except with negligible failure probability.
Hence, we lose only an additive $o( \eps)$ term when translating from Huber's contamination model to the strong contamination model, which will not change any of the guarantees in our paper. The second difference is that the adversary is allowed to inspect the uncorrupted points before deciding on the corrupted points. This makes the model genuinely stronger since the samples we are given are no longer completely independent of each other. 

\subsection{Deterministic Regularity Conditions}

In analyzing our algorithms, we only need certain deterministic regularity conditions to hold on the uncorrupted points. In this subsection, we formally state what these conditions are. It follows from known concentration bounds that these conditions all hold with high probability given a polynomial number of samples. Now with these regularity conditions defined once and for all, we will be able to streamline our proofs in the sense that each step in the analysis will only ever use one of these fixed set of conditions and will not use the randomness in the sampling procedure. We remark that some subroutines in our algorithm only need a subset of these conditions to hold, so we could improve the sample complexity by changing the regularity conditions we need at each step. However, since we will not be concerned with optimizing the sample complexity beyond showing that it is polynomial, we choose not to complicate our proofs in this manner.

\subsubsection{Regularity Conditions for Unknown Mean}
In the unknown mean case, we will require the following condition:
\begin{definition}
Let $G$ be a multiset of points in $\R^d$ and $\eta, \d >0$. 
We say that $G$ is $(\eta, \d)$-good with respect to $\normal (\mu, I)$ if the following hold:

\begin{itemize}
\item For all $x \in G$ we have $\|x-\mu\|_2 \leq O(\sqrt{d \log (|G|/\delta)})$.
\item For every affine function $L:\R^d \to \R$ we have
$|\Pr_{G}(L(X) \ge 0) - \Pr_{\normal (\mu, I)}(L(X) \ge 0)| \leq \eta/(d \log(d/\eta\d)) \;.$
\item We have that $\|\E_G [X] - \E_{\normal (\m, I)} [X] \|_2\leq \eta.$
\item We have that $\left\|\mathrm{Cov}_G [X] - I \right\|_2 \leq \eta/d.$
\item For any even degree-$2$ polynomial $p:\R^d\rightarrow \R$ we have that
\begin{align*}
\left|\E_{G}[p(X)] - \E_{\normal (\m, I)}[p(X)]\right| &\leq \eta \E_{\normal (\m, I)}[p^2(X)]^{1/2} ,\;  \\
\left|\E_{G}[p^2(X)] - \E_{\normal (\m, I)}[p^2(X)]\right| &\leq \eta \E_{\normal (\m, I)}[p^2(X)] ,\; \mathrm{and} \\
\Pr_{G} [p(X) \geq 0] &\leq  \Pr_{\normal (\mu, I)} [p(X)\geq 0] + \frac{\eta}{d \log(|G|/\d)} \; .
\end{align*}
\end{itemize}
\end{definition}
\noindent It is easy to show (see Lemma \ref{lem:GoodSamplesOpt}) that given enough samples from $\normal (\mu, I)$, the empirical data set will satisfy these conditions with high probability.

\subsubsection{Regularity Conditions for Unknown Covariance}
\label{sec:cov-reg}
In the unknown covariance case, we will require the following condition:
\begin{definition}
Let $G$ be a set of $n$ points of $\R^d$, and $\eta, \delta>0$. We say that $G$ is $(\eta, \delta)$-good with respect to $\normal (0, \Sigma)$ if the following hold:
\begin{itemize}
\item For all $x \in G$ we have that $x^T \Sigma^{-1} x = O( d \log(|G| / \delta)).$
\item For any even degree-$2$ polynomial $p: \R^d \to \R$ we have
\begin{align*}
\left| \E_{G} [p(X)] - \E_{\normal (0, \Sigma)} [p(X)] \right| &\leq \eta \E_{\normal (0, \Sigma)} [p^2(X)]^{1/2}\; , \\
\left| \E_{G} [p^2 (X)] - \E_{\normal (0, \Sigma)} [p^2(X)] \right| &\leq \eta \E_{\normal (0, \Sigma)} [p^2 (X)] \; \mbox{and} \\
\Pr_{X \sim G} [p(X) \geq 0] &\leq  \Pr_{\normal (0, \Sigma)} [p(X)\geq 0] + \frac{\eta^2}{d \log(|G| / \delta)} \; .
\end{align*}
\item
For any even degree-$4$ polynomial $p: \R^d \to \R$ we have
\begin{align*}
\left| \E_{G} [p(X)] - \E_{\normal (0, \Sigma)} [p(X)] \right| &\leq \eta \Var_{\normal (0, \Sigma)} [p(X)]^{1/2}\; , \\
\Pr_{G} [p(X) \geq 0] &\leq  \Pr_{\normal (\mu, I)} [p(X)\geq 0] + \frac{\eta^2}{2 \log (1 / \eps) (d \log(|G| / \delta))^2} \; .
\end{align*}
\end{itemize}
\end{definition}
\noindent As before, it is easy to show (see Lemma \ref{lem:GoodSamplesOptCov}) that given enough samples from $\normal (0, \Sigma)$, the empirical data set will satisfy these conditions with high probability.

\subsection{Bounds on the Total Variation Distance}

We will require some simple bounds on the total variation distance between two Gaussians. 
These bounds are well-known.
Roughly speaking, they say that the total variation distance between two Gaussians with identity covariance is governed by the $\ell_2$ norm between their means, and the total variation distance between two Gaussians with mean zero is governed by the Frobenius norm between their covariance matrices, provided that the matrices are close to the identity.

\begin{lemma}
\label{lem:mubound}
Let $\mu_1, \mu_2 \in \R^d$ be such that $\| \mu_1 - \mu_2 \|_2 = \eps$ for $\ve < 1$.
Then
\[
\dtv\left(\normal(\m_1, I), \normal(\m_2, I)\right) = \left( \frac{1}{\sqrt{2 \pi}} + o(1)\right) \eps \; .
\]
\end{lemma}
\noindent For clarity of exposition we defer this calculation to the Appendix.

We also need to bound the total variation distance between two Gaussians with zero mean and different covariance matrices. The natural norm to use is the Mahalanobis distance. But in our setting, we will be able to use the more convenient Frobenius norm instead (because we effectively reduce to the case that 
the covariance matrices will be close to the identity): 

\begin{lemma}[Cor. 2.14 in \cite{DKKLMS}]
\label{lem:frobBound}
Let $\Sigma, \Shat$ be such that $\| \Sigma - I \|_F \leq O(\eps \log 1 / \eps)$, and $\| \Sigma - \Shat \|_F \leq C \eps$.
Then $\dtv (\normal (0, \Sigma), \normal (0, \Shat)) \leq O(\eps)$.
\end{lemma}

These lemmata show that parameter estimation and approximation in total variation distance are essentially equivalent.
Indeed, in this paper, we achieve both guarantees, but state our results in terms of total variation estimation.


\section{Robustly Learning the Mean in Low Dimensions}
\label{sec:meanLowD}

This section is dedicated to the proof of the following theorem:
\begin{theorem}
\label{lowDimLearningLem}
Fix $\mu \in \R^d$, and let $\eps, \gamma, \delta > 0$.
Let $S_0 = (G_0, E_0)$ be such that $G_0$ is a $(\gamma \eps, \delta)$-good set with respect to $\normal (\mu, I)$, and $|E_0| / |S_0| \leq \eps$.
Let $S = (G, E)$ be another set with $\Delta (S, S_0) < \eps$.
Let $V \subseteq \R^d$ be a subspace.
For all $0 < \rho < 1$, the algorithm $\textsc{LearnMeanLowD} (V, \gamma, \eps, \delta, S, \rho)$ runs in time $\poly (d, |S|, (1 / \rho)^{O(\dim(V))}, \log (\rho \eps / (1 - \rho)), \log (1 / \rho))$ and returns a $\mutilde$ so that 
\[
\| \Pi_V( \mu - \mutilde) \|_2 = \frac{1 + 2 \rho}{1 - \rho} \left(\sqrt{\pi} +  O\left( \frac{\gamma}{d}\right) \right) \eps \; .\]
\end{theorem}
\noindent In particular, as we let $\rho, \gamma \to 0$, the parameter estimation error approaches $\sqrt{\pi} \eps$ (corresponding to a total variation approximation of $\ve/\sqrt{2}$). In Lemma \ref{lem:gaussian-huber-LB} in the Appendix we show that no algorithm can achieve parameter estimation error better than $\sqrt{\frac{\pi}{2}} \eps$. Thus, we achieve a $\sqrt{2}$ approximation to the optimal error. 

For simplicity, in the rest of this section, we will let $V = \R^d$, that is, we assume there is no projection.
It should be clear that this can be done without loss of generality. Our algorithm proceeds as follows:
First, we show that in one dimension, the median produces an estimate which is optimal, up to lower order terms, if the sample set is $(\gamma \eps, \delta)$-good with respect to the underlying Gaussian.
Then, we show that by using a net argument, we can produce a convex body in $\R^d$ with diameter at most $2 R = 2 (\sqrt{\frac{\pi}{2}} + o(1)) \eps$ which must contain the true mean. Finally, we use an old result of Jung \cite{Jung01} that such a set can be circumscribed by a ball of radius $\sqrt{2} R$ (see \cite{BW41} for an English language version of the result). We use the center of the ball as our estimate $\mutilde$. 

\subsection{Robustness of the Median}
First we show that if we project onto one dimension, then the median of the corrupted data differs from the true mean by at most $\sqrt{\frac{\pi}{2}} \eps + o(\eps)$. Our proof will rely only on the notion of a $(\gamma \epsilon, \delta)$-good set with respect to $\mathcal{N}(\mu, I)$ and thus it works even in the strong contamination model. 
Formally, we show:

\begin{lemma}
\label{lem:median}
Fix any $v \in \R^d$.
Fix $\mu \in \R^d$, and let $\delta > 0$.
Let $S_0 = (G_0, E_0)$ be so that $G_0$ be a $(\gamma \eps, \delta)$-good set with respect to $\normal (\mu, I)$, and $|E_0| / |S_0| \leq \eps$.
Let $S = (G, E)$ be another set with $\Delta (S, S_0) < \eps$.
Let $b$ be the median of $S$ when projected onto $v$.
Then, $|b - \Pi_v \mu | \leq \left(\sqrt{\frac{\pi}{2}} +  O\left( \frac{\gamma}{d}\right) \right) \eps$.
\end{lemma}
\begin{proof}
For any $a \in \R$, we have
\begin{equation*}
\Pr_{X \sim S} [\langle v, X \rangle > a] = \frac{|G|}{|S|} \Pr_{X \sim G} [\langle v, X \rangle > a] + \frac{|E|}{|S|} \Pr_{X \sim E} [\langle v, X \rangle > a] \; .
\end{equation*}
Observe that we have $\left| \frac{|E|}{|S|} \Pr_{X \sim E} [\langle v, X \rangle > a] \right| \leq \psi (S, G)$.
Moreover, by simple calculation we have
\begin{align*}
\left| \Pr_{X \sim G} [\langle v, X \rangle > a] - \Pr_{X \sim G_0} [\langle v, X \rangle > a] \right| &\leq 2 \phi (S, G) \; .
\end{align*}
Hence, we have
\[
\left| \Pr_{X \sim S} [\langle v, X \rangle > a] - \frac{|G|}{|S|} \Pr_{X \sim G_0} [\langle v, X \rangle > a] \right| \leq \psi (S, G) + o(\eps) \; ,
\]
since by assumption $\Delta (S, G) \leq \eps$.
Similarly, we have that for all $a \in \R$, 
\[
\left| \Pr_{X \sim S} [\langle v, X \rangle < a] - \frac{|G|}{|S|} \Pr_{X \sim G_0} [\langle v, X \rangle < a] \right| \leq \psi (S, G) + o(\eps) \; .
\]
For $|a| = O(\eps)$ we have that $\Pr_{\normal (0, I)} [X > a] = \frac{1}{2} - \frac{1}{\sqrt{2 \pi}} a + O(\eps^3)$.
Thus, by $(\gamma \eps, \delta)$-goodness of $G_0$, this implies that for $|a| = O(\eps)$, we have
\begin{equation}
\left| \Pr_{X \sim S} [\langle v, X \rangle > \Pi_y \mu + a] - \frac{|G|}{|S|} \left( \frac{1}{2} - \frac{1}{\sqrt{2 \pi}} a \right) \right| \leq \psi (S, G) + O\left( \frac{\gamma \eps}{d}\right) \; .
\end{equation}
Thus, for $|a| = O(\eps)$ we have
\begin{align*}
 \Pr_{X \sim S} [\langle v, X \rangle > \Pi_y \mu + \sqrt{\frac{\pi}{2}} a] &\leq \frac{|G|}{|S|} \left( \frac{1}{2} -\frac{1}{\sqrt{2 \pi}} a \right) + \frac{|E|}{|S|} + O\left( \frac{\gamma \eps}{d}\right) + o(\eps) \\
 &\leq \frac{1}{2} - \frac{|G|}{|S|} \frac{1}{\sqrt{2 \pi}} a + \frac{|E|}{2 |S|} +  O\left( \frac{\gamma \eps}{d}\right) + o(\eps) \; .
\end{align*}
In particular, we see that if $a > \sqrt{\frac{\pi}{2}} \eps +  O\left( \frac{\gamma \eps}{d}\right) + o(\eps)$, then $\Pr_{X \sim S} \left[ \langle v, X \rangle > \Pi_y \mu + \sqrt{\frac{\pi}{2}} \eps \right] < 1/2$.
By symmetric logic, we also have that $\Pr_{X \sim S} \left[ \langle v, X \rangle > \Pi_y \mu - \sqrt{\frac{\pi}{2}} \eps \right] > 1/2$.
Thus, the median in direction $v$ differs from $\Pi_y \mu$ by at most $\sqrt{\frac{\pi}{2}} \eps +  O\left( \frac{\gamma \eps}{d}\right) +  o(\eps)$.
\end{proof}

\subsection{Finding a Minimum Radius Circumscribing Ball}
For any $x \in \R^d$ and $r > 0$, let $B(x, r) = \{y \in \R^d: \| x - y \|_2 \leq r \}$ denote the closed ball of radius $r$ centered at $x$.
The following classical result of Gale gives a bound on the radius of the circumscribing ball of any convex set in terms of its diameter:
\begin{theorem}[see \cite{Jung01, BW41}]
\label{thm:circumscribe1}
Fix $R> 0$.
Let $\cc \subseteq \R^d$ be a convex body so that for all $x, y \in \cc$, we have $\| x - y \|_2 \leq 2 R$.
Then $\cc$ is contained within a ball of radius $R \sqrt{2}$.
\end{theorem}

The bound is asymptotically achieved for the standard simplex as we increase its dimension. The goal of this subsection is to show that the (approximately) minimum radius circumscribing ball can be found efficiently.
We will assume we are given an \emph{approximate projection} oracle for the convex body that given a point $y \in \R^d$, outputs a point which is almost the closest point in $\cc$ to $x$:

\begin{definition}
A \emph{$\rho$-projection oracle} for a convex body $\cc$ is a function $\mathcal{O}: \R^d \to \R^d$, which, given a point $y \in \R^d$, outputs a point $x \in \cc$ so that $\| x - y \|_2 \leq \inf_{x' \in \cc} \| x' - y \|_2 + \rho $.
\end{definition}

Our first step is to use such an oracle to construct a net for $\cc$. First, we need the following well-known bound on the size of the net. 
\begin{claim}
\label{claim:eps-net}
Fix $r > 0$.
Then, for any $\beta > 0$, there is a $\beta$-net $\cF$ for the sphere of radius $r$ around $0$ in $\R^d$ of size $(r / \beta)^{O(d)}$.
Moreover, this net can be constructed in time $\poly (d, |\cF|)$.
\end{claim}

With this, we can show:
\begin{lemma}
\label{lem:circum-net}
Fix $R, \cc$ as in Theorem \ref{thm:circumscribe1}, and let $1 > \rho > 0$.
Let $x \in \cc$ be arbitrary.
Let $\co$ be a $(\rho R / 3)$-projection oracle for $\cc$.
Suppose a call to $\co$ runs in time $T$.
Then, there is an algorithm $\textsc{CircumscribeNet}(R, \rho, \co, x)$ 
which runs in time $\poly ((R / \rho)^{O(d)}, T)$ and 
outputs a set $\cX \subseteq \R^d$ so that $\cX$ is a $(\rho R)$-net for $\cc$, 
and moreover, $|\cX| \leq (R / \rho)^{O(d)}$.
\end{lemma}
The algorithm is fairly straightforward.
First, we observe that $\cc$ is contained within $B(x, 2R)$.
We then form a $(\rho R) / 3$-net of $B(x, 2R)$ using Claim \ref{claim:eps-net}.
We then iterate over every element $v$ of this net, and use our projection oracle to (approximately) find the closest point in $\cc$ to $v$.
If this point is too far away, we throw it out, otherwise, we add this projected point into the net.
The formal pseudocode for $\textsc{CircumscribeNet}$ is given in Algorithm \ref{alg:circumscribe-net}.

\begin{algorithm}[htb]
\begin{algorithmic}[1]
\Function{CircumscribeNet}{$R, \rho, \co, x$}
\State Form an $\rho / 3$-net $\cF'$ of the sphere of radius $2$ of size $(1 / \rho)^{O(d)}$ as in Claim \ref{claim:eps-net}.
\State Let $\cF = R \cdot \cF' + x$.
\State Let $\cX \gets \emptyset$
\For{each $v \in \cF$}
	\State Let $u_v \gets \co(v)$
	\If{$\| v - u_v \|_2 \leq 2 \rho R / 3 $} \label{line:cir-net-1}
		\State Add $u_v$ to $\cX$
	\EndIf
\EndFor
\State \textbf{return} $\cX$
\EndFunction
\end{algorithmic}
\caption{Generating a net of $\cc$}
\label{alg:circumscribe-net}
\end{algorithm}

\begin{proof}
The runtime bound follows from Claim \ref{claim:eps-net}.
We now turn our attention to correctness.
By Claim \ref{claim:eps-net}, and rescaling and shifting, the set $\cF$ is clearly a $(\rho R) / 3$-net for a ball $B$ of radius $2R$ containing $\cc$.
We now claim that the set $\cX$ is indeed a $(\rho R) / 3$-net for $\cc$.
Fix $y \in \cc$.
Since $\cc \subseteq B$, this implies there is some $v \in \cF$ so that $\| y - v \|_2 \leq \rho R / 3$.
Thus, in Line \ref{line:cir-net-1}, when processing $v$, we must find some $u_v \in \cc$ so that $\| u_v - v \|_2 \leq 2 \rho R / 3$.
The claim then follows from the triangle inequality.
\end{proof}

With this, we obtain:
\begin{corollary}
\label{cor:circumscribe2}
Fix $R, \cc, \rho, \co, x$ as in Lemma \ref{lem:circum-net}.
Suppose a call to $\co$ runs in time $T$.
Then, there is an algorithm $\textsc{Circumscribe}(R, \rho, \co, x)$ 
which runs in time $\poly ((R / \rho)^{O(d)}, T)$ and returns 
a point $\yhat$ so that $\cc$ is contained within a ball of radius $\sqrt{2} (1 + 2 \rho) R$.
\end{corollary}
The algorithm at this point is very simple.
Using the output of $\textsc{CircumscribeNet}$, we iterate over all points in a net over $B(x, 2R)$, 
find an $x$ in this net so that the distance to all points in the net is at most $\sqrt{2} (1 + \rho ) R$, and output any such point.
The formal pseudocode for $\textsc{Circumscribe}$ is given in Algorithm \ref{alg:circumscribe}.

\begin{algorithm}[htb]
\begin{algorithmic}[1]
\Function{Circumscribe}{$R, \rho, \co, x$}
\State Form an $\rho / 3$-net $\cF'$ of $B(0, 2)$ of size $(1 / \rho)^{O(d)}$ as in Claim \ref{claim:eps-net}.
\State Let $\cF = R \cdot \cF' + x$.
\State Let $\cX \gets \textsc{CircumscribeNet} (R, \rho, \co, x)$.
\For{each $v \in \cF$}
	\If{for all $u \in \cX$, we have $\| u - v \|_2 \leq \sqrt{2} (1 + \rho) R$}
		\State \textbf{return} $u$
	\EndIf
\EndFor
\EndFunction
\end{algorithmic}
\caption{Finding a circumscribing ball of small radius}
\label{alg:circumscribe}
\end{algorithm}

\begin{proof}
The runtime bound is immediate.
By Theorem \ref{thm:circumscribe1}, there is some $y \in B(x, 2R)$ so that $\cc \subseteq B(y, R \sqrt{2})$.
Thus, by the triangle inequality, there is some $y' \in \cF$ so that $\cc \subseteq B(y, \sqrt{2}( 1 + \rho) R)$.
Thus, the algorithm will output some point $y'' \in \cF$.
By an additional application of the triangle inequality, 
since $\cX$ is a $\rho R$-net for $\cc$, this implies that $\cc \subseteq B(y'', \sqrt{2} (1 + 2 \rho) R)$, as claimed.
\end{proof}

\subsection{The Full Low-Dimensional Algorithm}
We now have all the tools to describe the full algorithm in low-dimensions.
Let $S$ be our corrupted dataset as in Theorem \ref{lowDimLearningLem}.
Fix $\rho > 0$.
We first produce a $\rho$-net $\cF$ over the unit sphere in $\R^d$.
By (a slight modification of) Claim \ref{claim:eps-net}, this net has size $(1 / \rho)^{O(d)}$ and can be constructed in time $\poly (d, |\cF|)$.
For each $v \in \cF$, we project all points in $S$ onto $v$, and take the median of these points to produce $b_v$.
We then construct the following set:
\begin{equation}
\label{eq:low-dim}
\cc = \bigcap_{v \in \cF} \{y \in \R^d: \langle v, y \rangle \in \left[ b_v - \beta, b_v + \beta \right] \} \;,
\end{equation}
where $\beta = \sqrt{\frac{\pi}{2}} \eps +  O\left( \frac{\gamma \eps}{d}\right) + o(\eps)$ is as in Lemma \ref{lem:median}.
We now show two properties of this set, which in conjunction with the machinery above, allows us to prove Theorem \ref{lowDimLearningLem}.
The first shows that $\cc$ has small diameter:
\begin{claim}
\label{claim:radius}
For all $x, y \in \cc$, we have $\| x - y \|_2 \leq 2 \beta / (1 - \rho)$.
\end{claim}
\begin{proof}
Fix any $x, y \in \cc$.
By definition of $\cc$, it follows that for all $v \in \cF$, we have $|\langle x - y, v \rangle| \leq 2 \beta$.
For any $u$ with $\| u \|_2 = 1$, there is some $v \in \cF$ with $\| u - v \|_2 \leq \eps$, and so we have
\begin{align*}
|\langle x - y, u \rangle| &\leq | \langle x - y, v \rangle| + | \langle x - y, u - v \rangle| \\
&\leq 2 \beta + \rho \| x - y \|_2 \; .
\end{align*}
Taking the supremum over all unit vectors $u$ and simplifying yields that $\| x - y \|_2 \leq 2 \beta / (1 - \rho)$, as claimed.
\end{proof}

The second property shows that we may find an $\alpha$-projection oracle for $\cc$ efficiently.

\begin{claim}
Fix $\rho' > 0$.
There is a $\rho'$-projection oracle $\textsc{ProjOracle}(y, \rho', \cc)$ for $\cc$ 
which runs in time $\poly ((1 / \rho)^{O(d)}, \log (\gamma \eps / (1 - \rho)), \log (1 / \rho'))$.
\end{claim}
\begin{proof}
The projection problem may be stated as 
\begin{align*}
\min \| x - y \|_2 &~\mbox{s.t.} \langle v, y \rangle \in [b_v - \beta, b_v + \beta], ~\forall v \in \cF \; .
\end{align*}
This is a convex minimization problem with linear constraints.
By the classical theory of optimization \cite{GLS:88}, finding a $\rho$-approximate $y$ can be done in $\poly (d, \log (\mathrm{vol}(\cc) / \rho' ))$ queries to a separation oracle for $\cc$.
Since the separation oracle must only consider the constraints induced by $\cF$, this can be done in time $(1 / \rho)^{O(d)}$.
Since by Claim \ref{claim:radius} we have $\mathrm{vol}(\cc) \leq (2 \beta / (1 - \rho))^{O(d)}$, the desired runtime follows immediately.
\end{proof}

We now finally describe \textsc{LearnMeanLowD}.
Using convex optimization, we first find an arbitrary $x \in \cc$.
By Lemma \ref{lem:median} we know $\mu \in \cc$ and so this step succeeds.
After constructing $\cc$, we run $\textsc{Circumscribe}$ with appropriate parameters, and return the outputted point.
The formal pseudocode for \textsc{LearnMeanLowD} is given in Algorithm \ref{alg:learn-mean-low-d}.

\begin{algorithm}[htb]
\begin{algorithmic}[1]
\Function{LearnMeanLowD}{$\eps, \delta, S, \rho$}
\State Form a $\rho$-net $\cF$ of $B(0, 1)$ of size $(1 / \rho)^{O(d)}$ as in Claim \ref{claim:eps-net}.
\For {each $v \in \cF$}
	\State Let $b_v$ be the median of $S$ projected onto $v$.
\EndFor
\State Form $\cc$ as in Equation (\ref{eq:low-dim}).
\State Find an $x \in \cc$ using convex optimization.
\State Let $\beta = \sqrt{\frac{\pi}{2}} \eps +  O\left( \frac{\gamma \eps}{d}\right) + o(\eps)$
\State Let $R =  \beta / (1 - \rho)$
\State Let $\co (\cdot) =  \textsc{ProjOracle}(\cdot, (\rho R) / 3, \cc)$
\State \textbf{return} the output of $\textsc{Circumscribe} (R, \rho, \co, x)$
\EndFunction
\end{algorithmic}
\caption{Finding a circumscribing ball of small radius}
\label{alg:learn-mean-low-d}
\end{algorithm}

\begin{proof}
The runtime claim follows from the runtime claims for \textsc{Circumscribe} and \textsc{ProjOracle}.
Thus, it suffices to prove correctness of this algorithm.
By Lemma \ref{lem:median}, we know that $\mu \in \cc$.
By Claim \ref{claim:radius} and Corollary \ref{cor:circumscribe2}, the output $y$ satisfies $B(y, \sqrt{2} \frac{1 + 2 \rho}{1 - \rho} \beta)$.
Thus, we have $\| \mu - y \|_2 \leq   \sqrt{2} \frac{1 + 2 \rho}{1 - \rho} \beta$, as claimed.
\end{proof}

\section{Robustly Learning the Mean in High Dimensions}
\label{sec:mean}
In this section, we prove the following theorem, which is our first main result:

\begin{theorem}
\label{thm:mean-main}
Fix $\eps, \gamma, \delta > 0$, and let $X_1, \ldots, X_n$ be an $\eps$-corrupted set of points from $\normal (\mu, I)$, where $\| \mu \|_2 \leq O(\eps \log 1 / \eps)$, and where 
\[
n = \Omega\left( \frac{(d\log(d/\gamma \ve\d))^6}{\gamma^2 \ve^2} \right) \; .
\]
Then, for every $\alpha, \b >0$, there is an algorithm $\textsc{RecoverMean} (X_1, \ldots, X_n, \eps, \delta, \gamma, \a, \b)$ which 
runs in time $\poly(d, 1/\g, 1/\ve^\b,1/\a, \log 1/\d)$ and
outputs a $\muhat$ so that with probability $1 - \delta$, we have 
$\| \muhat - \mu \|_2 \leq \left(\frac{\sqrt{\p} + O(\gamma)}{1-\a} + \frac{1}{\sqrt{\b}}\right) \eps$.
\end{theorem}

\noindent In particular, observe that Theorem \ref{thm:mean-main}, in conjunction with Lemma \ref{lem:mubound}, gives us Theorem \ref{thm:mean-intro}, if we set $\gamma = o(1)$.  With this, we may state our primary algorithmic contribution:

\begin{theorem}\label{optMeanThm}
Fix $\eps, \gamma, \alpha, \delta, \b > 0$, and let $S_0 = (G_0, E_0)$ be an $\eps$-corrupted 
set of samples of size $n$ from $\normal (\mu, I)$, where $\| \mu \|_2 \leq O(\eps \log 1 / \eps)$, and where $n = \poly (d, 1 / (\gamma \eps), \log 1 / \delta)$.
Suppose that $G_0$ is $(\gamma \eps, \d)$-good with respect to $\normal (\mu, I)$.
Let $S \subseteq S_0$ be a set so that $\Delta (S, G_0) \leq \ve$.
Then, there exists an algorithm \textsc{FilterMeanOpt} that given $S, \eps, \gamma, \a, \b$ outputs one of two possible outcomes:
\begin{enumerate}[(i)]
\item
A $\muhat$, so that $\| \muhat - \mu \|_2 \leq  \left(\frac{\sqrt{\p} + O(\gamma)}{1-\a} + \frac{1}{\sqrt{\b}}\right) \eps$.
\item
A set $S' \subset S$ so that $\Delta (S', G_0) < \Delta (S, G_0)$.
\end{enumerate}
Moreover, $\textsc{FilterMeanOpt}$ runs in time $\poly (d, 1 /\g, 1/\eps^\b,  1 / \alpha, \log 1 / \delta)$.
\end{theorem}

By first running the algorithm of \cite{DKKLMS} to obtain an estimate of the mean to error $O(\eps \sqrt{\log 1 / \eps})$, then running \textsc{FilterMeanOpt} at most polynomially many times, we clearly recover the guarantee in Theorem \ref{thm:mean-main}.
Thus, the rest of the section is dedicated to the proof of Theorem \ref{optMeanThm}.

At a high level, the structure of the argument is as follows:
We first show that if there is a subspace of eigenvectors of dimension 
at least $O(\log 1 / \eps)$ of the empirical covariance matrix with large associated eigenvalues, 
then we can produce a filter using a degree-2 polynomial (Section \ref{sec:mean-manyEig}).
Otherwise, we know that there are at most  $O(\log 1 / \eps)$ eigenvectors of the empirical covariance with a large eigenvalue.
We can learn the mean in this small dimensional subspace using our learning algorithm from the previous section, and then we can argue that the empirical mean on the remaining subspace is close to the true mean (Section \ref{sec:mean-fewEig}).

This outline largely follows the structure of the filter arguments given in \cite{DKKLMS}, however, the filtering algorithm we use here requires a couple of crucial new ideas.
First, to produce the filter, instead of using a generic degree-2 polynomial over this subspace, we construct an explicit, structured, degree-2 polynomial which produces such a filter.
Crucially, we can exploit the structure of this polynomial to obtain very tight tail bounds, e.g., via the Hanson-Wright inequality.
This is critical to avoid a quasi-polynomial runtime.
If instead we used arbitrary degree-$2$ polynomials in this subspace, 
it would need to be of dimension $O(\log^2 1 / \eps)$ and the low-dimensional algorithm in the second step would take quasi-polynomial time.

Second, we must be careful to throw out far fewer good points than corrupted points.
In particular, by our definition of $\Delta$ (which gives an additional logarithmic penalty to discarding good points) and our guarantee that $\Delta$ decreases, 
our filter can only afford to throw out an $\eps / \log (1 / \eps)$ fraction of good points in total, since $\Delta$ is initially $\epsilon$.
This is critical, as if we threw away an $\eps$-fraction of good points, then proving that the problem remains efficiently solvable becomes problematic.
In particular, if these points were thrown away arbitrarily, then this becomes the full additive and subtractive model, for which a statistical query lower bound prevents us from getting an $O(\ve)$-approximate answer in polynomial time~\cite{DiakonikolasKS16c}.
To avoid discarding too many good points, we exploit tight exponential tail bounds of Gaussians, and observe that by slightly increasing the threshold at which we filter away points, we decrease the fraction of good points thrown away dramatically.

\subsection{Making Progress with Many Large Eigenvalues}
\label{sec:mean-manyEig}
We now give an algorithm for the case when there are many eigenvalues which are somewhat large.
Formally, we show:
\begin{theorem}
\label{thm:meanManyEig}
Fix $\eps, \gamma, \delta, \a, \b> 0$, and let $S_0 = (G_0, E_0)$ be an $\eps$-corrupted set of samples of size $n$ from $\normal (\mu, I)$, where $\| \mu \|_2 \leq O(\eps \log 1 / \eps)$, and where $n = \poly (d, 1 / (\gamma \eps), \log 1 / \delta)$.
Suppose that $G_0$ is $(\gamma \eps, \d)$-good with respect to $\normal (\mu, I)$.
Let $S \subseteq S_0$ be a set so that $\Delta (S, G_0) \leq \ve$.
Let $\Shat$ be the sample covariance of $S$, let $\muhat$ be the sample mean of $S$, 
and let $V$ be the subspace of all eigenvectors of $\Shat - I$ with eigenvalue more than $\frac{1}{\b} \ve$.
Then, there exists an algorithm \textsc{FilterMeanManyEig} that given $S, \eps, \gamma, \d, \a, \b$ outputs one of two possible outcomes:
\begin{enumerate}
\item
If $\dim (V) \geq C_1 \b \log (1 / \ve)$, then it outputs an $S'$ so that $\Delta (S', G_0) < \Delta (S, G_0)$.
\item
Otherwise, the algorithm outputs ``OK'', and outputs an orthonormal basis for $V$.
\end{enumerate}
\end{theorem}

Our algorithm works as follows:
It finds all large eigenvalues of $\widehat \Sigma - I $, and if there are too many, 
produces an explicit degree-2 polynomial which, as we will argue, produces a valid filter.
The formal pseudocode for our algorithm is in Algorithm \ref{alg:filter-meanManyEigs}.

\begin{algorithm}[htb]
\begin{algorithmic}[1]
\Function{FilterMeanManyEig}{$S, \eps, \gamma, \delta, \a, \b$}
\State Let $C_1, C_2, C_3 > 0$ be sufficiently large constants.
\State Let $\muhat$ and $\Shat$ be the empirical mean and covariance of $S$, respectively.
\parState {Let $V$ be the subspace of $\R^d$ spanned by eigenvectors of $\Shat-I$ with eigenvalue more than $\frac{1}{\b}\ve$.}
\If{$\dim(V) \geq C_1\b\log(1/\ve)$}
\State Let $V'$ be a subspace of $V$ of dimension $C_1\b\log(1/\ve)$.
\parState {Let $\tilde \mu$ be an approximation to $\Pi_{V'}(\mu)$ with $\ell_2$-error $\frac{\sqrt{\p} + O(\gamma)}{1 - \a}\ve$, computed using $\textsc{LearnMeanLowD}(V, \gamma, \eps, \delta, S, \gamma)$.}
\State Let $p(x)$ be the quadratic polynomial
$
p(x) = \|\Pi_{V'}(x)-\tilde \mu\|_2^2 - \dim(V').
$
\State Find a value $T>0$ so that either: 
\begin{itemize}
\item[(a)] $T>C_2 d\log(|S| / \delta)$ and $p(x)>T$ for at least one $x\in S$, or 
\item[(b)] $T>2 C_3\log(1/\ve) / c_0$ and $\Pr_{S}(p(x)>T) > \exp(-c_0 T/(2 C_3)) + \g\ve/(d\log(|S| / \delta)).$
\end{itemize}
\State \textbf{return} $S' = \{x\in S: p(x)\leq T\}$
\Else
\State {\bf return} an orthonormal basis for $V$.
\EndIf
\EndFunction
\end{algorithmic}
\caption{Filter if there are many large eigenvalues of the covariance}
\label{alg:filter-meanManyEigs}
\end{algorithm}

For clarity of exposition, we defer the proof of Theorem \ref{thm:meanManyEig} to Appendix \ref{app:meanManyEig}.

\subsection{Returning an Estimate When There are Few Large Eigenvalues}
\label{sec:mean-fewEig}
At this point, we have run the filter of Algorithm \ref{alg:filter-meanManyEigs} until there are few large eigenvalues.
In the subspace with large eigenvalues, we again run the low dimensional algorithm to obtain an estimate for the mean in this subspace.
Recall that Lemma \ref{lowDimLearningLem} guarantees the accuracy of this estimator within this subspace.
In the complement of this subspace, where the empirical covariance is very close to the identity, Lemma \ref{errorCovLem} (stated below) shows that the empirical mean is close to the true mean.
This leads to a simple algorithm which outputs an estimate for the mean, described in Algorithm \ref{alg:filter-meanFewEigs}.
\begin{algorithm}[htb]
\begin{algorithmic}[1]
\Function{FilterMeanFewEig}{$S, \eps, \gamma, \delta, \a, \b, V$}
\parState {Let $\tilde \mu_V$ be an approximation to $\Pi_{V}(\mu)$ with $\ell_2$-error $\frac{\sqrt{\p} + O(\gamma)}{1 - \a}\ve$, computed using $\textsc{LearnMeanLowD}(V, \gamma, \eps, \delta, S, \gamma)$.}
\State Let $\tilde \m_{V^\perp}$ be the empirical mean on $V^\perp$, $\Pi_{V^\perp} \hat \m$.
\State \Return $\tilde \m_V + \tilde \m_{V^\perp}$.
\EndFunction
\end{algorithmic}
\caption{Return a mean if there are few large eigenvalues of the covariance}
\label{alg:filter-meanFewEigs}
\end{algorithm}

\begin{lemma}\label{errorCovLem}
Let $\mu, \eta, G_0, S$ be as in Theorem \ref{thm:meanManyEig}.
Let $\muhat$ be the sample mean of $S$, and let $v$ be a unit vector. 
Suppose that $\langle v, \mu-\muhat \rangle > \frac{\ve}{\b^{1/2}}$. Then $\Var_{S} [ \langle v, X \rangle ] > 1+\frac{\ve}{\b}$.
\end{lemma}

\noindent For clarity of exposition, we defer the proof of Lemma \ref{errorCovLem} to Appendix \ref{app:meanManyEig}.

\subsection{The Full High-Dimensional Algorithm}
We now have almost all the pieces needed to prove the full result.
The last ingredient is the fact that, given enough samples, the good set condition is satisfied by the samples from the true distribution.
Formally,
\begin{lemma}\label{lem:GoodSamplesOpt}
Fix $\eta, \delta >0$. 
Let $X_1, \ldots, X_n$ be independent samples from $\normal (\mu, I)$, where $n = \Omega((d\log(d/\e\d))^6/\e^2)$.
Then, $S = \{X_1, \ldots, X_n\}$ is $(\eta, \d)$-good with respect to $\normal (\mu, I)$ with probability at least $1 - \delta$.
\end{lemma}
\begin{proof}
This follows from Lemmas 8.3 and 8.16 of \cite{DKKLMS}.
\end{proof}

At this point, we conclude with the proof of Theorem \ref{thm:mean-main}.
Within the subspace $V$, Lemma \ref{lowDimLearningLem} guarantees that the mean is accurate up to $\ell_2$-error $\frac{\sqrt{\p} + O(\gamma)}{1-\a}\ve$.
Within the subspace $V^{\perp}$, the contrapositive of the statement of Lemma \ref{errorCovLem} guarantees the mean is accurate up to $\ell_2$-error $\frac{\ve}{\b^{1/2}}$.
The desired result follows from the Pythagorean theorem.

\subsection{An Extension, with Small Spectral Noise}
\label{sec:inexact-cov}
For learning of arbitrary Gaussians, we will need a simple extension that allows 
us to learn the mean even in the presence of some spectral norm error in the covariance matrix.
Since the algorithms and proofs are almost identical to the techniques above, we omit them for conciseness.
Formally, we require:
\begin{theorem}
\label{thm:mean-main2}
Fix $\chi, \eps, \delta > 0$, and let $X_1, \ldots, X_n$ be an $\eps$-corrupted set of points from $\normal (\mu, \Sigma)$, where $\| \Sigma - I \|_2 \leq O(\chi)$, $\| \mu \|_2 \leq O(\eps \log 1 / \eps)$, and where $n = \poly (d, 1 / \chi, 1 / \eps, \log 1 / \delta)$.
For any $\gamma > 0$, there is an algorithm $\textsc{RecoverMeanNoisy} (X_1, \ldots, X_n, \eps, \delta, \gamma, \chi)$ 
which runs in time $\poly (d, 1 / \chi, 1 / \eps, \log 1 / \delta)$ and
outputs a $\muhat$ so that with probability $1 - \delta$, we have $\| \muhat - \mu \|_2 \leq (C + \gamma) \eps + O(\chi)$.
\end{theorem}
This extension follows from two elementary observations:
\begin{enumerate}
\item
For the learning in low dimensions, observe that the median is naturally robust to error in the covariance, and in general, by the same calculation we did, the error of the median becomes $O(\eps + \alpha)$.
\item
For the filter, observe that we only need concentration of squares of linear functions, and whatever error we have in this concentration goes directly into our error guarantee.
Thus, by the same calculations that we had above, if we filtered for eigenvalues above $1 + O(\eps + \alpha)$, we would immediately get the desired bound.
\end{enumerate}

\section{Robustly Estimating the Mean of Degree Two Polynomials}
\label{sec:covLowD}

In this section, we give robust estimates of $\E [p^2 (X)]$ for degree-$2$ polynomials $p$ in subspaces of small dimension, which is an important prerequisite to learning the covariance in high-dimensions. A crucial ingredient in our algorithm is the following improvement theorem (stated and proved in the next section) which shows how to take any weak high-dimensional estimate for the covariance and use it to get an even better robust estimate for $\E [p^2 (X)]$. 

\subsection{Additional Preliminaries}
Here we give some additional preliminaries we require for the low-dimensional learning algorithm we present here.
We will need the following well-known tail bound for degree-$2$ polynomials:
\begin{lemma}[Hanson-Wright Inequality \cite{LM00, Vershynin}]
\label{lem:hanson-wright}
Let $X \sim \normal (0, I) \in \mathbb{R}^d$ and $A$ be a $d \times d$ matrix.
Then for some absolute constant $c_0$, for every $t \geq 0$,
$$\Pr\left(\left|X^TAX - \E[X^TAX]\right| > t\right) \leq 2 \exp\left(-c_0 \cdot \min \left(\frac{t^2}{\|A\|_F^2}, \frac{t}{\|A\|_2}\right)\right).$$
\end{lemma}

We will also require the following lemmata:
\begin{lemma}[H\"{o}lder's inequality for Schatten norms]
Let $A, B$ be matrices.
Then, for all $p, q$ so that $\frac{1}{p} + \frac{1}{q} = 1$, we have $\| A B \|_{S^1} \leq \| A \|_{S^p} \| B \|_{S^q}$.
\end{lemma}

\noindent This implies the following corollary:
\begin{corollary}
\label{cor:schatten1}
Let $\Sigma, \Shat, M$ be so that $\| \Sigma - \Shat \|_F \leq O(\delta)$, and so that $\| M \|_F = 1$.
Then, we have $\| \Sigma^{1/2} M \Sigma^{1/2} - \Shat^{1/2} M \Shat^{1/2} \|_{S^1} \leq 5 \delta$.
\end{corollary}
\begin{proof}
We have
\begin{align*}
\| \Sigma^{1/2} M \Sigma^{1/2} - \Shat^{1/2} M \Shat^{1/2} \|_{S^1} &\leq \| \Sigma^{1/2} M \Sigma^{1/2} -  M  \|_{S^1} + \| M- \Shat^{1/2} M \Shat^{1/2} \|_{S^1} \; .
\end{align*}
We will bound the first term on the RHS by $5 \delta / 2$; the second term is bounded symmetrically.
We have
\begin{align*}
\| \Sigma^{1/2} M \Sigma^{1/2} -  M  \|_{S^1} &\leq \| \Sigma^{1/2} M \Sigma^{1/2} -  M  \Sigma^{1/2} \|_{S^1} + \|  M \Sigma^{1/2} -  M  \|_{S^1} \\
& \| \Sigma^{1/2} - I \|_{S^2} \| M \Sigma^{1/2} \|_{S^2} + \| \Sigma^{1/2} - I \|_{S^2} \| M \|_{S^2} \\
&\leq 5 \delta / 2 \; ,
\end{align*}
where the last line follows from H\"{o}lder's inequality for Schatten norms.
\end{proof}

\subsection{An Improvement Theorem}

Here we state and prove one of the main technical ingredients in our algorithm for robustly learning the covariance. 
\begin{theorem}
Fix $\eps, \delta, \tau > 0$. 
Let $\Sigma$ be so that $\| \Sigma - I \|_F \leq O(\eps \log 1 / \eps)$, and fix a $p \in \cP_2$, 
where $\cP_2$ denotes the set of even degree-$2$ polynomials in $d$ variables.
Let $G_0$ be an $( \eps, \delta)$-good set of samples from $\normal (0, \Sigma)$, and let
$S = \{ X_1, \ldots, X_n \}$ be so that $\Delta (S, G_0) \leq \eps$.
Then, for any $C > 0$ there is an algorithm \textsc{LearnMeanChiSquared} which, given $p, X_1, \ldots, X_n$, and $\eps$, outputs a $\muhat$ so that with probability $1 - \tau$ over the randomness of the algorithm,
\[
\left| \muhat - \E_{X \sim \normal (0,\Sigma)} [p(X)] \right| \leq \|\Sigma - I\|_F/C  + O(\log(C) \eps) \; .
\]
Moreover, the algorithm runs in time $O(|S| + \log (1 / \tau) / \eps^2 )$.
\end{theorem}

\noindent The way to think about how this result fits into the overall strategy is that 
robustly estimating the covariance is equivalent to robustly estimating the mean of every (normalized) degree-two polynomial $p$. 
The above theorem shows how a weak estimate in high-dimensions can be used to obtain stronger estimates in one dimension, 
which ultimately we will use to improve the high-dimensional estimate as well. The above theorem is the {\em workhorse} in our proof. 

Our algorithm itself is simple, however, its correctness is quite non-trivial.
We define some threshold $T$.
Given our corrupted set of samples from $\normal (0, \Sigma)$, 
we use our corrupted data set to estimate the mean of $p(X)$ conditioned on the event that $|p(X)| \leq T$.
Then, to estimate the contribution of the mean from points $X$ so that $|p(X)|  > T$, we estimate this by $\E_{X \sim \normal (0, I)} [p(X) 1_{|p(X)| > T}] $.
In other words, we are replacing the contribution of the true tail by an estimate of the contribution of $p(X)$ when $X \sim \normal (0, I)$ on this tail.
The formal pseudocode is given in Algorithm \ref{alg:LearnMeanChiSquared}.

Intuitively, this algorithm works because of two reasons.
First, it is not hard to show that the influence of points $p(X)$ within the threshold $T$ on the estimator are bounded by at most $T$.
Hence, the adversary cannot add corrupted points within this threshold and cause our estimator to deviate too much.
Secondly, because we know that $\| \Sigma - I \|_F$ is small, 
by carefully utilizing smoothness properties of sums of chi-squared random variables, 
we are able to show that our estimate for the contribution of the tail is not too large.
At a high level, this is because ``most'' of the distance between two chi-squared random variables must remain close to the means, so the difference in the tails is much smaller.
Proving that this holds in a formal sense is the majority of the technical work of this section.

\begin{proof} We know the distribution of $p(X')$ for $X' \sim N(0,I)$  explicitly and wish to use this to get a better estimate for the mean of $p(X)$ for $X \sim N(0,\Sigma)$ than might be given by the mean of the $\eps$-corrupted set of samples.

\begin{algorithm}[htb]
\begin{algorithmic}[1]
\Function{LearnMeanChiSquared}{$X_1,\dots, X_n, p(x), \eps,\tau $}
\State Let $T=O(\log C)$.
\State Let $f(x) = \begin{cases} x-T, & \text{for } x \geq T\\
   \quad 0, & \text{for } |x| \leq T \\
    x+T, & \text{for } x \leq -T
	\end{cases}$.
\State Compute $\alpha=\sum_{i=1}^n (p(X_i) - f(p(X_i)))/n$.
\State Simulate $m=O((\ln \tau)/\eps^2)$ samples $X'_1,\ldots, X'_m$ from $X' \sim N(0,I)$. 
\State Return $\muhat=\alpha + \sum_{i=1}^m f(p(X_i'))/n$.
\EndFunction
\end{algorithmic}
\caption{Approximating $\E[p(X)]$ for $X \sim N(0,\Sigma)$ with corrupted samples.}
\label{alg:LearnMeanChiSquared}
\end{algorithm}

It follows from $( \eps, \delta)$-goodness that 
$|\Pr_{X' \sim \normal(0,\Sigma)}[p(X')>t] - \#\{X_i:p(X_i)>t\}/n| \leq 2\eps$ for all $t$. 
We need to express the expectation in terms that we can use this to bound.
For $Z=p(X')$, we have that
\begin{align*}
\E[Z-f(Z)] & = \E[\max\{T,\min\{-T,Z\}\}] = \int_0^T  \Pr[Z > t] dt - \int_0^T \Pr[Z < -t] dt \;.
\end{align*}
Similarly, the samples have
$$\alpha=\sum_{i=1}^n (p(X_i) - f(p(X_i))/n  = \int_0^T \frac{\#\{X_i:p(X_i)>t\}}{n} dt - \int_0^t \frac{\#\{X_i:p(X_i)<-t\}}{n} dt \;.$$
Thus, we have $|\E[Z-f(Z)]-\alpha| \leq 2T \eps$.

Since $p \in \cP_2$, we have $\E[p(X')]=1$ for $X' \sim N(0,I)$. 
Thus, we have $\Var[f(p(X'))]\leq \E[f(p(X'))^2] \leq \E[p(X')^2] = 1$. 
It follows by standard concentration results that the empirical after taking $m=O(\ln(1-\tau)/\eps^2)$ samples has
$|\sum_{i=1}^m f(p(X_i)')/n - \E[p(X')]| \leq  \eps$
with probability $1-\tau$. When this holds, we have
$$\left| \muhat - \E_{X \sim \normal (0,\Sigma)} [p(X)] \right| \leq (2T+1) \eps + \left| \E_{X \sim \normal(0,I)}[f(p(X))] - \E_{X \sim \normal(0,\Sigma)}[f(p(X))] \right| \;.$$
To prove the correctness of the algorithm it remains to show that:

\begin{lemma}For any constant $C>0$, for $T=O(\log C)$, we obtain
$$\left|\E_{X \sim \normal(0,I)}[f(p(X))] - \E_{X' \sim \normal(0,\Sigma)}[f(p(X'))]\right| \leq \frac{\|\Sigma - I\|_F}{C} \;.$$
\end{lemma}
\begin{proof}

Let $M$ be the symmetric matrix with $\|M\|_F=1$ such that $p(x)=x^T M x$ for $x \in \R^d$. We can write $p(X)$ with $X \sim \normal(0,I)$ as
$$p(X) + \tr(M)=X^TMX=X^TO^T D OX= Y^T D Y=\sum_{i=1}^d a_i Y_i^2 \;,$$
where $O$ is orthogonal, $D$ is diagonal, $a_i$ are the eigenvalues of $M$, 
and $Y \sim \normal(0,I)$ hence $Y_i$ are i.i.d. from $\normal(0,1)$. 
Since $\|M\|_F=1$, here we have $\sum a_i^2=1$. If instead we express  $p(X')$ in terms of $X' \sim N(0,\sigma^2 )$, we obtain:
$$p(X')+\tr(M)=Y'^T \Sigma^{1/2} M \Sigma^{1/2} Y'= Y'^T O'^T D' O' Y'= Y^T D' Y=\sum_{i=1}^d b_i Y_i^2 \;,$$
where $O'$ is orthogonal, $D'$ is diagonal, $b_i$ are the eigenvalues of $\Sigma^{1/2} M \Sigma^{1/2}$, 
and  $Y',Y \sim \normal(0,I)$ hence $Y_i$ are i.i.d. from $\normal(0,1)$. 

By Corollary \ref{cor:schatten1}, we have that $\sum_i |b_i-a_i| \leq (5/2)\|\Sigma - I\|_F$. Now consider the random variables 
$$Z_{i,\lambda} =  -\tr(M) + \sum_{j=1}^{d}  c_j (Y_i^2-1) \;,
$$
where
\[
c_j = \begin{cases} b_j, & \text{for } j < i \\
(1-\lambda) a_i + \lambda b_i  & \text{for } j=i\\
a_j & \text{for } j > i  \end{cases} \; ,
\]
for $1 \leq i \leq d$ and $0 \leq \lambda \leq 1$. Note that $Z_{i,1}=Z_{i+1,0}$, for $1 \leq i \leq d-1$. 
Note that, to prove the lemma, it suffices to show that
$$|\E[f(Z_{1,0})] - \E[f(Z_{d,1})]| \leq \sum_i 2|b_i-a_i|/5C \; .$$
To this end, consider how $\E[f(Z_{i,\lambda})]$ varies with $\lambda$. 
We can write $Z_{i,\lambda}=Z_{-i} + Z_i$, where $Z_{-i} = - \tr(M) + \sum_{j < i} b_j Y_j^2 + \sum_{j>i} a_j Y_j^2$;
and $Z_i=c_i Y_i^2$, where $c_i =((1-\lambda) a_i + \lambda b_i)$. 
We assume for now that $c_i \neq 0$. Since only $Z_i$ depends on $\lambda$, we have
\begin{align*}
\frac{d\E[f(Z_{i,\lambda})]}{d\lambda} & = \frac{d}{d\lambda} \E[f(Z_{-i} + Z_i)] \\
& = \E_{Z_{-i}} \left[ \frac{d}{d\lambda}\E_{Z_i}[f(Z_{-i} + Z_i)] \right] \\
& = \E_{Z_{-i}} \left[ \frac{d}{d\lambda} \int_{-\infty}^\infty f(Z_{-i} + x) P_i(x) dx \right] \\
& = \E_{Z_{-i}} \left[ \int_{-\infty}^\infty f(Z_{-i} + x) dP_i(x)/d\lambda dx \right] \;,
\end{align*}
where $P_i(x)$ is the probability density function 
of the random variable $Z_i$. Standard results about the $\chi^2$ distribution give that:
\begin{fact} Let $Y_1,Y_2,Y_3 \sim \normal(0,1)$. Then the probability density function of $Y_1^2$ is
$\frac{1}{\sqrt{2 \pi x}} e^{-x/2}$, of $Y_1^2+Y_2^2$ is $\frac{1}{2}e^{-x/2}$, and of $Y_1^2+Y_2^2+Y_3^2$ is
$\frac{\sqrt{x}}{\sqrt{2 \pi}} e^{-x/2}$. \end{fact}
This gives that $P_i(x)= \frac{1}{\sqrt{2 \pi xc_i}} e^{-x/2c_i}.$
Now consider the derivative:
\begin{align*}
dP_i(x)/d\lambda & = d\frac{1}{\sqrt{2 \pi x c_i}}/d\lambda \cdot e^{-x/2c_i} + \frac{1}{\sqrt{2 \pi x c_i}} \cdot de^{-x/2c_i}/d\lambda \\
& = - \frac{1}{2\sqrt{2 \pi x }c_i^{3/2}} \cdot \frac{dc_i}{d\lambda} \cdot e^{-x/2c_i} + \frac{1}{\sqrt{2 \pi x c_i}}\cdot \frac{x}{2c_i^2} \cdot e^{-x/2c_i} \cdot  \frac{dc_i}{d\lambda}  \\
& = - \frac{b_i-a_i}{2\sqrt{2 \pi x}c_i^{3/2}} e^{-x/2c_i} +  \frac{(b_i-a_i)\sqrt{x}}{\sqrt{2 \pi} c_i^{5/2}} e^{-x/2c_i}  \\
& = ((b_i-a_i)/2c_i) (P_i(x) -  P^{(3)}_i(x)) \;,
\end{align*}
where $P^{(3)}_i(x)$ is the distribution of $Z_i+Z'_i+Z''_i$, 
where $Z'_i$ and $Z''_i$ are i.i.d. copies of $Z_i$. We thus have
\begin{align*}
\frac{d\E[f(Z_{i,\lambda})]}{d\lambda}  & = \E_{Z_{-i}}\left[ \int_{-\infty}^\infty f(Z_{-i} + x) dP_i(x)/d\lambda dx \right] \\
& =((b_i-a_i)/2c_i) \E_{Z_{-i}}[\E[f(Z_{-i}+Z_i)]-\E[f(Z_{-i} + Z_i +Z'_i+Z''_i)]] \\
& = ((b_i-a_i)/2c_i) \E[f(Z_{i,\lambda}) - f(Z_{i,\lambda} +Z'_i+Z''_i)] \\
& = ((b_i-a_i)/2c_i) \E_{Z_{i,\lambda}}[f(Z_{i,\lambda})- \E_{Z'_i,Z''_i}[f(Z_{i,\lambda} +Z'_i+Z''_i)]] \;.
\end{align*}
Since $f$ has Lipschitz constant $1$, $|f(Z_{i,\lambda})- \E_{Z'_i,Z''_i}[f(Z_{i,\lambda} +Z'_i+Z''_i]| \leq \E[Z'_i+Z''_i] = 2c_i$ whatever value $Z_{i,\lambda}$ takes.

Using the probability distribution of $Z'_i+Z''_i$, in the case $c_i > 0$, we have for all $|z| \leq T$,
\begin{align*}
\E_{Z'_i,Z''_i}[f(z+Z'_i+Z''_i)] & = \int_{T-z}^\infty e^{-x/2c_i}/2c_i \cdot (z+x-T) dx \\
& = e^{-(T-z)/2c_i} (z-T + 2c_i + (T-z)) \\
& = 2c_i e^{-(T-z)/2c_i}  \;.
\end{align*}
Since $c_i \leq \max{a_i,b_i} \leq 1 + O(\eps \log(1/\eps)) \leq 2$, for large enough $T=O(\log C)$, we have  $e^{-(T/2)/2c_i} \leq 1/C^2 T$. We assume that this holds.

For $-T \leq  z  \leq T/2$, we have $2c_i e^{-(T-z)/2c_i} \leq 2c_i/C^2T$. 
A similar argument when $c_i < 0$ gives that for $-T/2 \leq  z  \leq T$, we have $2|c_i| e^{-(T-z)/2c_i} \leq 2|c_i|/C^2T$.
Now we have enough to show that
\begin{align*}
\left| \frac{d\E[f(Z_{i,\lambda})]}{d\lambda} \right| & \leq |(b_i-a_i)/2c_i| |\E_{Z_{i,\lambda}}[f(Z_{i,\lambda})- \E_{Z'_i,Z''_i}[f(Z_{i,\lambda}) +Z'_i+Z''_i)]| \\
& \leq \Pr_{Z_{i,\lambda}} [|Z_{i,\lambda} - \tr(M)| \geq T/2] |(b_i-a_i)|  + |(b_i-a_i)|/C^2 \;.
\end{align*}
By the Hanson-Wright inequality, we have that for any $x$,
$$\Pr_{Z_{i,\lambda}} [|Z_{i,\lambda}  - \E[Z_{i,\lambda}]| \geq 2 \|c\|_2 \sqrt{x} + \|c\|_\infty x] \leq 2 e^{-x} \;,$$ 
where $c=(b_1,\dots,b_i-1,c_i,a_{i+1},\dots, a_d)$. Note that $\|c\|_2 \leq 1 + O(\|\Sigma - I\|_F) \leq 2$. 
Also, $\E[Z_{i,\lambda}]+ \tr(M)$ is the sum of coordinates of $c$, 
and so $|\E[Z_{i,\lambda}]| \leq \sum_{j=1}^i |b_j-a_j| \leq O(\|\Sigma - I\|_F)  \leq 1$. 
Putting this together, we obtain
$$\Pr_{Z_{i,\lambda}} [|Z_{i,\lambda} - \tr(M)| \geq T/2] \leq 2 \exp(-((T-1)/8)^2) \leq 1/C^2 \;.$$
Finally, we have that, assuming $c_i \neq 0$,
$$\left| \frac{d\E[f(Z_{i,\lambda})]}{d\lambda} \right|  \leq 2|(b_i-a_i)|/C^2 \;.$$

Now we need to deal with the special case $c_i=0$. 
Note that since $f$ is Lipschitz, for any $\beta \in \R$, we have $|\E[f(Z_{-i} + \beta Y_i^2) - f(Z_{-i})]|/|\beta| \leq |\E[Y_i^2]| \leq 1$. 
By considering the limit as $\beta$ tends to zero from above or below, 
we get that when $c_i=0$, the derivative still exists and 
$\left|\frac{d\E[f(Z_{i,\lambda})]}{d\lambda}\right|  \leq 1$. 
Since the limit $2|(b_i-a_i)|/C^2$ only does not apply at one point 
where $\E[f(Z_{i,\lambda})]$ is still continuous as a function of $\lambda$, this is not an issue. We still obtain that
$$|\E[f(Z_{i,0})] - \E[f(Z_{i,1})]| \leq 2|(b_i-a_i)|/C^2 \;.$$
Recalling that $Z_{i,1}=Z_{i+1,0}$ with $Z_{1,0}=p(X')$ and $Z_{d,1}=p(X)$ for $X' \sim N(0,1)$ 
and $X \sim N(0,\Sigma)$, we have by the Mean Value Theorem, that
\begin{align*}
|\E[f(p(X'))] - \E[f(p(X)]| &\leq \sum_i |\E[f(Z_{i,0})] - \E[f(Z_{i,1})]| \\
& \leq 2|(b_i-a_i)|/C^2 \leq O(\|\Sigma - I\|_F/C^2) \\
& \leq \|\Sigma - I\|_F/C \;,
\end{align*}
as required.
\end{proof}
\noindent This completes the proof of the theorem.
\end{proof}

\subsection{Working in a Low-Dimensional Space of Degree-Two Polynomials}
We now show that via similar techniques as before, we can patch our estimates 
together to find a matrix which agrees with the ground truth on all degree-two polynomials 
in a fixed subspace of low dimension.
Formally, we show:
\begin{theorem}
\label{thm:covLowDDeg2}
Fix $\eps, \tau >0$. 
Let $\Sigma$ be so that $\| \Sigma - I \|_F \leq O(\eps \log 1 / \eps)$.
Let $G_0$ be an $( \eps, \delta)$-good set of samples from $\normal (0, \Sigma)$, and let
$S = \{ X_1, \ldots, X_n \}$ be so that $\Delta (S, G_0) \leq \eps$.
Let $W_1$ be a subspace of degree-$2$ polynomials, and let $W_2$ be an orthogonal subspace of degree-$2$ polynomials, 
so that we have a $\Shat$ so that $\left| \E_{X \sim \normal (0, \Sigma)} [p(X)] - \E_{X \sim \normal (0, \Shat)} [p(X)] \right| \leq \xi$ for all $p \in W_2$.
Then there is an algorithm \textsc{LearnMeanPolyLowD} which given $\eps, S, W_1, W_2, \Shat$ runs in time $\poly(d, |S|, 2^{O(\dim (W_1))}, \log 1 / \tau)$, and returns a $\Sigma'$ so that 
\[
\left| \E_{\normal (0, \Sigma)} \left[ p(X) \right] - \E_{\normal (0, \Sigma')} [p(X)] \right| \leq 4 \left( \|\Sigma - I\|_F/C  + O(\log(C) \eps \right) + \xi \; ,
\]
for all $p \in \mathrm{span}(W_1 \cup W_2)  \cap \cP_2$, with probability $1 - \tau$.
\end{theorem}

In particular, this implies:
\begin{corollary}
\label{cor:covLowD}
Fix $\eps, \tau >0$. 
Let $\Sigma$ be so that $\| \Sigma - I \|_F \leq O(\eps \log 1 / \eps)$.
Let $G_0$ be an $( \eps, \delta)$-good set of samples from $\normal (0, \Sigma)$.
Let $S = \{ X_1, \ldots, X_n \}$ be so that $\Delta (S, G_0) \leq \eps$.
Let $V$ be a subspace of $\R^d$.
Then there is an algorithm \textsc{LearnCovLowDim} 
which given $S, \eps, \xi, \tau, V$ runs in time $\poly(|S|, 2^{O(\dim (V)^2)}, \log 1 / \tau)$ and returns a $\Sigma'$ so that 
\[
\| \Pi_V \left( \Sigma - \Sigma' \right) \Pi_V \|_2 \leq 4 \left( \|\Sigma - I\|_F/C  + O(\log(C) \eps \right) \;,
\]
with probability $1 - \tau$.
\end{corollary}
\begin{proof}
Observe that the dimension of the space of degree-$2$ polynomials $W$ in $V$ is $O(\dim(V)^2)$.
Run the algorithm in Theorem \ref{thm:covLowDDeg2} with the same parameters as before, with $W_1 = W$ and $W_2 = \emptyset$ (so that we may take $\xi = 0$), and then the guarantee of that algorithm, along with Lemma \ref{lem:poly-props}, gives our desired guarantee.
\end{proof}

We now describe the algorithm for Theorem \ref{thm:covLowDDeg2}.
Essentially, we do the same thing as we did for low-dimensional learning 
in the unknown mean case: we take a constant net over $V \cap \cP_2$, l
earn the mean over every polynomial in the net, and then find a $\Sigma'$ 
which is close in each direction to the learned mean.
Since we will not attempt to optimize the constant factor here, 
will will use a naive LP-based approach to find a point which is close to optimal.
The formal pseudocode is given in Algorithm \ref{alg:lowdim-deg2}.

\begin{algorithm}[htb]
\begin{algorithmic}[1]
\Function{LowDimCovLearning}{$S,\eps, \xi, \tau, W_1, W_2$}
\State Generate a $1/2$-cover $\mathcal{C}$ for $W_1 \cap \cP_2$.
\State Let $\tau' = 2^{-|\mathcal{C}|} \tau$
\For{$p \in \mathcal{C}$}
\State Compute $m_p = \textsc{LearnMeanChiSquared}(S, p, \eps, \tau')$.
\State Generate a linear constraint $c_p (\Sigma')$: $\left| \E_{\normal (0, \Sigma')} [p(X)] - m_p \right| \leq  \|\Sigma - I\|_F/C  + O(\log(C)) \eps$.
\EndFor
\State Generate the convex constraint that $\left| \E_{\normal (0, \Sigma')} [p(X)] - \E_{\normal (0, \Sigma')} [p(X)] \right| \leq \xi$ for all $p \in W_2$.
\State Using a convex program, \Return any matrix $\Sigma'$ which obeys $c_p (\Sigma')$ for all $p \in \mathcal{C}$.
\EndFunction
\end{algorithmic}
\caption{Filter if there are many large eigenvalues of the covariance}
\label{alg:lowdim-deg2}
\end{algorithm}

Observe that every constraint for each polynomial in $W_1$ is indeed linear in $\Sigma'$, by Lemma \ref{lem:poly-props}.
Moreover, the constraint for $W_2$ has an explicit separation oracle, 
since it induces a norm, and for any $p \in W_2$, we may explicitly compute $\E_{\normal (0, \Sigma')} [p(X)] - \E_{\normal (0, \Sigma')} [p(X)] $.
Thus, we may use separating hyperplane techniques to solve this convex program in the claimed running time.

\begin{proof}[Proof of Theorem \ref{thm:covLowDDeg2}]
Let us condition on the event that \textsc{LearnMeanChSquared} succeeds for each $p \in \mathcal{C}$.
By a union bound, this occurs with probability at least $1 - \tau$.
Thus, in each $p \in \mathcal{C}$, we have that $| m_p - \E_{X \sim \normal (0, \Sigma)} [p(X)] | \leq \beta$, where $\beta = \|\Sigma - I\|_F/C  + O(\log(C)) \eps$.
Let $\Sigma'$ be the matrix we find.
By the triangle inequality, we then have that for every $p \in \mathcal{C}$, that $|\E_{\normal (0, \Sigma')} [p(X)] - \E_{\normal (0, \Sigma)} [p(X)] | \leq 2 \beta$.
Hence, by the usual net arguments, we know that for every $p \in V \cap \cP_2$,
\[
|\E_{\normal (0, \Sigma')} [p(X)] - \E_{\normal (0, \Sigma)} [p(X)] | \leq 4 \beta \; .
\]
Moreover, by triangle inequality, for every $p \in W_2$, we have $\left| \E_{\normal (0, \Sigma')} [p(X)] - \E_{\normal (0, \Sigma')} [p(X)] \right| \leq 2 \xi$.
The result then follows from the Pythagorean theorem.
\end{proof}


\section{Robustly Learning the Covariance in High-Dimensions}
\label{sec:covHighDim}
In this section, we show how to robustly estimate the covariance of a mean-zero 
Gaussian in high-dimensions up to error $O(\eps)$. We use our low-dimensional learning algorithm from the previous section as a crucial subroutine in what follows. 

Our main algorithmic contribution is as follows:
\begin{theorem}\label{thm:optCov}
Fix $\eps, \delta > 0$, and let $S_0 = (G_0, E_0)$ be an $\eps$-corrupted 
set of samples of size $n$ from $\normal (0, \Sigma)$, where $\| \Sigma - I \|_F \leq \xi$ where $\xi = O(\eps \log 1 / \eps)$, and where $n = \poly (d, 1 / \eps, \log 1 / \delta)$.
Suppose that $G_0$ is $(\eps, \delta)$-good with respect to $\normal (0, \Sigma)$.
Let $S \subseteq S_0$ be a set so that $\Delta (S, G_0) \leq \eps$.
Then, there exists an algorithm \textsc{ImproveCov} that given $S, \xi, \eps$, fails with probability at most $\poly(\eps, 1/d, \delta)$, 
and otherwise outputs one of two possible outcomes:
\begin{enumerate}[(i)]
\item
A matrix $\Shat$, so that $\| \Shat - \Sigma \|_F \leq \| \Sigma - I \|_F / 2$.
\item
A set $S' \subset S$ so that $\Delta (S', G_0) < \Delta (S, G_0)$.
\end{enumerate}
Moreover, $\textsc{ImproveCov}$ runs in time $\poly (d, (1 / \eps)^{O(\log^4 1/\eps)}, \log 1 / \delta)$.
\end{theorem}

By first applying the algorithm in \cite{DKKLMS} to produce an initial estimate for $\Sigma$, and then iterating the above algorithm polynomially many times, this immediately yields:
\begin{corollary}
\label{cor:cov-main}
Fix $\eps, \delta > 0$, and let $G_0$ be a set of $i.i.d.$ samples from $\normal (0, \Sigma)$, where  $n = \poly (d, 1/ \eps, \log 1 / \delta)$.
Let $S$ be so that $\Delta (S, G_0) \leq \eps$.
There is an universal constant $C$ and an algorithm which outputs a $\Shat$ so that with probability $1 - \delta$, we have $\| \Shat^{-1/2} \Sigma \Shat^{-1/2} - I \|_F \leq C \eps$.
In particular, this implies that
\[
\dtv \left( \normal (0, \Sigma), \normal (0, \Shat) \right) \leq 2 C \eps \; .
\]
\item
\end{corollary}

\subsection{Technical Overview}


Our strategy for obtaining a high-dimensional estimate for the covariance based on solving low-dimensional subproblems will be substantially more challenging than it was for the unknown mean case. The natural approach is to take the $\poly \log (1 / \eps)$-dimensional subspace of degree-$2$ polynomials of largest empirical variance and construct a filter. However, this fails because, unlike in the mean case, we do not know the variance of these degree-$2$ polynomials to small error.
For the unknown mean case, because we assumed that we knew the covariance was the identity (or spectrally close to the identity), this was not an issue.
Now, the variance of our polynomials depends on the (unknown) covariance of the true Gaussian, which may be more than $O(\eps)$-far from our current estimate.
Indeed, it is not difficult to come up with counterexamples where there are many large eigenvalues of the empirical covariance matrix, but no filter can make progress.

We overcome this hurdle in several steps.
First, in Section \ref{sec:covFilterCovEig}, we show how to find a filter if there are many medium-sized eigenvalues of the empirical covariance matrix.
This will proceed roughly in the same way that the filter for the unknown mean does.
If no filter is created, then we know there are at most logarithmically large eigenvalues of the empirical covariance.
In the subspace $V \subseteq \R^d$ spanned by their eigenvectors, we can then learn the covariance to high accuracy using our low-dimensional estimator. 

Then, in Section \ref{sec:covFilterCovFourth}, we show that if we restrict to the orthogonal subspace, i.e., the subspace where the empirical covariance  matrix does not have large eigenvalues, we can indeed either produce a filter or improve our estimate of the covariance restricted to this subspace using our low-dimensional estimator. While the blueprint is similar to the filter for the unknown mean, the techniques are much more involved and subtle. 

Supposing we have not yet created a filter, we have now estimated the covariance on a poly-logarithmic dimensional subspace $V$, and on $V^\perp$.
This does not in general imply that we have learned the covariance in Frobenius norm.
In block form, if we write
$$
\Sigma = \left[ \begin{matrix} \Sigma_V & A^T \\ A & \Sigma_{V^\perp} \end{matrix}\right] \; ,
$$
where here $\R^d$ is written as $V \oplus V^\perp$, this implies we have learned $\Sigma_V$ and $\Sigma_{V^\perp}$ to high accuracy.
Thus, it remains to estimate the cross term $A$.

In Section \ref{sec:stitching}, we show, given a polylogarithmically sized subspace $V$, and a good estimate of the covariance matrix on $V$ and $V^\perp$, how to fill in the entire covariance matrix.
Roughly, we do this by randomly fixing directions in $V$, and performing rejection sampling based on the correlation in the direction in $V$, and showing that the problem reduces to one of robustly learning the mean of a Gaussian, which (conveniently) we have already solved.
These steps together yield our overall algorithm \textsc{ImproveCov}. Finally, in Section~\ref{sec:quasibarrier} we explain why there is a natural barrier that makes reducing the running time from quasi-polynomial to polynomial (in $1/\eps$) difficult. 

\subsection{Additional Preliminaries}
Here we give some additional preliminaries we will require in this Section.
\subsubsection{The Agnostic Tournament}
We also require the following classical result, which allows us to do agnostic hypothesis selection with corrupted samples 
(see e.g.,~\cite{DL:01, DDS12stoclong, DK14, DDS15}).
\begin{theorem}
Fix $\eps, \delta > 0$. 
Let $D_1, \ldots, D_k, D$ be a set of distributions where $\min_i \dtv (D_i, D) = \gamma$. Set $n = \Omega \left( \frac{\log k + \log 1 / \delta}{\eps^2} \right)$.
There is an algorithm \textsc{Tournament} which given oracles for evaluating the pdfs of $D_1, \ldots, D_k$ along with $n$ independent samples $X_1, \ldots X_n$ from $D$, outputs a $D_i$ so that $\dtv (D_i, D) \leq 3 \gamma + \eps$ with probability $1 - \delta$.
Moreover, the running time and number of oracle calls needed is at most $O(n^2 / \eps^2)$.
\end{theorem}

\begin{remark}
As a simple corollary of the agnostic tournament, observe that this allows us to do agnostic learning without knowing the precise error rate $\eps$.
Throughout the paper, we assume the algorithm knows $\eps$. 
However, if the algorithm is not given this information, and instead given an $\eta$ and asked to return something with error at most $O(\eps + \eta)$, we may simply grid over $\{ \eta, (1 + \gamma) \eta, (1 + \gamma)^2 \eta, \ldots, 1 \}$ (here $\gamma$ is some arbitrary constant that governs a tradeoff between runtime and accuracy), run our algorithm with $\eps$ set to each element in this set, and perform hypothesis selection via \textsc{Tournament}.
Then it is not hard to see that we are guaranteed to output something which has error at most $O(\eps + (1 + \gamma) \eta)$.
\end{remark}

\subsubsection{The Fourth Moment Tensor of a Gaussian}
As in \cite{DKKLMS}, it will be crucial for us to understand the behavior of the fourth moment tensor of a Gaussian.
Let $\otimes$ denote the Kronecker product on matrices.
We will make crucial use of the following definition:
\begin{definition}
For any matrix $M \in \R^{d \times d}$, let $M^\flat \in \R^{d^2}$ denote its canonical flattening into a vector in $\R^{d^2}$,
and for any vector $v \in \R^{d^2}$, let $v^\sharp$ denote the unique matrix $M \in \R^{d \times d}$ so that $M^\flat = v$.
\end{definition}

We will also require the following definition:
\begin{definition}
Let $\Ssym = \{M^\flat \in \R^{d^2}: \mbox{$M$ is symmetric}\}$.
\end{definition}

The following result was proven in \cite{DKKLMS}:
\begin{theorem}[Theorem 4.15 in \cite{DKKLMS}]
\label{thm:fourth-order}
Let $X \sim \normal (0, \Sigma)$. 
Let $M$ be the $d^2 \times d^2$ matrix given by $M = \E[ (X \otimes X) (X \otimes X)^T]$.
Then, as an operator on $\Ssym$, we have
\[ M = 2 \Sigma^{\otimes 2} +  \left(\Sigma^\flat \right) \left( \Sigma^\flat \right)^T \; .\]
\end{theorem}

\subsubsection{Polynomials in Gaussian Space}
\label{sec:poly-gaussian}
Here we review some basic facts about polynomials under Gaussian measure, 
which will be crucial for our algorithm for learning Gaussians with unknown covariance.
We equip the set of polynomials over $\R^d$ with the Gaussian inner product, 
defined by $\langle f, g \rangle = \E_{X \sim \normal (0, I)} [f(X) g(X)]$, and we let $\| f \|_2^2 = \langle f, f \rangle$.

For any symmetric $M$ with $\| M \|_F = 1$, define the degree-$2$ polynomial $p(x) = \frac{1}{\sqrt{2}} (x^T M x - \tr (M))$.
We call $p$ the \emph{polynomial associated to $M$}.
Observe that $p$ is even (i.e., has no degree-$1$ terms).
We will use the following properties of such polynomials:
\begin{lemma}
\label{lem:poly-props}
Let $M$ be symmetric, so that  $\| M \|_F = 1$.
Let $p$ be its associated polynomial.
Then, we have:
\begin{enumerate}[(i)]
\item
$\E_{X \sim \normal (0, I)} [p(X)] = 0$.
\item
More generally, for any positive definite matrix $\Sigma$, we have $\E_{X \sim \normal (0, \Sigma)}[p(X)] =  \langle M, \Sigma - I \rangle$.
\item
$\Var_{X \sim \normal (0, I)} [p(X)] = \E_{X \sim \normal (0, I)} [p^2 (X)] = \langle p, p \rangle = 1$. 
\item
More generally, for any positive definite matrix $\Sigma$, we have
\[
\E_{X \sim \normal (0, \Sigma)} [p^2 (X)] = {M^\flat}^T \Sigma^{\otimes 2} M^\flat + \frac{1}{2} \left( \langle \Sigma - I, M^\flat \rangle \right)^2 \; .
\]
\end{enumerate}
\end{lemma}
\begin{proof}
The first three properties are a straightforward calculation.
We show the last one here.
By definition, we have
\begin{align*}
\E_{X \sim \normal (0, \Sigma)} [p^2 (X)] &= \frac{1}{2} \E_{X \sim \normal (0, \Sigma)} [(X^T M X - \tr (M))^2] \\
&= \frac{1}{2} \E_{X \sim \normal (0, \Sigma)} [(X^T M X)^2 - 2 (X^T M X) \tr (M) + \tr (M)^2 ] \\
&= \frac{1}{2} \left( \E_{X \sim \normal (0, \Sigma)} [{M^\flat}^T (X \otimes X) (X \otimes X)^T M^\flat ] - 2 \langle M, \Sigma \rangle \tr (M) + \tr (M)^2 \right) \\
&\stackrel{(a)}{=} \frac{1}{2} \left( {M^\flat}^T \left( 2 \Sigma^{\otimes 2} + \Sigma^\flat {\Sigma^\flat}^T \right) M^\flat - 2 \langle M, \Sigma \rangle \tr (M) + \tr (M)^2 \right) \\
&={M^\flat}^T \Sigma^{\otimes 2} M^\flat + \frac{1}{2} \left( \langle M^\flat, \Sigma \rangle^2  - 2 \langle M, \Sigma \rangle \tr (M) + \tr (M)^2\right) \\
&={M^\flat}^T \Sigma^{\otimes 2} M^\flat + \frac{1}{2} \left( \langle \Sigma - I, M^\flat \rangle \right)^2 \;, 
\end{align*}
as claimed, where (a) follows from Theorem \ref{thm:fourth-order}.
\end{proof}
Observe that Lemma~\ref{lem:poly-props}(iv) implies that if we take the top eigenvector of the $d^2\times d^2$ matrix
\[
\Sigma^{\otimes 2} + \frac{1}{2} \left( M^\flat \right) \left( M^\flat \right)^\top 
\]
on the linear subspace 
\[
V = \{ M^\flat: \mbox{M is a symmetric $d \times d$ matrix} \} \; ,
\] 
then the associated polynomial maximizes $\E_{X \sim \normal (0, \Sigma)} [p^2(X)]$, and so we can find these polynomials efficiently.
More generally, if we take any linear subspace of degree two polynomials with associated matrix subspace $V'$, so that $V' \subseteq V$, then the top eigenvector of the same matrix restricted to $V'$ allows us to find the polynomial in this subspace which maximizes $\E_{X \sim \normal (0, \Sigma)} [p^2(X)]$ efficiently.

We have the following tail bound for degree-$2$ polynomials in Gaussian space:
We will use $\Pi_V(x)$ and $\Pi_V(S)$ to denote projection to a subspace $V$, of a point $x$ and a set of points $S$, respectively.
We will also need the following hypercontractivity theorem for low-degree polynomials in Gaussian space:
\begin{theorem}[Hypercontractivity in Gaussian space, see e.g. \cite{ODonnell14}] \label{thm:hc}
Let $p: \R^d \to \R$ be a degree $m$ polynomial, and let $q \geq 2$ be even.
Then $\E_{X \sim \normal (0, I)} [p(X)^q]^{1/q} \leq (\sqrt{q - 1})^m \| p \|_2$.
\end{theorem}

We need the following definition:
\begin{definition}
Let $\cP_k$ denote the set of even degree-$k$ polynomials over $d$ variables 
satisfying $\Var_{X \sim \normal (0, 1)} [p (X)] = 1$.
Moreover, for any subspace $W \subseteq \R^d$, let $\cP_k (W)$ 
denote the set of even polynomials over $d$ variables 
which only depend on the coordinates in $W$.
\end{definition}

Then by the arguments above, we have that for any two matrices $\Sigma, \Shat$,
\[
\| \Sigma - \Shat \|_F = \sup_{p \in \cP_2} \left( \E_{X \sim \normal (0, \Sigma)} [p(X)] - \E_{X \sim \normal (0, \Shat)} [p(X)] \right) \; .
\]
In particular, by Lemma \ref{lem:frobBound}, this implies that when $\| \Sigma - I \|_2$ is small, 
then learning a Gaussian with unknown covariance in total variation distance is equivalent to learning the expectation of every even degree-$2$ polynomial.

Theorem~\ref{thm:hc} implies the following concentration for degree-$4$ (more generally, low-degree) polynomials of Gaussians:
\begin{corollary}
\label{cor:hypercontractivity}
Let $p$ be a degree-$4$ polynomial.
Then there is some $A, C \geq 0$ so that for all $t \geq C$, we have
\[
\Pr_{\normal (0, I)} [|p (X) - \E_{\normal (0, I)} [p(X)]| \geq t \| p \|_2 ] \leq \exp \left( -A t^{1/2}\right) \; .
\]
\end{corollary}
\begin{proof}
Hypercontractivity in particular implies the following moment bound: for all $q \geq 2$ even, we have
\[
\E_{\normal (0, I)} [(p (X) - \E_{\normal (0, I)} [p(X)])^q] \leq (q - 1)^{q m / 2} \| p (X) - \E_{\normal (0, I)} [p(X)] \|_2^q \; .
\]
By a typical moment argument, and optimizing the choice of $q$, this gives the desired bound.
\end{proof}

\paragraph{Hermite polynomials}
Hermite polynomials are what arise by Gram-Schmidt orthogonalization applied with respect to this inner product.
For a vector of non-negative integers $a = (a_1, \ldots, a_d)$, we let $H_a (x) : \R^d \to \R$ be the Hermite polynomial associated with multi-index $a$.
It is well-known that the degree of $H_a$ is $|a| = \sum_{i = 1}^d a_i$, and moreover, $\langle H_a, H_b \rangle = \delta_{a, b}$.
In particular, for any $r \geq 1$, the Hermite polynomials of degree at most $r$ form an orthonormal basis with respect to the Gaussian inner product for all polynomials with degree at most $r$.

Therefore, given any polynomial $p: \R^d \to \R$ with degree $r$, we may write it uniquely as
\[
p(x) = \sum_{|a| \leq r} c_a (p) H_a (x) \mbox{, where } c_a (p) = \langle p, H_a \rangle \; .
\]
We define the $k$th harmonic component of $p$ to be
\[
p^{[k]}(x) = \sum_{|a| = k} c_a (p) H_a (x) \; ,
\]
and we say $p$ is \emph{harmonic} of degree $k$ if it equals its $k$th part.

\subsection{Working with Many Large Eigenvalues of the Second and Fourth Moment}
\label{sec:covFilterCovEig}\label{sec:covFilterCovFourth}
As in the unknown mean case, we will need a filter to detect if there are many directions of the empirical covariance which have too large an eigenvalue.
Formally, we need:
\begin{theorem}
\label{thm:many-eigenvaluesDeg2}
Fix $\eps, \delta > 0$.
Assume $\| \Sigma - I \|_F \leq \xi$, where $\xi = O(\eps \log 1 / \eps)$.
Suppose that  $G_0$ is $(\eps, \delta)$-good with respect to $\normal (0, \Sigma)$.
Let $S$ be a set so that $\Delta (S, G_0) \leq \eps$.
Let $\Shat = \E_S [X X^T]$.
Then there is an algorithm \textsc{FilterCovManyDeg2Eig} and a universal constant $C$ such that the following guarantee holds:
\begin{enumerate}
\item
If $\Shat - I$ has more than $O(\log 1 / \eps)$ eigenvalues larger than $C \xi$, then the algorithm outputs a $S'$ so that $\Delta (S', G_0) < \Delta (S, G_0)$.
\item
Otherwise, the algorithm outputs ``OK'', and outputs an orthonormal basis $v_1, \ldots, v_k$ for the subspace $V$ of vectors spanned by all eigenvectors of $\Shat - I$ with eigenvalue larger than $C \xi$.
\end{enumerate}
\end{theorem}

\noindent The filter developed here is almost identical to the one developed for unknown mean.
Thus, for conciseness we describe and prove the theorem in Appendix \ref{sec:many-eigenvaluesDeg2}.

We will also need a subroutine to enforce the condition that not only does the fourth moment tensor have spectral norm which is at most $O(\eps \log^2 1/ \eps)$ (restricted to a certain subspace of polynomials), but there can only be at most $O(\poly \log 1/ \eps)$ directions in which the eigenvalue is large.
However, the techniques here are a bit more complicated, for a number of reasons.
Intuitively, the main complication comes from the fact that we do not know what the fourth moment tensor looks like, whereas in the unknown mean case, we knew that the covariance was the identity by assumption. Our main result in this subsection is the following subroutine:
\begin{theorem}[Filtering when there are many large eigenvalues]
\label{thm:many-eigenvalues}
Fix $\eps, \delta > 0$.
Assume $\| \Sigma - I \|_F \leq \xi$, where $\xi = O(\eps \log 1 / \eps)$.
Let $C$ be the universal constant in \textsc{FilterCovManyDeg2Eig}.
Let $W \subseteq \R^d$ be a subspace, so that for all $v \in W$ with $\| v \|_2 = 1$, we have $v^T \E_S [XX^T] v \leq 1 + C \xi$.
Suppose that $G_0$ is $(\eps, \delta)$-good with respect to $\normal (0, \Sigma)$.
Let $S$ be a set so that $\Delta (S, G_0) \leq \eps$.
Let $k = O(\log^4 1/ \eps)$.
Then there is an algorithm \textsc{FilterCovManyDeg4Eig} and universal constants $C_1, C_2$ such that the following guarantee holds:
\begin{enumerate}
\item If there exist $p_1, \ldots, p_k \in \cP_2 (W)$ so that $\langle p_j, p_\ell \rangle = \delta_{j \ell}$ for all $j, \ell$, 
and so that $\E_S [p_j^2 (Y)] - 1 \geq C_1 \eps$ for all $j$, 
then the algorithm outputs an $S'$ so that $\Delta (S', G_0) < \Delta (S, G_0)$.
\item
Otherwise, the algorithm outputs ``OK'', and outputs an orthonormal basis $p_1, \ldots, p_{k'}$ for a subspace $V$ 
of degree-$2$ polynomials in $\cP_2 (W)$ with $k' \leq k$ so that for all $p \in V^\perp \cap \cP_2$, 
we have $\E_S [p^2 (X)] - 1 \leq C_2 \eps$.
\end{enumerate}
Moreover, $\textsc{FilterCovManyEig}$ runs in time $\poly (d, 1 / \eps, \log 1 / \delta)$.
\end{theorem}

Roughly, we will show that if there are many polynomials with large empirical variance, 
this implies that there is a degree-four polynomial whose value is much larger than it could be if $w$ 
were the set of uniform weights over the uncorrupted points.
Moreover, we can explicitly construct this polynomial, 
and it has a certain low-rank structure which allows us to use the concentration bounds we have previously derived.

\begin{algorithm}[htb]
\begin{algorithmic}[1]
\Function{FilterCovManyEig}{$S, \eps, \xi, \delta, W$}
\State Let $\Shat = \E_S [X X^T]$
\State Let $C_1, C_2, C_3$ be some universal constants sufficiently large
\State Let $A$ be the constant in Corollary \ref{cor:hypercontractivity}
\State Let $B$ be the constant in Claim \ref{claim:var-bound}
\State Let $m = 0$
\State Let $k = O(\log^4 1/ \eps)$
\While{there exists $p \in \cP_2 (W)$ so that $p \in V^\perp$ and $\E_S [p^2(X)] - 1 > C_1 \xi$} \label{line:frob-check1}
	\State Let $V_{m + 1} = \mathrm{span} (V_m \cup p)$
	\State Let $m \gets m + 1$
\EndWhile
\State Let $p_1, \ldots, p_{m}$ be an orthonormal basis for $V_{m}$
\If{$m \geq k$}
	\State Let $q_i = {(p_i^2)}^{[4]}$  be the $4$th harmonic component of $p_i^2$
	\State Let $r_i = p_i^2 - q_i$ be the degree-$2$ component of $p_i^2$
	\State Let $Q(x) = \sum_{i = 1}^{k} q_i$
	\State Find a $T$ so that either:
	\begin{itemize}
		\item $T>C_3 d^2 \sqrt{k} \log(|S|)$ and $p(X)>T$ for at least one $x\in S'$, OR
		\item $T>4 A^2 C_2 B \sqrt{k} \log^2 (1/ \eps)$ and
		\[
		\Pr_{X\in_u S}[Q(X)>T] > \exp(- A (T / 4 B \sqrt{k})^{1/2}) + \eps^2/(d\log(|S| / \delta))^2 \; .
		\]
	\end{itemize}
	\State {\bf return} the set $S' = \{X \in S: Q(X) \leq T \} $
\Else
	\State {\bf return} ``OK'', and output $p_1, \ldots, p_{m}$
\EndIf

\EndFunction
\end{algorithmic}
\caption{Filter if there are many large eigenvalues of the fourth moment tensor}
\label{alg:filter-cov-manyEig}
\end{algorithm}


\subsection{Stitching Together Two Subspaces}
\label{sec:stitching}
This section is dedicated to giving an algorithm which allows us to fully reconstruct the covariance matrix given that we know it up to small error on a low-dimensional subspace $V$ and on $W = V^\perp$.

\begin{theorem}
\label{thm:stitching}
Let $1>\xi>\eta>\eps>0$, and let $\tau > 0$.  Let $\Sigma$ so that $\|\Sigma-I\|_F \leq \xi$. 
Suppose that $\R^d$ is written as $V\oplus W$ for orthogonal subspaces $V$ and $W$ with $\dim(V) = O(\log(1/\epsilon))$. Suppose furthermore that
$$
\Sigma = \left[ \begin{matrix} \Sigma_V & A^T \\ A & \Sigma_W \end{matrix}\right] \; ,
$$
with $\|\Sigma_V-I_V\|_F, \|\Sigma_W-I_W\|_F = O(\eta)$.
Let $S_0 = (G_0, E_0)$ be an $\eps$-corrupted set of samples from $\normal (0, \Sigma)$, and let $S \subseteq S_0$ with $\Delta (S, G) \leq O(\eps)$ of size $\poly (d, 1 / \eta, \log 1 / \delta)$.

Then, there exists a universal constant $C_5$ and an algorithm $\textsc{Stitching}$ that given $V,W,\xi,\eta,\eps, \tau$ and $S$ runs in polynomial time and with probability at least $1 - \tau$ returns a matrix $\Sigma_0$ with $\|\Sigma_0-\Sigma\|_F = C_5 \eta+O(\xi^2)$.
\end{theorem}
In the latter, we will show the algorithm works when $\tau = 2/3$.
As usual the probability of success can be boosted by repeating it independently.\footnote{Observe the only randomness at this point is in the random choices made by the algorithm. 
Thus, one can just run this algorithm $O(\log 1 / \delta)$ times to obtain $\Sigma_0^{(1)}, \ldots, \Sigma_{0}^{(\ell)}$ and find any $ \Sigma_{0}^{(j)}$ which is $O(\eta + \xi^2)$ close to at least a $2/3 + o(1)$ fraction of the other outputs.} 
The basic idea of the proof is as follows. 
Since we already know good approximations to $\Sigma_V$ and $\Sigma_W$, it suffices to find an approximation to $A$.
In order to do this, we note that if we take a sample $x$ from $G$ conditioned on its projection to $V$ being some vector $v$, we find that the distribution over $W$ is a Gaussian with mean approximately $Av$. 
Running our algorithm for approximating the mean of a noisy-Gaussian, we can then compute the mapping $v \rightarrow Av$, which will allow us to compute $A$.

There are three main technical obstacles to this approach. The first is that we cannot condition on $x_V$ taking a particular value, 
as we will likely see no samples from $X$ with exactly that projection. Instead, what we will do is given samples from $X$ we will reject them with probabilities depending on their projections to $V$ in such a way to approximate the conditioning we require. 
The second obstacle is that the errors in $X$ may well be concentrated around some particular projection to $V$. 
Therefore, some of these conditional distributions may have a much larger percentage of errors than $\eps$. 
To circumvent this, we will show that by carefully choosing how we do our conditioning and by carefully picking 
the correct distribution over vectors $v$, that on average these errors are only $O(\eps)$. 
Finally, we need to be able to reconstruct $A$ from a collection of noisy approximations to $Av$. 
We show that this can be done by computing these approximations at a suitably large random set of $v$'s,
and finding the matrix $A$ that minimizes the average $\ell_2$ error between $Av$ and its approximation.

Our algorithm is given in Algorithm \ref{alg:stitching}:

\begin{algorithm}[htb]
\begin{algorithmic}[1]
\Function{Stitching}{$V, W, \delta, \eps, \tau, S$}
\State Given a vector $x$, let $x_V$ and $x_W$ be the projections onto $V$ and $W$, respectively. 
\parState {Let $C$ be a sufficiently large constant (where $C$ may depend on the constants in the big-$O$ terms in the guarantee that $\dim(V) = O(\log(1/\epsilon))$).}
\State Generate a set $V = \{ v_1, \ldots,  v_m \}$ of $(n / \eps)^C$ independent random samples from $\normal (0, 2 I_V)$. \label{SStep}
\For{$v\in V$} \label{fStep} 
	\State For each sample $x \in S$, add $x_W$ to a new set $T$ independently with probability 
	\[
	\exp(-\|x_V-v\|^2/2) \; .
	\]
	\parState {Treat $T$ as a collection of independent samples from a noisy Gaussian with covariance matrix $I_W+O(\eta)$.}
	\parState {Set $a_v$ equal to $0$ if $T$ did not contain enough samples for our algorithm or if $\|\tilde \mu\|_2>C \log(1/\eps)$.}
	\For{$\eps \in \{1,1/2,1/4,1/8,\ldots,\eta\}$}
		\State Let $\mutilde$ be the output of $\textsc{RecoverMeanNoisy} (T, \eps, (\eps/n)^{2C}, o(1), O(\eta))$.
	\EndFor
	\State Run \textsc{Tournament}	 with the output hypotheses.
	\State Set $a_v=\tilde  \mu$, where $\mutilde$ is the winning hypothesis.
\EndFor
\parState {Use linear programming to find the $\dim(W)\times \dim(V)$-matrix $B$ that minimizes the convex function\label{LPStep}
$
\E_{v\in_u S}[|a_v-Bv|].
$}
\State {\bf return}
$$
\Sigma_0 = \left[ \begin{matrix} I_V & 2B^T \\ 2B & I_W \end{matrix}\right].
$$

\EndFunction
\end{algorithmic}
\caption{Stitching the two subspaces together}
\label{alg:stitching}
\end{algorithm}

Before we prove Theorem \ref{thm:stitching}, we will need the following definition.
\begin{definition}
A function $f: \R^d \to R$ is a \emph{positive measure} if $f \geq 0$ and $\int f \leq 1$.
We write $| f |_1 = \int f$, and we say $f_1 \leq f_2$ if $f_1 (x) \leq f_2(x)$  pointwise.
\end{definition}

\begin{proof}[Proof of Theorem \ref{thm:stitching}]

Throughout this proof, let $G = \normal (0, \Sigma)$. It is clear that this algorithm has polynomial runtime and sample complexity. We have yet to show correctness. The first thing that we need to understand is the procedure of rejection sampling, where we reject $x$ except with probability $\exp(-\|x_V-v\|^2/2)$. Therefore, given a distribution $D$, we let the positive measure $D_v$ be what is obtained by sampling from $D$ and accepting a sample $x$ only with probability $\exp(-\|x_V-v\|^2/2)$. We need to understand the distributions $G_v$ and $X_v$.

Note that the pdf of $G$ is given by
$$
(2\pi)^{-n/2} \sqrt{\det(\Sigma^{-1})} \exp(-x\Sigma^{-1}x/2) dx.
$$
Therefore, the density of $G_v$ is
\begin{align*}
  & (2\pi)^{-n/2} \sqrt{\det(\Sigma^{-1})} \exp(-x\Sigma^{-1}x/2-\|x_V-v\|_2^2/2) dx\\
= & (2\pi)^{-n/2} \sqrt{\det(\Sigma^{-1})} \exp(-x(\Sigma^{-1}+I_V)x/2+x\cdot v-\|v\|_2^2/2) dx.
\end{align*}
Note that letting $\mu_v := (\Sigma^{-1}+I_V)^{-1}v=v/2+O(\delta \|v\|_2)$ that this equals
\begin{align*}
= & (2\pi)^{-n/2} \sqrt{\det(\Sigma^{-1})} \exp(-(x-\mu_v)(\Sigma^{-1}+I_V)(x-\mu_v)/2-\| v \|_2^2/2+\mu_v\cdot v/2) dx\\
= & (2\pi)^{-n/2} \sqrt{\det(\Sigma^{-1})} \exp(-(x-\mu_v)(\Sigma^{-1}+I_V)(x-\mu_v)/2-\| v \|_2^2(1+O(\delta))/4) dx.
\end{align*}
Note that this is a Gaussian with mean $\mu_v$ weighted by
$$
\sqrt{\frac{\det(\Sigma^{-1})}{\det(\Sigma^{-1}+I_V)}}\exp(-\| v \|_2^2(1+O(\delta))/4) = \Theta(2^{-\dim(V)/2}\exp(-\| v \|_2^2/4)) \;,
$$
so long as $\| v \|_2 \ll \delta^{-1/2}$. Therefore, if this condition holds, a random sample from $G$ is accepted by this procedure with probability $\Theta(2^{-\dim(V)}\exp(-\| v \|_2^2/4)$.

We also need to understand the fraction of samples from $X_v$ that are erroneous. 
In a slight abuse of notation, let $E$ also denote the distribution which is uniform over the points in $E$, and let $L$ be the distribution which is uniform over the points in $L$.
Therefore, we define
$$
\eps_v := 2^{n/2}\exp(\| v \|_2^2/4) (|E_v|_1 + |L_v|_1 \log(1/\eps)) \;,
$$
that is approximately the fraction of samples from $X_v$ that are errors (where subtractive errors are weighted more heavily).

Note that so long as $\| v \|_2<\sqrt{C\log(1/\eps)}$ that so long as $|L_v|_1 = o(1)$ that $T$ in Step \ref{fStep} will have sufficiently many samples with high probability. Furthermore, if this is the case, we have (assuming our mean estimation algorithm succeeds) that $\|\tilde \mu-\pi_W(\mu_v)\|_2 \leq O(\eps_v+\eta)$. 
Therefore, unless $\|\mu_v\|_2=\Omega(C\log(1/\eps))$ (which can only happen in $\| v \|_2=\Omega(C\log(1/\eps))$), 
we have with high probability that $\|a_v -\pi_W(\mu_v)\|_2 = O(\eps_v+\eta)$. 
In fact, we assume that this holds for all $v\in S$ with $\| v \|_2\ll(C\log(1/\eps))$.

In order to show that this is generally a good approximation, we need to know that $\eps_v$ is not too large on average. In particular, we show:

\begin{lemma}
If $v \sim N(0,2I_V)$, then $\E_v[\eps_v] = O(\eps)$.
\end{lemma}
\begin{proof}
Letting $F=E+\log(1/\epsilon)L$, we have that
\begin{align*}
\E_v[\eps_v] & = (2\pi)^{-\dim(V)/2} \int \exp(\| v \|_2^2/4)\exp(-\| v \|_2^2/4)dF_v dv\\
& = (2\pi)^{-\dim(V)/2} \int \exp(-\|x-v\|^2/2)dF(x)dv = \int dF  = O(\eps).
\end{align*}
\end{proof}

Note that $\pi_W(\mu_v) = Mv$, 
where $M=\pi_W(\Sigma^{-1}+I_V)^{-1}\pi_V^T$. Let $B$ be any $\dim(W)\times \dim(V)$ matrix. Note that
$$
\E_{v \sim N(0,2I_V)}[\|Bv-Mv\|_2^2] = 2\|B-M\|_F^2.
$$
Therefore, $\E_{v \sim N(0,2I_V)}[\|Bv-Mv\|_2] = O(\|B-M\|_F)$.
Since $\|Bv-Mv\|_2^2$ is a degree-$2$ polynomial in $v$, by Corollary 4 in \cite{Kane12},
$$\Pr_{v \sim N(0,2I_V)}\left[ \|Bv-Mv\|_2>\|B-M\|_F/2 \right] > 
\Pr_{v \sim N(0,2I_V)} \left[ \|Bv-Mv\|_2^2> \sqrt{\E_{v \sim N(0,2I_V)}[\|Bv-Mv\|_2^4]}/2 \right]>1/81.$$ 
Combining these we find that
$$
\E_{v \sim N(0,2I_V)}[\|Bv-Mv\|_2] = \Theta(\|B-M\|_F) \;.
$$

Next we show that our choice of $S$ derandomizes this result.
\begin{lemma}\label{S1Lem}
With high probability over the choice of $S$, we have that
$$
\E_{v\in_u S}[\|Bv-Mv\|_2] = \Theta(\|B-M\|_F).
$$
for all $\dim(W)\times \dim(V)$ matrices $B$.
\end{lemma}
\begin{proof}
This is equivalent to showing that for all matrices $U=B-M$ it holds
$$
\E_{v\in_u S}[\|Uv\|] = \Theta(\|U\|_F).
$$
By the standard scaling laws, it suffices to show this only for $U$ with $\|U\|_F=1$.

We also note that it suffices to show this only for $U$ in an $\eps$-net for all such matrices. 
This is because if $\|U-U'\|_F < \eps$ and if $U'$ satisfies the desired condition, then
$$
\E_{v\in_u S}[\|Uv\|_2] = \E_{v\in_u S}[\|U'v\|_2] +O(\E_{v\in_u S}[\|(U-U')v\|_2]) = \Theta(1) +O(\eps) \E_{v\in_u S}[\| v \|_2] \;,
$$
and with high probability $\E_{v\in_u S}[\| v \|_2] = O(\log(1/\eps)).$

Note that such nets exist with size $\exp(\poly(n/\eps))$. 
Therefore, it suffices to show that this condition holds for each such $U$ 
with probability $\exp(-(n/\eps)^{\Omega(C)})$.

As noted above, it suffices that
$$
\E_{v\in_u S}[\|Uv\|^2]=O(\E_{v\in_u S}[\|Uv\|^2]) = O(1)
$$
and that
$$
\Pr_{v\in_u S}[\|Uv\|^2>1/4] \geq \Pr_{v\in_u S}[\|Uv\|^2\geq 1/4] - 1/100.
$$
The first follows because if $S=\{v_1,\ldots,v_m\}$ then $\E_{v\in_u S}[\|Uv\|^2]$ is a degree-$2$ polynomial in the $v_i$ with mean $O(1)$ and variance $O(1/\sqrt{m})$, so by
standard concentration results, is $O(1)$ with $1-\exp(-\Omega(|S|))$ probability. The latter follows from standard concentration bounds. This completes the proof.
\end{proof}

We also note that the random choice of $S$ has another nice property:
\begin{lemma}\label{S2Lem}
With $S$ and $a_v$ as above, with high probability we have that
$$
\E_{v\in_u S}[\|a_v-Mv\|_2] = O(\eta) \;.
$$
\end{lemma}
\begin{proof}
Let $S=\{v_1,\ldots,v_m\}$. Then $\|f(v_i)-Mv_i\|_2$ are independent random variables with mean $\E[O(\eps_v+\eta)]=O(\eta)$ and variance at most
$$
\E[\|a_v\|_2^2+\|Mv\|_2^2] \leq O(\log^2(1/\eps))+\E[\| v \|_2^2] = O(\log^2(1/\eps)).
$$
The result follows by Chernoff bounds.
\end{proof}

Now by Lemmas \ref{S1Lem} and \ref{S2Lem} with high probability over the choice of $S$ in Step \ref{SStep}, and the $a_v$ in Step \ref{fStep} we have that for all $\dim(W)\times \dim(V)$-matrices $B$ that
$$
\E_{v\in_u S}[\|Bv-Mv\|_2] = \Theta(\|B-M\|_F),
$$
and have that
$$
\E_{v\in_u S}[\|a_v-Mv\|_2] = O(\eta).
$$
Combining these statements, we have that for all $\dim(W)\times \dim(V)$-matrices $B$ that
$$
\E_{v\in_u S}[\|a_v-B\|_2] = \Theta(\|B-M\|_F) + O(\eta).
$$

Note that by taking $B=M$, this quantity is $O(\eta)$. Therefore, the $B$ found in Step \ref{LPStep} satisfies
$$
O(\eps) = \E_{v\in_u S}[\|a_v-B\|_2] = \Theta(\|B-M\|_F) + O(\eta) \;,
$$
and therefore, $\|B-M\|_F = O(\eta)$.

The rest of the proof is a simple computation of the matrices involved. In particular, recall that
$$
\Sigma = \left[ \begin{matrix} \Sigma_V & A^T \\ A & \Sigma_W \end{matrix}\right] = \left[ \begin{matrix} I_V & A^T \\ A & I_W \end{matrix}\right]+O(\eta) \;,
$$
where the $O(\eta)$ denotes a matrix with Frobenius norm $O(\eta)$ and where $\|A\|_F = O(\delta)$. It is easy to see that
$$
\Sigma^{-1} = \left[ \begin{matrix} I_V & -A^T \\ -A & I_W \end{matrix}\right]+O(\eta+\delta^2).
$$
Therefore,
$$
\Sigma^{-1} +I_V = \left[ \begin{matrix} 2I_V & -A^T \\ -A & I_W \end{matrix}\right]+O(\eta+\delta^2).
$$
Hence,
$$
(\Sigma^{-1} +I_V)^{-1} = \left[ \begin{matrix} I_V/2 & A^T/2 \\ A/2 & I_W \end{matrix}\right]+O(\eta+\delta^2).
$$
Therefore, $M=A/2+O(\eta+\delta^2)$. Therefore, $A=2M+O(\eta+\delta^2) = 2B+O(\eta+\delta^2)$. And finally, we conclude that
$$
\Sigma =  \left[ \begin{matrix} I_V & 2B^T \\ 2B & I_W \end{matrix}\right]+O(\eta+\delta^2) = \Sigma_0+ O(\eta+\delta^2).
$$
\end{proof}

\subsection{The Full High-Dimensional Algorithm}
\label{sec:covFull}
We now show how to prove Theorem \ref{thm:optCov}, given the pieces we have.
We first show that given enough samples from $\normal (0, \S)$, the empirical data set without corruptions satisfies the regularity conditions in Section \ref{sec:cov-reg} with high probability.
For clarity of exposition, the proof of this lemma is deferred to Appendix \ref{sec:GoodSamplesOpt}.
\begin{lemma}\label{lem:GoodSamplesOptCov}
Fix $\eta, \delta >0$.
Let $X_1, \ldots, X_n$ be independent samples from $\normal (\mu, I)$, where $n = \poly (d, 1 / \eta, \log 1 / \delta)$.
Then, $S = \{X_1, \ldots, X_n\}$ is $(\eta, \delta)$-good with respect to $\normal (\mu, I)$ with probability at least $1 - \delta$.
\end{lemma}

Finally, we require the following guarantee, which states that if there is a degree-$2$ polynomial 
whose expectation under $S$ and the truth differs by a lot (equivalently, if the empirical covariance 
differs from the true covariance in Frobenius norm substantially), then it must also have very large variance under $S$.
\begin{lemma}
\label{lem:covBoundedOffLarge}
Fix $\eps, \delta > 0$.
Assume $\| \Sigma - I \|_F \leq \xi$, where $\xi = O(\eps \log 1 / \eps)$.
Suppose that $G_0$ is $( \eps, \delta)$-good with respect to $\normal (0, \Sigma)$, and
let $S \subseteq S_0$ be a set so that $\Delta (S, G_0) \leq \eps$.
There is some absolute constant $C_5$ so that if $p \in \cP_2$ is a polynomial so that $\left| \E_S [p (X)] - \E_{\normal (0,\Sigma)} [p(X)] \right| > C_5 \sqrt{\xi \eps} $, then $\E_S [p^2 (X)] - 1 > C_1 \xi$.
\end{lemma}
We defer the proof of this lemma to the Appendix.

We are now ready  to present the full algorithm as Algorithm \ref{alg:improveCov}.
\begin{algorithm}[htb]
\begin{algorithmic}[1]
\Function{ImproveCov}{$S, \xi, \eps, \delta$}
\State Let $C$ be the universal constant in \textsc{FilterCovManyDeg2Eig}
\State Let $\tau = \poly (\eps, 1 / d, \delta)$.
\State Run $\textsc{FilterCovManyDeg2Eig} (S, \eps, \xi)$
\If{\textsc{FilterCovManyDeg2Eig} outputs $S'$} \label{line:deg2check}
	\State \textbf{return} $S'$
\Else
	\State Let $V$ be the subspace returned by \textsc{FilterCovManyDeg2Eig}
	\State Let $W = V^\perp$.
	\State Run $\textsc{FilterCovManyDeg4Eig} (S, \eps, \xi, \delta, W)$
	\If{\textsc{FilterCovManyDeg4Eig} outputs $S'$} \label{line:deg4check}
		\State \textbf{return} $S'$
	\Else
		\State Let $U_1$ be the subspace of degree $2$ polynomials over $W$ it returns
		\State Let $U_2$ be the perpendicular subspace of degree $2$ polynomials over $W$
		\State Let $\Shat = \E_S [X X^T]$
		\State Let $\Sigma_V = \textsc{LearnCovLowDim} (S, \eps, \xi, \tau, V)$
		\State Let $\Sigma_W = \textsc{LearnMeanPolyLowD} (S, \eps, \xi, \tau, U_1, U_2, \Shat)$
		\State Take $\poly (n, 1/ \eps)$ fresh $\eps$-corrupted samples $S'$
		\State {\bf return} $\textsc{Stitching}(V, W, \Sigma_V, \Sigma_W, \xi, \eps, S')$.
	\EndIf
\EndIf
\EndFunction
\end{algorithmic}
\caption{Filter if there are many large eigenvalues of the covariance}
\label{alg:improveCov}
\end{algorithm}

\begin{proof}[Proof of Theorem \ref{thm:optCov}]
Condition on the events that neither \textsc{LearnCovLowDim} nor \textsc{Stitching} fail.
This happens with probability at least $\poly (\eps, 1 / d, \delta)$.
Observe that if we pass the ``if'' statement in Line \ref{line:deg2check}, 
then by the guarantee of \textsc{FilterCovManyDeg2Eig} this is indeed an $S'$ satisfying the desired properties.
Otherwise, by the guarantees of \textsc{FilterCovManyDeg2Eig}, we have that $W$ satisfies the conditions needed by \textsc{FilterCovManyDeg4Eig}.
Hence, if we pass the ``if'' statement in Line \ref{line:deg4check}, 
then the guarantee of \textsc{FilterCovManyDeg4Eig} this is indeed a $S'$ satisfying the desired properties.
Otherwise, by Lemma \ref{lem:covBoundedOffLarge}, we know that for all polynomials $p \in \cP_2$ over $W$ orthogonal to $U_1$, 
we have $|\E_{\normal (0, \Sigma)} - \E_S [X X^T] | \leq C_5 \sqrt{\xi \eps}$.
Thus, $\Sigma_W$ satisfies the conditions needed by \textsc{Stitching}.

By Corollary \ref{cor:covLowD}, we know that $\Sigma_V$ satisfies the conditions for \textsc{Stitching}, and so the correctness of the algorithm follows from Theorem \ref{thm:stitching}.
\end{proof}

\subsection{The Barrier at Quasi-Polynomial}
\label{sec:quasibarrier}

Here we explain why improving the running time from quasi-polynomial to polynomial in $1/\eps$ 
will likely be rather difficult. Recall that our strategy is to project the problem onto lower dimensional subproblems and stitch together the answer. We need the dimension of the subspace to be large enough that we can find a polynomial $Q$ that is itself the sum of squares of $k$ orthogonal degree two polynomials $p_i$ so that the value of $Q$ on the corrupted points is considerably larger than the value on the uncorrupted points. More precisely, if we let $S = (G, E)$ denote our corrupted set of samples then we want $\E_{E} [Q(X)]$ to be larger than $Q(X)$ for all but a $\poly (\eps)$ fraction of $X \in G$. We then remove all points $X \in S$ with large $Q(X)$ and by the properties of $Q$ we are guaranteed that we throw out mostly corrupted points. It turns out that the most aggressive we could be is removing points where $Q(X)$ is more than $\sqrt{k}$ standard deviations away from its expectation under the true Gaussian. But since $Q$ is a degree-four polynomial 
and we want $Q(X)$ to be smaller than our cutoff for all but a $\poly (\eps)$ fraction of $X \in G$, 
we are forced to choose $\sqrt{k} = \Omega (\log 1 / \eps)$, which means that we need to reduce to $k = \Omega (\log^2 1 / \eps)$ dimensional subproblems. 
Thus, if we solve low-dimensional subproblems in time exponential in the dimension, 
we naturally arrive at a quasi-polynomial running time. 
It seems that any approach for reducing the running time to polynomial would require fundamentally new ideas. 

\section{The General Algorithm}
\label{sec:full}
We now have all the tools to robustly learn the mean and covariance of an arbitrary high-dimensional Gaussian.
We first show how to reduce the problem of robustly learning the covariance of $\normal (\mu, \Sigma)$ to learning the covariance of $\normal (0, \Sigma)$, by at most doubling error, a trick previously used in \cite{DKKLMS} and \cite{LaiRV16}.
Given an $\eps$-corrupted set of samples  $X_1, \ldots, X_{2n}$ of size $2n$ from $\normal (\mu, \Sigma)$, we may let $Y_i = (X_i - X_{n + i}) / \sqrt{2}$.
Then we see that if $X_i$ and $X_{n + i}$ are uncorrupted, then $Y_i \sim \normal (0, \Sigma)$.
Moreover, at most $2 \eps n$ of the $Y_i$ can be corrupted, since there are at most $2 \eps n$ corrupted $X_i$.
Therefore, by doubling the error rate, we may assume that $\mu = 0$.
We may then apply the algorithm in Corollary \ref{cor:cov-main} to obtain a $\Shat$ so that with high probability, we have $\| \Shat^{-1/2} \Sigma \Shat^{-1/2} - I \|_F \leq O(\eps)$ with polynomially many samples, and in $\poly (d, (1 / \eps)^{O(\log^4 1/ \eps)})$ time.

We may then take an additional set of $\eps$-corrupted samples $\{ X_i', \ldots, X_n' \}$, and let $Y_i' = \Shat^{-1/2} X_i'$.
Then, by our guarantee on $\Shat$, we have that if $X_i'$ is uncorrupted, then $Y_i' \sim \normal (0, \Stilde)$ where $\| \Stilde - I \|_F \leq O(\eps)$.
We then run \textsc{RecoverMeanNoisy} with the $Y_i'$ to obtain a $\muhat$ so that $\| \muhat - \Shat^{-1/2} \mu \|_2 \leq O(\eps)$.
This guarantees that $\dtv (\normal (\muhat, \Stilde), \normal (\Shat^{-1/2} \mu, \Shat)) \leq O(\eps)$, which in turn implies that $\dtv (\normal (\widehat{\mu}, \Shat), \normal (\mu, \Sigma)) \leq O(\eps)$, as claimed.

Therefore, we have shown:
\begin{theorem}
Fix $\eps, \delta > 0$.
Given an $\eps$-corrupted set of samples $S$ from $\normal (\mu, \Sigma)$, where $n = \poly (d, 1/ \eps, \log 1 / \delta)$, there is an algorithm \textsc{RecoverGaussian} which takes as input $S, \eps, \delta$, and outputs a $\muhat, \Shat$ so that
\[
\dtv (\normal (\mu, \Sigma), \normal (\muhat, \Shat)) \leq O(\eps) \; .
\]
Moreover, the algorithm runs in time $\poly (d, (1 / \eps)^{O(\log^4 1/ \eps)}, \log 1 / \delta)$.
\end{theorem}

\bibliographystyle{alpha}
\bibliography{allrefs}
\appendix
\section{Lower Bounds on Agnostic Learning}
\label{sec:app-lb}
In this section, we prove information theoretic lower bounds for robust distribution learning.
In particular, we will prove that no algorithm (efficient or inefficient) can learn a general distribution from an $\ve$-corrupted set of samples to total variation distance less than $\frac{\ve}{1 - \ve}$.
We note that our lower bound applies in more general settings than we consider in this paper in couple of ways:
\begin{enumerate} \item Our construction is univariate;
\item Our construction works for \emph{any} pair of distributions which are $\frac{\ve}{1-\ve}$-close, and not just for Gaussian distributions; and
\item Our construction holds for Huber's $\ve$-contamination model, which is weaker than the noise model studied in this paper.
\end{enumerate}

In Huber's $\ve$-contamination model, data is drawn from a mixture distribution $(1-\ve)P + \ve Q$, where $P$ is some distribution that we wish to estimate and $Q$ is arbitrary (in particular, it might depend on $P$).
We will show that any two distributions which are $\frac{\ve}{1-\ve}$-close can be made indistinguishable under this contamination model.
\begin{lemma}
Let $p_1$ and $p_2$ be two distributions such that $\dtv(p_1,p_2) = \frac{\ve}{1-\ve}$.
Then there exist Huber $\ve$-contaminations $p_{1'}$ and $p_{2'}$ of $p_1$ and $p_2$ respectively, such that $p_{1'} = p_{2'}$.
Therefore, no algorithm can learn a distribution $p$ up to accuracy $< \frac{\ve}{1-\ve}$ in Huber's $\ve$-contamination model.
\end{lemma}
\begin{proof}
Since $\dtv(p_1, p_2) = \frac{\ve}{1-\ve}$, we can write
\begin{align*}
p_1 &= \left(1 - \frac{\ve}{1 - \ve}\right)p_c + \frac{\ve}{1-\ve}q_1, \\
p_2 &= \left(1 - \frac{\ve}{1 - \ve}\right)p_c + \frac{\ve}{1-\ve}q_2,
\end{align*}
where $p_c, q_1, q_2$ are distributions.
Let $p_{1'}$ be the Huber $\ve$-contamination of $p_1$ with $q_2$,
    $p_{2'}$ be the Huber $\ve$-contamination of $p_2$ with $q_1$.
In other words, 
\begin{align*}
p_{1'} &= \left(1 - 2\ve\right)p_c + \ve q_1 + \ve q_2, \\
p_{2'} &= \left(1 - 2\ve\right)p_c + \ve q_2 + \ve q_1, \\
\end{align*}
and thus $p_{1'} = p_{2'}$ as desired.
\end{proof}

We note that a similar lower bound holds when $p_1$ and $p_2$ are required to be Gaussians, but weakened by a factor of $2$.
This is due to the geometry of the space of Gaussian distributions, and we note that in the case where the variance is known, this is achieved by the median.

\begin{lemma}
\label{lem:gaussian-huber-LB}
No algorithm can output an estimate for the mean of a unit-variance Gaussian at accuracy $< \left(\sqrt{\frac{\pi}{2}} - o(1)\right)\ve$ with probability $> 1/2$ in Huber's $\ve$-contamination model.
Consequently, by Lemma~\ref{lem:mubound}, no algorithm can learn a Gaussian to total variation distance $< \left(\frac12 - o(1)\right)\ve$ in Huber's $\ve$-contamination model.
\end{lemma}
\begin{proof}
We will consider the distributions $p_1 = \mathcal{N}(-\a, 1)$ and $p_2 = \mathcal{N}(\a, 1)$, where $\a$ is to be specified later.
We note that if $p_1$ and $p_2$ can be $\ve$-corrupted into the same distribution, then the best estimate for the mean is to output $0$ (by symmetry), and the result holds.

We will show that $p_1$ can be $\ve$-corrupted into a distribution $f$, where $f(x) = \frac{\max\left\{p_1(x), p_2(x)\right\}}{\eta}$ and $\eta$ is a normalizing constant (the case of $p_2$ follows similarly).
In other words, $f$ can be written as $(1-\ve)p_1 + \ve q_1$ for some distribution $q_1$.
Since $q_1$ is a distribution, it is non-negative, and thus we require that $(1 - \ve)p_1(x) \leq f(x)$ for all $x \in \mathbb{R}$.
Note that $f(x) = \frac{\max\left\{p_1(x), p_2(x)\right\}}{\eta} \geq \frac{p_1(x)}{\eta}$.

We can compute $\eta$ as follows:
\begin{align*}
\eta &= \int_{-\infty}^{\infty} \max \left\{ p_1(x), p_2(x)\right\} dx = 2 \int_{0}^{\infty} p_2(x) dx \\
&= 2\left(\frac{1}{2} + \frac{1}{2}\erf\left(\frac{\a}{\sqrt{2}}\right)\right) = 1 +  \erf\left(\frac{\a}{\sqrt{2}}\right) = 1 +  \sqrt{\frac{2}{\p}}\a - O(\a^3)
\end{align*}

Thus, $$f(x) \geq \frac{p_1(x)}{\eta} = \frac{p_1(x)}{1 + \sqrt{\frac{2}{\p}}\a - O(\a^3)} \geq \left(1 - \sqrt{\frac{2}{\p}}\a  + O(\a^2)\right)p_1(x).$$
Again, since we require that $(1 - \ve)p_1(x) \leq f(x)$, it suffices that $1 - \ve \leq 1 - \sqrt{\frac{2}{\p}}\a + O(\a^2)$ and thus that $\sqrt{\frac{2}{\p}}\a - O(\a^2) \leq \ve$.
This corresponds to $\a \leq \left(\sqrt{\frac{\p}{2}} - o(1)\right) \ve$, and the proof is complete.
\end{proof}

Observe that this lower bound can be strengthened by a factor of two in the subtractive adversary model, where the adversary is able to both remove $\ve N$ samples and add $\ve N$ samples.
In this infinite sample regime, this means that a distribution can be corrupted to any distribution which is $\ve$-far in total variation distance.
By noting that $\mathcal{N}(-\a, 1)$ and $\mathcal{N}(\a, 1)$ are both $\ve$-far from $\mathcal{N}(0,1)$ for $\a = \left(\sqrt{2\p} + o(1)\right)\ve$, we obtain the following lemma.
Note that this lower bound is also achieved by the median in this model.

\begin{lemma}
No algorithm can output an estimate for the mean of a unit-variance Gaussian at accuracy $< \left(\sqrt{2\pi} + o(1)\right)\ve$ with probability $> 1/2$ in the subtractive adversary model.
Consequently, by Lemma~\ref{lem:mubound}, no algorithm can learn a Gaussian to total variation distance $< \ve$ in the subtractive adversary model.
\end{lemma}

Finally, we conclude by sketching a lower bound for mean estimation of sub-Gaussian distributions.
\begin{lemma}
No algorithm can output an estimate for the mean of a sub-Gaussian distribution at accuracy $o(\ve \log^{1/2}(1/\ve))$ with probability $> 1/2$ in the Huber's $\ve$-contamination model.
\end{lemma}
\begin{proof}
We start with the distribution $q = \mathcal{N}(0,1)$.
We construct $p_1$ by truncating the right tail of $q$ at the point $x_r = c_1\log^{1/2} (1/\ve)$ for some constant $c_1$, and rescaling the rest of the distribution appropriately.
Observe that, for an appropriate choice of $c_1$:
\begin{itemize}
\item $p_1$ is sub-Gaussian with a constant of $O\left(\frac{1}{1-\ve}\right)$;
\item The mean of $p_1$ is $-c_2\ve \log^{1/2} (1/\ve)$ for some constant $c_2$;
\item $p_1$ can be corrupted in Huber's $\ve$-contamination model to be $q$.
\end{itemize}
We can similarly consider $p_2$, which is constructed by truncating the left tail of $q$ at the point $x_l = -c_1\log^{1/2} (1/\ve)$.
Since $p_1$ and $p_2$ are indistinguishable when they are both $\ve$-corrupted to $q$, and the mean of all three distributions are separated by $\geq c_2 \ve \log^{1/2} (1/\ve)$, the lemma follows.
\end{proof}

\section{Omitted Proofs from Section \ref{sec:prelims}}
\subsection{Proof of Lemma \ref{lem:mubound}}
\begin{proof}
Observe that by rotational and translational invariance, it suffices to consider the problem when $\mu_1 = - \eps e_1 / 2$ and $\mu_2 = \eps e_1 / 2$, where $e_1$ is the first standard basis vector.
By the decomposability of TV distance, we have that the TV distance can in fact be written as a 1 dimensional integral:
\[
\dtv\left(\normal(\m_1, I), \normal(\m_2, I)\right) = \frac12 \cdot \frac{1}{\sqrt{2 \pi}} \int_{-\infty}^\infty \left| e^{- (x - \eps / 2)^2 / 2} - e^{- (x + \eps / 2)^2 / 2} \right| dx \; .
\]
The value of the function $f(x) = e^{- (x - \eps / 2)^2 / 2} - e^{- (x + \eps / 2)^2 / 2}$ is negative when $x < 0$ and positive when $x > 0$, hence this integral becomes
\begin{align*}
\dtv\left(\normal(\m_1, I), \normal(\m_2, I)\right) &= \frac{1}{\sqrt{2 \pi}} \int_{0}^\infty e^{- (x - \eps / 2)^2 / 2} - e^{- (x + \eps / 2)^2 / 2} dx \\
&= F(\eps / 2) - F(- \eps / 2) \; ,
\end{align*}
where $F(x) = \frac{1}{\sqrt{2 \pi}}\int_{-\infty}^x e^{-t^2 / 2} dt$ is the CDF of the standard normal Gaussian.
By Taylor's theorem, and since $F'' (x)$ is bounded when $x \in [-1, 1]$, we have
\begin{align*}
F(\eps / 2) - F(- \eps / 2) &= F' (-\eps / 2) \eps + O(\eps^3 ) \\
&= \frac{1}{\sqrt{2 \pi}} e^{-(\eps / 2)^2 / 2} \eps + O(\eps^3) \\
&= \left( \frac{1}{\sqrt{2 \pi}} + o(1) \right) \eps \; ,
\end{align*}
which proves the claim.
\end{proof}

\section{Omitted Proofs from Section \ref{sec:mean}}
\label{app:meanManyEig}

\subsection{Proof of Theorem \ref{thm:meanManyEig}}

First, we require the following several claims about $p(x)$ under various distributions, including the true Gaussian and from choosing points uniformly at random from the sets of interest.
Note that throughout this section we will use Lemma~\ref{lowDimLearningLem}, which implies that $\|\Pi_{V'} \m - \tilde \m\|_2 < 2\ve$ when the parameters $\a, \g$ are chosen appropriately.

\begin{claim}
\label{claim:poly-mean}
$\E_{\normal (\m, I)}[p(X)] = \|\Pi_{V'} \m - \tilde \m\|_2^2$.
\end{claim}
\begin{proof}
Letting $v_1, \dots, v_{C_1 \b \log (1/\ve)}$ be an orthonormal basis of $V'$, we have 
$$\E_{\normal (\m, I)}[p(X)] 
= \sum_{i=1}^{C_1 \b \log(1/\ve)} \E_{\normal(v_i^T\m, 1)}[\langle v_i, X - \tilde \m \rangle^2 - 1] 
= \sum_{i=1}^{C_1 \b \log(1/\ve)} (v_i^T (\m - \tilde \m))^2  = \|\Pi_{V'} \m - \tilde \m\|_2^2.$$
\end{proof}

\begin{claim}
\label{claim:poly-tail-mean}
For some absolute constant $c_0$ and all $t \geq \|\Pi_{V'} \m - \tilde \m\|_2^2$, 
$$\Pr_{\normal (\m, I)}(p(X) > t) \leq 2\exp\left(-c_0 \cdot \min \left(\frac{(t-\|\Pi_{V'} \m - \tilde \m\|_2^2)^2}{C_1 \b \log(1/\ve)}, t-\|\Pi_{V'} \m - \tilde \m\|_2^2\right)\right).$$
\end{claim}
\begin{proof}
This follows from Lemma \ref{lem:hanson-wright}, after re-centering the polynomial using Claim \ref{claim:poly-mean} and noting that the spectral norm and squared Frobenius norm of the corresponding $A$ matrix are at most $1$ and $C_1 \b \log (1/\ve)$, respectively.
\end{proof}

\begin{claim}
\label{claim:meanSBound}
$\E_S[p(X)] \geq  C_1 \ve \log (1/\ve)$.
\end{claim}
\begin{proof}
Recall that $\hat \m$ is the empirical mean of the point set.
\begin{align*}
\E_S[p(X)] 
&= \E_S[\|\P_{V'}X - \tilde \m\|_2^2] - \dim(V') \\
&\geq \E_S[\|\P_{V'}X - \P_{V'} \hat \m\|_2^2] - \dim(V') \\
&\geq \left(1 + \frac{1}{\b} \ve\right) \dim(V') - \dim(V')  =  C_1 \ve\log (1/\ve) \; .
\end{align*}
\end{proof}

\begin{claim}
\label{claim:meanG0Bound}
$\E_{G_0}[p(X)] \leq \|\m - \tilde \m\|_2^2 + O(\g)\ve$.
\end{claim}
\begin{proof}
\begin{align*}
\E_{G_0}[p(X)] &\leq \int_{0}^\infty \Pr_{G_0}[p(X) \geq t] dt \\
&\leq \int_{0}^{O(d \log (d/\ve\d))} \Pr_{G_0}[p(X) \geq t] dt \\
&\leq \int_{0}^{O(d \log (d/\ve\d))} \Pr_{\normal (\m, I)}[p(X) \geq t] dt + O(\g)\ve \\
&\leq \|\m - \tilde \m\|_2^2 + O(\g)\ve
\end{align*}
The inequalities follow from $(\g\ve,\d)$-goodness and Claim \ref{claim:poly-mean}.
\end{proof}

\begin{claim}
\label{claim:meanLBound}
$\phi(S, G_0)\E_{L}[p(X)] = (C_1\b + O(\gamma) + o(1) ) \ve $.
\end{claim}
\begin{proof}
\begin{align}
\phi(S, G_0)\E_{L}[p(X)] 
&= \int_{0}^\infty \phi(S, G_0) \Pr_L[p(X) > t] dt \nonumber \\
&= \int_{0}^{O(d \log(|G|/\d))} \phi(S, G_0) \Pr_L[p(X) > t] dt \label{eq:exp-L-mean1}\\
&= \int_{0}^{C_1\b\log(1/\ve) + 4\ve^2} \phi(S, G_0) \Pr_L[p(X) > t] dt \nonumber \\
&~~~~~+ \int_{C_1\b\log(1/\ve) + 4\ve^2}^{O(d \log(|G|/\d))} \phi(S, G_0) \Pr_L[p(X) > t] dt \nonumber \\
&\leq (C_1\b\log(1/\ve) + 4\ve^2)\phi(S, G_0) + 2\int_{C_1\b\log(1/\ve) + 4\ve^2}^{O(d \log(|G|/\d))}  \Pr_{G_0}[p(X) > t] dt \nonumber \\
&\leq (C_1\b\log(1/\ve) + 4\ve^2)\phi(S, G_0) \nonumber \\
&~~~~~+ 2\int_{C_1\b\log(1/\ve) + 4\ve^2}^{O(d \log(|G|/\d))}  \Pr_{\normal (\m, I)}[p(X) > t]  + \frac{\gamma \ve}{d \log(|G|/\d)} dt \label{eq:exp-L-mean2} \\
&\leq (C_1\b\log(1/\ve) + 4\ve^2)\phi(S, G_0) + O(\gamma \ve) + 8 \ve^{c_0 C_1\b}\label{eq:exp-L-mean3} \\
&\leq(C_1\b + O(\gamma) + o(1) ) \ve. \label{eq:exp-L-mean4}
\end{align}
(\ref{eq:exp-L-mean1}) and (\ref{eq:exp-L-mean2}) follow from $G_0$ being $(\gamma \ve, \d)$-good, (\ref{eq:exp-L-mean3}) is from Claim \ref{claim:poly-tail-mean}, and (\ref{eq:exp-L-mean4}) is because 
\[
\phi(S, G_0) \log (1/\phi(S,G_0)) \leq (1 + o(1))\ve \; .
\]
\end{proof}
This gives:

\begin{claim}
\label{claim:meanEBound}
$\psi \E_E [p(X)] \geq C_1 \eps \log 1/ \eps - (C_1\b + 1 + O(\gamma) + o(1) ) \ve$.
\end{claim}
\begin{proof}
This immediately follows from Claims \ref{claim:meanSBound}, \ref{claim:meanG0Bound}, and \ref{claim:meanLBound}.
\end{proof}

We now show that in the case that there are many large eigenvalues, there will be a $T$ satisfying the conditions of the filter.

\begin{claim}
\label{claim:meanTExists}
Suppose $\dim (V) \geq C_1 \b \log (1 / \eps)$.
Then there is a $T$ satisfying the conditions in the algorithm.
\end{claim}
\begin{proof}
Suppose not.
Then, we have
\begin{align*}
\psi \E_{E} [p(X)] &\leq \int_0^{2C_3 \log (1 / \eps) /c_0 } \psi \Pr_E [p(X) \geq t] dt + \int_{2 C_3 \log (1 / \eps)/c_0 }^\infty \psi \Pr_E [p(X) \geq t] dt \\
&\leq 2(C_3 / c_0) \psi \log 1 / \eps + \int_{2 (C_3/c_0) \log (1 / \eps)}^{C_2 d \log |S| / \delta} \psi \Pr_E [p(X) \geq t] dt \\
&\leq 2 (C_3 / c_0) \psi \log 1 / \eps + \int_{2 (C_3 / c_0) \log (1 / \eps) }^{C_2 d \log |S| / \delta} \Pr_S [p(X) \geq t] dt \\
&\leq (2 C_3 / c_0) \eps \log 1 / \ve +  \int_{2 (C_3 / c_0) \log (1 / \eps) }^{C_2 d \log |S| / \delta} \Pr_S [p(X) \geq t] dt \; .
\end{align*}
By applying the contradiction assumption, we have
\begin{align*}
\int_{2 (C_3 / c_0) \log (1 / \eps) }^{C_2 d \log |S| / \delta} \Pr_S [p(X) \geq t] dt  &\leq \int_{2 (C_3/c_0) \log (1 / \eps)}^{C_2 d \log |S| / \delta} \exp\left( -\frac{c_0 T}{2 C_3} \right) + \frac{\g\eps}{d \log |S| / \delta} dt  \\
&\leq \int_{2(C_3 / c_0) \log (1 / \eps) }^{C_2 d \log |S| / \delta}  \exp\left( -\frac{c_0 T}{2 C_3} \right) dt + C_2 \g\eps \\
&\leq \int_{2 (C_3 / c_0) \log (1 / \eps) }^{\infty} \exp\left( -\frac{c_0 T}{2 C_3} \right) dt + C_2 \g\eps \\
&\leq (C_3 / c_0) \int_{2 \log (1 / \eps)}^{\infty} \exp \left( - t \right) dt + C_2 \g\eps \\
&\leq (C_3 / c_0) \cdot O(\eps^2) + C_2 \g\eps \; ,
\end{align*}
and thus, we have
\begin{align*}
\psi \E_{E} [p(X)] &\ll C_1 \eps \log 1/ \eps - (C_1\b + 1 + O(\gamma) + o(1) ) \ve\; , 
\end{align*}
which contradicts Claim \ref{claim:meanEBound}.
\end{proof}

It now suffices to prove that if we construct a filter, then the invariant that $\Delta$ decreases is preserved.
Formally, we show:
\begin{claim}
\label{claim:meanDeltaBound}
Suppose $\dim (V) \geq C_1 \b \log (1 / \eps)$.
Let $S'$ be the set of points we return. 
Then $\Delta (S', G) < \Delta (S, G)$.
\end{claim}
\begin{proof}
Let $T$ be the threshold we pick.
If  $T>C_2 d \log(|S| / \delta)$ then the invariant is satisfied since we remove no good points, by $(\gamma \eps, \delta)$-goodness.
It suffices to show that in the other case, we remove $\log (1 / \eps)$ times many more bad points than good points.
By definition we remove at least 
\[
|S| \cdot \left( \exp\left( -\frac{c_0 T}{2 C_3} \right) + \frac{\g\ve}{d \log |S| / \delta} \right) \; .
\]
points.
On the other hand, by $(\gamma \eps, \delta)$-goodness, Claim \ref{claim:poly-tail-mean}, we know that we throw away at most
\begin{align*}
&|G| \cdot \Pr_{G} [p(X) > T] \leq |G_0| \cdot \Pr_{G_0} [p(X) > T]  \\
&~~~~\leq |G_0| \left( \exp \left( -c_0 \cdot \min \left(\frac{(T-\|\m - \tilde \m\|_2^2)^2}{C_1 \b \log(1/\ve)}, (T-\|\m - \tilde \m\|_2^2) \right) \right) + \frac{\g\ve}{2 \log (1 / \eps) (d  \log(|S| / \delta))} \right) \; .
\end{align*}
By our choice of $T$, and since $\| \mu - \mutilde \|_2 \leq O(\delta) \ll 1$, the first term is upper bounded by 
\begin{align*}
\exp \left( -c_0 \cdot \min \left(\frac{(T-\|\m - \tilde \m\|_2^2)^2}{C_1 \b \log(1/\ve)}, (T-\|\m - \tilde \m\|_2^2) \right) \right) & \leq \exp \left( -c_0 \cdot \min \left( \frac{(C_3 - 1)^2}{C_1\b} T, (C_3 - 1) T) \right) \right)
\end{align*}
and so 
\begin{align*}
\exp \left( -c_0 \cdot \min \left( \frac{(C_3 - 1)^2}{C_1\b} T, (C_3 - 1) T) \right) \right) &\leq \exp\left( -\frac{c_0 T}{C_3} \right) \\
&= \exp\left( -\frac{c_0 T}{2 C_3} \right)^2 \\
&\leq \exp\left( -\frac{c_0 T}{2 C_3} \right) \cdot \eps \\
&\ll \frac{1}{\log^2 1 / \eps} \exp\left( -\frac{c_0 T}{2 C_3} \right) \; .
\end{align*}
Since we have maintained this invariant so far, in particular, we have thrown away more bad points than good points, and so $|G_0| \leq (1 + \eps) |S|$.

Therefore, we have $\log (1 / \eps) \cdot |G| \cdot \Pr_{G} [p(X) > T]  \leq |S| \cdot \left( \exp\left( -\frac{c_0 T}{2 C_3} \right) + \frac{\g\ve}{d \log |S| / \delta} \right)$, and hence, the invariant is satisfied.
\end{proof}

Claims \ref{claim:meanTExists} and \ref{claim:meanDeltaBound} together imply the correctness of Theorem \ref{thm:meanManyEig}.

\subsection{Proof of Lemma \ref{errorCovLem}}

We require the following basic fact of good sets (see, e.g., Fact 8.6 in \cite{DKKLMS}):
\begin{fact}\label{conClaim}
Let $v \in \R^d$ be any unit vector. 
Let $G$ be $(\eta,\d)$-good with respect to $\normal (\mu, I)$.
Then for any $T>0$,
\[
\Pr_{G}(|v\cdot(X-\mu)| > T ) \leq 2 \exp(-T^2/2)+\eta/(d\log(d /\eta \delta)) \; ,
\]
and
\[
\Pr_{\normal (\m,I)}(|v\cdot(X-\mu)| > T ) \leq 2 \exp(-T^2/2).
\]
\end{fact}

\begin{proof}
Let $\langle v, \mu-\muhat \rangle = R>\frac{\ve}{\b^{1/2}}$.
Observe that by direct calculation, we have $\E_{S}[\langle v, X-\mu\rangle^2] - \E_{S}[\langle v, X-\muhat \rangle^2] \leq 3 R^2$.
Hence, it suffices to show that $\E_{S}[\langle v, X-\mu\rangle^2] \geq 1+\Omega(R^2/\ve) - \left(\gamma + O\left(\frac{1}{\b}\right)\right)\ve.$

We write $S= (G_0 \setminus L, E)$ where $|E|=\psi |S|$ and $|L|=\phi|S|$. 
First we consider the expectation of $v\cdot X$ over $X$ in $S$. 
This is
\[
\frac{\mu_0-\phi \mu_L+\psi\mu_E}{1-\phi+\psi} \; .
\]
where $\mu_0$ is the mean over $G_0$, $\mu_L$ the mean over $L$ and $\mu_E$ the mean over $E$. 
By $(\gamma \eps, \d)$-goodness, we have that $\left| \langle v, \mu_0 - \mu \rangle \right| \leq  \g \ve$.
We also have that 
\begin{align*}
\left| \langle v, \mu_L - \mu \rangle \right| &\leq \int_0^\infty \Pr_{L} \left[ |\langle v, X - \mu \rangle| \geq t \right] dt \\
&= \int_0^{O(\sqrt{d \log (d / \ve \delta)})} \Pr_{L} \left[ |\langle v, X - \mu \rangle| \geq t \right] dt
\end{align*} 
Since we have
\begin{equation}
\label{eq:L-bound}
\Pr_L \left[ | \langle v, (X - \mu) \rangle| > t \right] \leq \min \left( 1, \frac{|G_0|}{|L|} \Pr_{G_0} \left[ |\langle v, X - \mu \rangle | > t \right] \right) \; ,
\end{equation}
by Fact \ref{conClaim} and the calculations done in the proof of Corollary 8.8 in \cite{DKKLMS}, we have that $| \langle v, \mu_L - \mu \rangle | \leq O(\log |S| / |L|) \leq O(\log 1 / \phi)$.
Since by assumption we have $\langle v, \mu-\muhat \rangle = R > \frac{\ve}{\b^{1/2}}$, this implies that $\langle v, \mu_E-\mu \rangle = \Omega(R/\epsilon)$.
In particular, this implies that
\begin{align*}
\E_E \left[ \langle v, X - \mu \rangle^2 \right] &\geq \E_E \left[ \langle v, X - \mu \rangle \right]^2 \\
&\geq \Omega \left( \frac{R^2}{\ve^2} \right) \; .
\end{align*}

Next we consider the expectation of $\langle v, X-\mu \rangle^2$. 
By $(\gamma \eps, \d)$-goodness, we have $\E_{G_0} [\langle v, X - \mu \rangle^2] = 1+O(\ve)$. 
We also have that 
\begin{align*}
\phi \E_L\left[ \langle v, X - \mu \rangle^2 \right] &\leq \int_0^\infty \phi \Pr_{L} \left[ |\langle v, X - \mu \rangle|^2 \geq t \right]  dt \\
&\leq \frac{5}{\b}\log(1/\ve)\phi + \int_{(5/\b)\log(1/\ve)}^{O(d \log{d/\ve\d})} \phi \Pr_{L} \left[ |\langle v, X - \mu \rangle|^2 \geq t \right]  dt \\
&\leq \frac{5}{\b}\log(1/\ve)\phi + \int_{(5/\b)\log(1/\ve)}^{O(d \log{d/\ve\d})}  \Pr_{G_0} \left[ |\langle v, X - \mu \rangle|^2 \geq t \right]  dt \\
&= \frac{5}{\b}\log(1/\ve)\phi + \int_{(5/\b)\log(1/\ve)}^{O(d \log{d/\ve\d})}  \Pr_{\normal (\m,I)} \left[ |\langle v, X - \mu \rangle|^2 \geq t \right] dt + \g \ve  \\
&= \left(O\left(\frac{1}{\b}\right)+   \g + o(1)\right) \ve.
\end{align*} 

Since
$$
\E_{S}[\langle v,  (X-\mu)\rangle^2] =\frac{\E_{G_0}[\langle v, (X-\mu)\rangle^2]-\phi \E_{L}[\langle v, (X-\mu)\rangle^2]+\psi \E_{E}[\langle v, (X-\mu)\rangle^2]}{1-\phi+\psi},
$$
this is at least $1+\Omega(R^2/\ve) - \left(\gamma + O\left(\frac{1}{\b}\right)\right)\ve$.

\end{proof}

\section{Omitted Proofs from Section \ref{sec:covHighDim}}

\subsection{Proof of Theorem \ref{thm:many-eigenvaluesDeg2}}
\label{sec:many-eigenvaluesDeg2}
Our algorithm works as follows, just as for Algorithm \ref{alg:filter-meanManyEigs}.
It finds all large eigenvalues of $\Shat - I $, and if there are too many, produces an explicit degree-$2$ polynomial which, as we will argue, produces a valid filter.
The formal pseudocode for our algorithm is in Algorithm \ref{alg:filter-covManyDeg2Eigs}.

\begin{algorithm}[htb]
\begin{algorithmic}[1]
\Function{FilterCovManyDeg2Eig}{$S, \eps, \xi, \delta$}
\State Let  $\Shat$ be the empirical second moment of $S$, respectively.
\parState {Let $V$ be the subspace of $\R^d$ spanned by eigenvectors of $\Shat-I$ with eigenvalue more than $C \xi$.}
\If{$\dim(V) \geq C_1\log(1/\ve)$}
\State Let $V'$ be a subspace of $V$ of dimension $C_1\log(1/\ve)$.
\State Let $p(x)$ be the quadratic polynomial
$$
p(x) = \|\Pi_{V'}(x) \|_2^2 - \dim(V').
$$
\State Find a value $T>0$ so that either: 
\begin{itemize}
\item $T>C_2 d\log(|S| / \delta)$ and $p(x)>T$ for at least one $x\in S$, or 
\item $T>2 C_3\log(1/\ve) / c_0$ and $\Pr_{S}(p(x)>T) > \exp(-c_0 T/(2 C_3)) + \ve/(d\log(|S| / \delta)).$
\end{itemize}
\State \textbf{return}
\[
S' = \{x\in S: p(x)\leq T\} \; .
\]
\Else
\State {\bf return} an orthonormal basis for $V$.
\EndIf
\EndFunction
\end{algorithmic}
\caption{Filter if there are many large eigenvalues of the covariance}
\label{alg:filter-covManyDeg2Eigs}
\end{algorithm}

The proofs of the following claims are identical to the proofs of Claim \ref{claim:meanSBound}-\ref{claim:meanEBound}, by applying the corresponding property of $(\gamma \eps, \delta)$-goodness for this setting, and so we omit them.

\begin{claim}
\label{claim:covDeg2SBound}
$\E_S[p(x)] \geq C \xi.$
\end{claim}

\begin{claim}
\label{claim:covDeg2G0Bound}
$\E_{G_0}[p(X)] \leq \frac{\eps}{(d\log 1 / \eps)^2}$.
\end{claim}

\begin{claim}
\label{claim:covDeg2LBound}
$\phi(S, G_0)\E_{L}[p(X)] = (C + O(\gamma) + o(1) ) \eps $.
\end{claim}

\begin{claim}
\label{claim:covDeg2EBound}
$\psi \E_E [p(X)] \geq C \xi - (C_1 + 1 + O(\gamma) + o(1) ) \eps$.
\end{claim}

These claims imply just as before that there is a $T$ with the desired properties.
Again, the proof is identical.
Formally:
\begin{claim}
\label{claim:covDeg2TExists}
Suppose $\dim (V) \geq C \log (1 / \eps)$.
Then there is a $T$ satisfying the conditions in the algorithm.
\end{claim}
Finally, we show the invariant that $\Delta$ decreases.
This is almost identical to the proof of Claim \ref{claim:meanDeltaBound}, however, we need to slightly change our application of the Hanson-Wright inequality.
Formally, we show:
\begin{claim}
\label{claim:covDeg2DeltaBound}
Suppose $\dim (V) \geq C_1 \log 1 / \eps$.
Let $S'$ be the set of points we return. 
Then $\Delta (S', G) < \Delta (S, G)$.
\end{claim}
\begin{proof}
Let $T$ be the threshold we pick.
If  $T>C_2 d \log(|S| / \delta)$ then the invariant is satisfied since we remove no good points, by $(\gamma \eps, \delta)$-goodness.
It suffices to show that in the other case, we remove $\log 1 / \eps$ times many more bad points than good points.
By definition we remove at least 
\[
|S| \cdot \left( \exp\left( -\frac{c_0 T}{2 C_3} \right) + \frac{\ve}{d \log |S| / \delta} \right) 
\]
points.
On the other hand, observe that if $v_1, \ldots, v_k$ is an orthonormal basis for $V$, we have 
\begin{align*}
\| \Pi_V x \|_2^2 &= \sum \langle v_i, X \rangle^2 \\
&= y \Sigma^{1/2} \left( \sum v_i v_i^T \right) \Sigma^{1/2} X \; ,
\end{align*}
so that if $X \sim \normal (0, \Sigma)$, we have $Y \sim \normal (0, I)$.
Let $M = \Sigma^{1/2} \left( \sum v_i v_i^T \right) \Sigma^{1/2}$.
We have that $\| M \|_F \leq \sum_{i = 1}^n \| \Sigma \|_2 \leq (1 + \xi) k$, and $\| M \|_2 \leq \| \Sigma \|_2 \leq (1 + \xi)$.
Since 
\begin{align*}
\left| \E_{\normal (0, \Sigma)} [ \| \Pi_V x \|_2^2 ] - k \right| &= \left| \sum_{i = 1}^k v_i^T (\Sigma - I) v_i \right| \\
&\leq  \sum_{i = 1}^k \| \Sigma - I \|_2 \\
&\leq (1 + \xi) k \; ,
\end{align*}
we have by Hanson-Wright that
\[
\Pr_{\normal (0, \Sigma)} \left[ \| \Pi_V x \|_2^2 - \dim (V) > T \right] \leq \exp \left( -c_0 \cdot \min \left(\frac{(T- (1 + \xi)k )^2}{C_1 \log(1/\ve)}, (T-\|\m - (1 + \xi) k\|_2^2) \right) \right) \; ,
\]
so by our choice of $T$, we have
\begin{align*}
\exp \left( -c_0 \cdot \min \left(\frac{(T-\|\m - \tilde \m\|_2^2)^2}{C_1 \log(1/\ve)}, (T-\|\m - \tilde \m\|_2^2) \right) \right) & \leq \exp \left( -c_0 \cdot \min \left( \frac{(C_3 - 1)^2}{C_1} T, (C_3 - 1) T) \right) \right) \; .
\end{align*}
The remaining proof now proceeds identically to the proof of Claim \ref{claim:covDeg2DeltaBound}.
%
\end{proof}

\subsection{Proof of Theorem \ref{thm:many-eigenvalues}}
Clearly, the only non-trivial condition to certify for Theorem \ref{thm:many-eigenvalues} is that if we are in Case (1), the returned set satisfies the desired properties.

Our proof will roughly follow the same structure as the proof of Theorem \ref{thm:meanManyEig}.
We will first show that the empirical average of the polynomial $Q$ can only be large because of the contribution of the points in $E$ (Claim \ref{claim:bad-points}).
We will then show that this implies that there exists a threshold $T$ which the algorithm will find in this case (Claim \ref{claim:covFindT}).
Finally, we will show that for any such $T$ we find, the returned set of points will indeed satisfy $\Delta (S', G_0) < \Delta (S, G_0)$, which implies the correctness of the algorithm (Claim \ref{claim:covDeltaBound}).

We first show the following claim:
\begin{claim}
\label{claim:ri-bound}
$\E_{S} [r_i] - \E_{\normal (0, \Sigma)} [r_i] \leq (4 C + 1) \xi$
\end{claim}
\begin{proof}
Let us suppose that $p_i$ corresponds to the matrix $A_i$, given by $(A_i)_{a,b}=\nabla_a\nabla_b p_i/\sqrt{2!}$, so that the $A_i$ are orthonormal with respect to the Frobenius norm.
Then the constant harmonic part of $p_i^2$ corresponds to $\| A_i \|_F^2 \leq 1$.
The degree-$2$ harmonic part of $p_i^2$ corresponds to the matrix $2 \sqrt{2} A_i^2$.
This is because if we let $B_i$ be the matrix corresponding to $p_i^2$, we have
\begin{align*}
(B_i)_{a, b} &= \frac{1}{\sqrt{2}} \E_{X \sim \normal (0, I)} \left[ \nabla_a \nabla_b p_i^2 (X) \right] \\
&= \frac{2}{\sqrt{2}} \E_{X \sim \normal (0, I)} \left[ \nabla_b \left( p_i (X) \nabla_a p_i (X)  \right) \right] \\
&= \frac{2}{\sqrt{2}} \E_{X \sim \normal (0, I)} \left[ \nabla_a p_i(X) \nabla_b p_i (X) + p_i (X) \nabla_{a} \nabla_{b} p_i (X)  \right] \\
&= \frac{2}{\sqrt{2}} \E_{X \sim \normal (0, I)} \left[ \nabla_a p_i(X) \nabla_b p_i (X) \right] \\
&= 2 \sqrt{2} \langle (A_i)_a , (A_i)_b \rangle \ ,
\end{align*}
where the second to last line follows since $ \nabla_{a} \nabla_{b} p_i (X)$ is a constant, and the last line follows from explicit computation.
In particular, this implies that the non-constant component of $r_i$ corresponds to matrix with trace norm at most $2 \sqrt{2} \leq 4$.
Therefore, $r_i$ can be written as $r_i (x) = \sum \alpha_i \langle v_i, x \rangle^2 + C_0$ for some constant $C_0$, where $\sum |\alpha_i| \leq 4$.
Thus, by our assumption, we have $\E_S [r_i] \leq 1 + C \xi$.
The claim then follows since $\Sigma$ and $I$ are differ in Frobenius norm by at most $\xi$.
\end{proof}
\begin{claim}
\label{claim:mu-bound}$\E_S [Q(X)] - \E_{X \sim \normal (0, \Sigma)} [Q(X)] \geq (C_1 - 6) \xi k$
\end{claim}
\begin{proof}
Observe that since $\| \Sigma - I \|_F \leq \xi$, in particular we have $\| \Sigma^{\otimes 2} - I^{\otimes 2} \|_2 \leq \xi$, and hence by Lemma \ref{lem:poly-props} we have that $\left| \E_{\normal (0, \Sigma)} [p^2 (X)] - \E_{X \sim \normal (0, I)} [p^2 (X)] \right| \leq 2 \xi$ for all $p \in \cP_2$.
Hence, in particular, we have
\[
\E_S [p_i^2 (X)] - \E_{\normal (0, \Sigma)} [p_i^2 (X)] \geq (C_1 - 2) \xi \; .
\]
By Claim \ref{claim:ri-bound}, we have
\[
\E_S \left[ \sum_{i = 1}^k r_i \right] - \E_{\normal (0, \Sigma)} \left[ \sum_{i = 1}^k r_i \right] \leq (4C + 1) k \; .
\]
In particular, this implies that
\[
\E_S \left[ Q(X) \right] - \E_{\normal (0, I)} \left[Q(X) \right] \geq (C_1 - 4 C - 5) \xi k \; ,
\]
as claimed.
\end{proof}

We now show:
\begin{claim}
\label{claim:var-bound}
There is some universal constant $B$ so that
\[
\E_{X \sim \normal (0, \Sigma)} [Q^2(X)]^{1/2} \leq B \sqrt{k} \; .
\]
\end{claim}
Before we prove this, we need the following lemma:
\begin{lemma}
For any degree-$4$ polynomial $p$, and any $\Sigma$, if we let $F(y, p, \Sigma)$ denote the $y$th percentile of $p$ under $\Sigma$, then we have $F(1/4, p^2, \Sigma) , F(3/4, p^2 ,\Sigma)=\Theta(\E_{X \sim \normal (0, \Sigma)} [p^2(X)])$.
\end{lemma}
\begin{proof}
Let $\mu' = \E_{X \sim \normal (0, \Sigma)} [p^2(X)]$.
First, we note that $\Pr(p^2(X) > 4\mu')\leq 1/4,$ so $F(3/4, p^2 ,\Sigma)\leq 4 \mu'$. On the other hand, by known anti-concentration bounds~\cite{CW:01}, 
we have that $\Pr(p^2(X) \leq \eps \mu') = \Pr(|p(X)| \leq \sqrt{\eps \mu'}) = O(\eps^{1/8})$. So, for $\eps$ a sufficiently small constant, $\Pr(p^2(X) \leq \eps \mu') < 1/4$, and therefore, $F(1/4, p^2, \Sigma) \geq \eps \mu'$. Since $\eps \mu' \leq F(1/4, p^2, \Sigma) \leq F(3/4, p^2, \Sigma)\leq 4 \mu'$, this completes our proof.
\end{proof}
\begin{proof}[Proof of Claim \ref{claim:var-bound}]
Since $\dtv (\normal (0, \Sigma), \normal (0, I)) = O(\eps \log 1 / \eps) = o(1)$, we have that $F(3/4, Q^2, \Sigma) \leq F(1/4, Q^2, \Sigma)$, so by the above lemma, we have that
\[
\E_{X \sim \normal (0, \Sigma)} [Q^2(X)]^{1/2} \leq O( \E_{X \sim \normal (0, I)} [Q^2(X)]^{1/2}) \; .
\]
Hence, it suffices to bound $\| Q \|_2^2$.

Let us again suppose that $p_i$ corresponds to the matrix $A_i$, given by $(A_i)_{a,b}=\nabla_a\nabla_b p_i/\sqrt{2!}$. Note that the $A_i$ are symmetric matrices that form an orthonormal set. We note that $q_i$ is the harmonic degree-$4$ polynomial corresponding to the rank-$4$ tensor
\begin{align*}
T_{a,b,c,d} & = \nabla_a \nabla_b \nabla_c \nabla_d q_i/\sqrt{24} \\
& = \nabla_a \nabla_b \nabla_c\nabla_d (p_i^2)/\sqrt{24} \\
& = (2(\nabla_a \nabla_b p_i)(\nabla_c \nabla_d p_i) + 2(\nabla_a \nabla_c p_i)(\nabla_b \nabla_d p_i) + 2(\nabla_a \nabla_d p_i)(\nabla_b \nabla_c p_i))/\sqrt{24}\\
& = (4 (A_i)_{a,b}(A_i)_{c,d} + 4 (A_i)_{a,c} (A_i)_{b,d} + 4 (A_i)_{a,d} (A_i)_{b,c})/\sqrt{24}\\
& = \sqrt{6} \mathrm{Sym}(A_i\otimes A_i).
\end{align*}
By linearity, $Q$ corresponds to the rank-$4$ tensor $T=\sqrt{6}\sum_{i=1}^k \mathrm{Sym}(A_i\otimes A_i)$. It thus suffices to show that $\|T\|_2 = O(k)$, where here $\| \cdot \|_2$ denotes the square root of the sum of the squares of the entries of the tensor. 
In order to show this, we note that $\|T\|_2 = \sup_{\|V\|_2=1} \langle V,T \rangle.$ Therefore, it suffices to show that for all $4$-tensors $V$ with $\|V\|_2 \leq 1$, that $\langle V,T \rangle = O(\sqrt{k})$. We note that
\begin{align*}
\langle V,T \rangle & = \sqrt{6}\left\langle V,\sum_{i=1}^k \mathrm{Sym}(A_i) \otimes A_i \right\rangle  = \sqrt{6}\left\langle \mathrm{Sym}(V),\sum_{i=1}^kA_i \otimes A_i \right\rangle.
\end{align*}
Note that $\mathrm{Sym}(V)$ is a symmetric $4$-tensor of $\ell_2$ norm at most $1$. Thinking of $\mathrm{Sym}(V)$ as a symmetric matrix over $2$-Tensors, 
we can write it as $\mathrm{Sym}(V) = \sum_{j=1}^{n^2} \lambda_j B_j \otimes B_j$, 
where the $B_j$'s are an orthonormal basis for the set of $2$-tensors and $\sum_{j=1}^{n^2} \lambda_j^2 = \|\mathrm{Sym}(V)\|_2^2 \leq 1$. 
Then, if $c_{i,j}:= \langle A_i, B_j\rangle$, and $c_j=\sum_{i=1}^k c_{i,j}^2$, we have that
\begin{align*}
\langle V,T \rangle & = \sqrt{6}\sum_{j=1}^{n^2} \lambda_j \sum_{i=1}^k c_{i,j}^2 = \sqrt{6}\sum_{j=1}^{n^2} \lambda_j c_j.
\end{align*}
Since the $A_i$ are orthonormal, $c_j\leq 1$ for all $j$. Furthermore, 
$$\sum_{j=1}^{n^2} c_j = \sum_{i=1}^k \sum_{j=1}^{n^2} c_{i,j}^2 = \sum_{i=1}^k \|A_i\|_F^2 = k.$$
Therefore,
$$
\langle V,T \rangle = \sqrt{6}\sum_{j=1}^{n^2} \lambda_j c_j \leq \sqrt{6} \sqrt{\sum_{j=1}^{n^2} \lambda_j^2} \sqrt{\sum_{j=1}^{n^2} c_j^2} \leq \sqrt{6\cdot 1 \cdot k} = O(\sqrt{k}) \;.
$$
\end{proof}
Let $\mu' =  \E_{X \sim \normal (0, \Sigma)} [Q(X)]$.
We now show that since $G_0$ is $ \eps$-good, then almost all of the difference in Claim \ref{claim:mu-bound} must be because of the points in $E$.
\begin{claim}
\label{claim:bad-points}
$\E_E [Q(Y)] - \mu' \geq \frac{C_1 - 6}{2} \cdot \frac{\xi k}{\phi}$
\end{claim}
\begin{proof}
It suffices to show that 
\begin{equation}
\label{eq:good-small}
\E_G [Q(X)] - \mu' \leq \frac{C_1 - 6}{2} \xi k \;,
\end{equation}
that is, the good points do not contribute much to the difference.
By $(\eps, \delta)$-goodness of $G_0$ and Claim \ref{claim:var-bound} we have $\E_0 [Q(Y)] - \mu' \leq O \left( \eps \sqrt{k} \right)$.
Moreover, we have
\begin{align*}
\phi (\E_L [Q(X)] - \E_{X \sim \normal (0, \Sigma)} [Q(X)]) &\leq \int_0^\infty \phi \Pr_L \left[ Q(X) - \mu' \geq t \right] dt \\
&\stackrel{(a)}{=} \int_0^{O(\sqrt{k} d^2 \log |S| / \delta)} \phi \Pr_L \left[ Q(X) - \mu' \geq t \right] dt \\
&\leq \phi \cdot O(\sqrt{k} \log^2 1 / \xi) +  \int_{O(\sqrt{k} \log^2 1 / \xi)}^{O(\sqrt{k} d^2 \log |S| / \delta)} \phi \Pr_L \left[ Q(X) - \mu' \geq t \right] dt \\
&\leq \phi \cdot O(\sqrt{k} \log^2 1 / \xi) + \int_{O(\sqrt{k} \log^2 1 / \xi)}^{O(\sqrt{k} d^2 \log |S| / \delta)} \Pr_{G_0} \left[ Q(X) - \mu' \geq t \right] dt \\
&\stackrel{(b)}{\leq} \phi \cdot O(\sqrt{k} \log^2 1 / \xi) + \int_{O(\sqrt{k} \log^2 1 / \xi)}^{O(\sqrt{k} d^2 \log |S| / \delta)} \Pr_{\normal (0, \Sigma)} [Q(X) - \mu' \geq t] dt + O(\xi) \\
&\stackrel{(c)}{\leq} O(\xi \sqrt{k}) + \exp (- \Omega( \log 1 / \xi)) \\
&\ll 5 \xi k \;,
\end{align*}
where (a) follows from the boundedness condition of $( \eps, \delta)$-goodness, (b) follows from the last condition of goodness, and (c) follows from hypercontractivity.
This shows (\ref{eq:good-small}), which completes the proof.
\end{proof}
\noindent We now show that this implies that there must be a $T$ satisfying the conditions in Algorithm \ref{alg:filter-cov-manyEig}.
\begin{claim}
\label{claim:covFindT}
If $\dim (V_m) \geq k$, then Algorithm \ref{alg:filter-cov-manyEig} returns a $T$ satisfying the conditions in the algorithm.
\end{claim}
\begin{proof}
Suppose not.
By the assumption that $\| I - \Sigma \|_F \leq \xi$, we have $|\mu'| = | \E_{X \sim \normal (0, \Sigma)} [Q(X)] | \leq 2 \xi \| Q \|_2 \leq \sum_{i = 1}^k \| p_i \|_2 + \| r_i \|_2 \leq 5 \xi k$.
Thus, we have 
\[
\E_E [Q(X)] \geq \left( \frac{C_1 - 6}{2 \psi} - 5 \right) \xi k \gg \left( \frac{C_1 - 6}{3 \psi} \right) \xi k \;.
\]
But since we assume there is no $T$ satisfying the conditions in Algorithm \ref{alg:filter-cov-manyEig}, we have
\begin{align*}
\psi \E_{X \in E} [p(X)] &\leq \psi \int_0^\infty \Pr_{E} [p(X) > t] dt \\
&\leq \psi \cdot 4 B \sqrt{k} \log^2 (1 / \xi) + \int_{4 B \sqrt{k} \log^2 1/ \xi}^{O(d^2 \sqrt{k} \log (|S| / \delta))} \psi \Pr_{E} [p(X) > t] dt \\
&\leq  4 B \psi \cdot \sqrt{k} \log^2 (1 / \xi) + \int_{4 B \sqrt{k} \log^2 1/ \xi}^{O(d^2 \sqrt{k} \log (|S| / \delta))} \Pr_{S} [p(X) > t] dt \\
&\leq 4 C_2 B \psi \cdot \sqrt{k} \log^2 (1 / \xi)  + \int_{4 A^2 C_2 B \sqrt{k} \log^2 1/ \xi}^{O(d^2 \sqrt{k} \log (|S| / \delta))} \left( \exp \left( - A \left( \frac{t}{4 B \sqrt{k}} \right)^{1/2} \right) + \frac{\eps^2}{d^2 \log |G| / \delta} \right) dt \\
&\leq 4 A^2 C_2 B \psi  \sqrt{k} \log^2 (1 / \xi) + 4 B \sqrt{k} \int_{A^2 \log^2 1/ \xi}^\infty  \exp (- \Omega (t^{1/2})) dt + \widetilde{O}(\eps^2) \\
&\leq O(\sqrt{k} \xi) + \widetilde{O}(\eps^2)  \ll 10 \xi k \; ,
\end{align*}
which is a contradiction, for our choice of $C_1, C_2, C_3$, and since we chose $k = O(\log^4 1 / \eps)$.
\end{proof}
The final thing we must verify is that the number of good points we remove is much smaller than the number of bad points we remove.
Formally, we show:
\begin{claim}
\label{claim:covDeltaBound}
If $\dim (V_m) \geq k$, then Algorithm \ref{alg:filter-cov-manyEig} returns a $S'$ satisfying $\Delta (S', G_0) < \Delta (S, G_0)$.
\end{claim}
\begin{proof}
Let $T$ be the threshold we pick.
If  $T>C_3 d^2 \sqrt{k} \log(|S|)$ then the invariant is satisfied since we remove no good points, by $( \eps, \delta)$-goodness.
It suffices to show that we remove $\log 1 / \eps$ times many more bad points than good points.
Otherwise, by definition we remove a total of 
\[
|S| \cdot \left( \exp \left( - A \left( \frac{T}{2B \sqrt{k}} \right)^{1/2} \right) + \frac{\eps^2}{d^2 \log^2 |S| / \delta} \right)
\]
points.
On the other hand, by $( \eps, \delta)$-goodness, hypercontractivity, and Claim \ref{claim:var-bound}, we know that the total number of points 
we throw away is at most
\begin{align*}
|G| \cdot \Pr_{G} [Q(X) > T] &\leq |G_0| \cdot \Pr_{G_0} [Q(X) > T]  \\
&\leq |G_0| \left( \exp \left( -A \left(\frac{T}{ \| Q \|_2} \right)^{-1/2} \right) + \frac{\eps^2}{2 \log (1 / \eps) (d \log(|G| / \delta))^2} \right) \\
&\leq |G_0| \left( \exp \left( -A \left( \frac{T}{B \sqrt{k}} \right)^{1/2} \right) + \frac{\eps^2}{2 \log (1 / \eps) (d \log(|G| / \delta))^2} \right) \; .
\end{align*}
Since we have maintained this invariant so far, in particular, we have thrown away more bad points than good points, and so $|G_0| \leq (1 + \eps) |S|$.
Moreover, since $T \geq 4 B \sqrt{k} \log^2 1 / \eps$, we have
\begin{align*}
\exp \left( -A \left( \frac{T}{B \sqrt{k}} \right)^{1/2} \right) &= \exp \left( -A \left( \frac{T}{4 B \sqrt{k}} \right)^{1/2} \right)^2 \\
&\leq \exp \left( -A \left( \frac{T}{4 B \sqrt{k}} \right)^{1/2} \right) \cdot \exp \left( -A \left( \frac{4 B \sqrt{k} \log^2 1 / \eps}{4 B \sqrt{k}} \right)^{1/2} \right) \\
&\leq  \exp \left( -A \left( \frac{T}{4 B \sqrt{k}} \right)^{1/2} \right)  \cdot \eps^{-1} \\
&\leq  \exp \left( -A \left( \frac{T}{4 B \sqrt{k}} \right)^{1/2} \right) \cdot \frac{1}{\log 1 / \eps} \; .
\end{align*}
Therefore, we have $\log (1 / \eps) \cdot |G| \cdot \Pr_{G} [Q(X) > T]  \leq |S| \cdot \left( \exp \left( - A \left( \frac{T}{4B \sqrt{k}} \right)^{1/2} \right) + \frac{\eps^2}{d^2 \log |G| / \delta} \right)$, and hence, the invariant is satisfied.
\end{proof}
\noindent Claims \ref{claim:covFindT} and \ref{claim:covDeltaBound} together prove the theorem.

\subsection{Proof of Lemma \ref{lem:GoodSamplesOpt}}
\label{sec:GoodSamplesOpt}
By Lemma 8.16 in \cite{DKKLMS}, the first two items hold together with probability $1 - O(\delta)$ after taking $O(\poly (d, 1 /\eta, \log 1 / \delta))$ samples.
Thus, it suffices to show that the last property holds with probability $1 - O(\delta)$ given $O(\poly (d, 1 /\eta, \log 1 / \delta))$ samples.
We may clearly WLOG take $\Sigma = I$.
Moreover, these properties are clearly invariant under scaling, and hence, it suffices to prove them for degree-four polynomials with $\Var_{\normal (0, I)} [p(X)] = 1$.
For any fixed degree-$4$ polynomial $p$, by hypercontractivity, since $P(X) = \frac{1}{n} \sum_{i = 1}^n p(X_i)$ has variance $\Var [P(X)] = \frac{1}{n}$, we have that
\[
\Pr \left[ \left| \frac{1}{n} \sum_{i = 1}^n p(X_i) - \E_{\normal (0, I)} [p(X)] \right| \geq t \right] \leq \exp \left( - A t^{1/2} n^{1/4} \right) \; .
\]
In particular, if we take 
\[
n = \Omega \left( \frac{\log^4 1 / \delta}{\eta^2} \right) \; ,
\]
then 
\[
\Pr \left[ \left| \frac{1}{n} \sum_{i = 1}^n p(X_i) - \E_{\normal (0, I)} [p(X)] \right| \geq \eta \right] \leq O(\delta) \; .
\]
Since there is a $1/3$-net over all degree-$4$ polynomials with unit variance of size $(1/3)^{O(d^4)}$, by union bounding, we obtain that if we take $\Omega \left( \frac{d^4 \log^{1 / \delta}}{\eta^2} \right)$ samples, then the second to last property holds with probability $1 - \delta$.

Finally, the same net technique used to prove Lemma 8.16, along with hypercontractivity, may be used to show the last property holds when given $\poly (d, \log 1 /\eta, \log 1 / \delta)$ samples.

\subsection{Proof of Lemma \ref{lem:covBoundedOffLarge}}
\begin{proof}
Let $\mu' = \E_{\normal (0, \Sigma)} [p(X)]$, $s = \E_{\normal (0, \Sigma)} [p^2 (X)]$, and let $ \E_S [p (X)] - \E_{\normal (0,\Sigma)} [p(X)] = K$.
By $( \eps, \delta)$-goodness and Lemma \ref{lem:poly-props}, we have
\begin{equation}
\label{eq:covG0Bound2}
\left| \E_{G_0} [p(X)] - \mu' \right| \leq \eps  \E_{\normal (0, \Sigma)} [p^2(X)]^{1/2} \leq 2 \eps \; .
\end{equation}
Moreover, for some appropriate choice of $\beta_1$, we have
\begin{align*}
\left| \phi \E_{L} [p(X)] - \phi \mu' \right| &\leq \int_0^\infty \phi \Pr_L \left[ \left| p(X) - \mu \right| \geq t \right] dt \\
&\leq \beta_1 \phi \log 1 / \eps + \int_{\beta_1 \log 1 / \eps}^{O(d \log |S| / \delta)} \phi \Pr_L \left[ \left| p(X) - \mu \right| \geq t \right] dt \\
&\stackrel{(a)}{\leq} \beta_1 \phi \log 1 / \eps + \frac{|G_0|}{|S|} \int_{\beta_1 \log 1 / \eps}^{O(d \log |S| / \delta)} \Pr_{G_0} \left[ \left| p(X) - \mu \right| \geq t \right] dt \\
&\stackrel{(b)}{\leq} \beta_1 \phi \log 1 / \eps + \frac{|G_0|}{|S|} \int_{\beta_1 \log 1 / \eps}^{O(d \log |S| / \delta)} \Pr_{\normal (0, \Sigma)} \left[ \left| p(X) - \mu \right| \geq t \right] dt + O(\eps) \\
&\stackrel{(c)}{\leq} \beta_1 \phi \log 1 / \eps + \eps^{O(\beta_1)} + O( \eps) \\
&= O(\eps)  \numberthis \label{eq:covLBound2} \; ,
\end{align*}
where (a), (b), follows from the definition of goodness, and (c) follows from standard Guassian concentration bounds.
In particular, (\ref{eq:covG0Bound2}) and (\ref{eq:covLBound2}) together imply that
\[
\frac{2K}{\psi} \geq \left| \E_{E} [p(X) - \mu'] \right| \geq \frac{K}{2 \psi} \; ,
\]
so this implies
\[
\E_{E} [(p(X) - \mu')^2 ] \geq \frac{K^2}{4 \psi^2} \; .
\]
Expanding this yields that
\begin{align*}
\E_{E} [(p(X) - \mu')^2 ] &= \E_E [p^2 (X)] - 2 \mu' \E_E [p(X)] + (\mu')^2 \; ,
\end{align*}
so since $|\mu'| \leq \| \Sigma - I \|_F \leq \xi$, we have
\begin{equation}
\label{eq:covEBound3}
\E_E [p^2 (X)] \geq \frac{K^2}{4 \psi^2} - \frac{2K}{\psi} + O(\xi^2) \; .
\end{equation}

On the other hand, by Lemma \ref{lem:poly-props}, we have
\begin{align*}
|s - 1| &\leq \| \Sigma^{\otimes 2} - I^{\otimes 2} \|_2 + \frac{1}{2} \| \Sigma - I \|_F^2  \\
&\leq \xi + O(\xi^2) \; . \numberthis \label{eq:covSBound}
\end{align*}
Also, by $(\eps, \delta)$-goodness and Lemma \ref{lem:poly-props}, we have
\begin{equation}
\label{eq:covG0Bound3}
\left| \E_{G_0} [p^2 (X)] - s \right| \leq \eps (1 + O(\xi)) \; ,
\end{equation}
and, for $A$ as in Corollary \ref{cor:hypercontractivity}, and $\beta$ appropriately chosen,
\begin{align*}
\left| \phi (\E_{L} [p^2 (X)] - s) \right| &\leq \int_0^\infty \phi \Pr_L [| p^2 (X) - s | \geq t] dt \\
&\leq \beta A \phi \log^2 1 / \eps + \int_{2A \log^2 1 / \eps}^{O((d \log |S| / \delta)^2)} \phi \Pr_L [| p^2 (X) - s | \geq t] dt \\
&\leq \beta A \phi \log^2 1 / \eps +  \frac{|G_0|}{|S|} \int_{\beta A \log^2 1 / \eps}^{O((d \log |S| / \delta)^2)} \Pr_{G_0} [| p^2 (X) - s | \geq t] dt \\
&\stackrel{(a)}{\leq} \beta A \phi \log^2 1 / \eps +  \frac{|G_0|}{|S|} \int_{\beta A \log^2 1 / \eps}^{O((d \log |S| / \delta)^2)} \Pr_{\normal(0, \Sigma)} [| p^2 (X) - s | \geq t] dt + O( \eps^2) \\
&\stackrel{(b)}{\leq} \beta A \phi \log^2 1 / \eps + \int_{\beta A \log^2 {1/\eps}}^\infty e^{-A t^{1/2} / 2} dt + O(\eps^2) \\
&\ll 2 \beta A \xi + O(\eps) + O( \eps^2) \; , \numberthis \label{eq:covLBound3}
\end{align*}
where (a) follows from goodness, and (b) follows from Corollary \ref{cor:hypercontractivity} and since $\| \Sigma - I \|_F = o(1)$.
Thus, (\ref{eq:covEBound3}), (\ref{eq:covSBound}), (\ref{eq:covG0Bound3}), and (\ref{eq:covLBound3}) together imply that
\[
\E_S [p^2 (X)] - 1 \geq 1 - (\xi + O(\xi^2)) - O( \eps) - 2 \beta A \xi O(\eps) + \frac{K^2}{4 \psi} - 2K + O(\xi^2) \; ,
\]
and so since $K \geq \Omega \left( \sqrt{\eps \xi} \right)$, this gives the desired bound.
\end{proof}

\end{document}